\def\llncs{0}
\def\mnotes{0}
\documentclass[11pt]{article}

\usepackage{wrapfig}

\usepackage{amsfonts,amsmath,epsfig, graphicx, fullpage,mathtools}
\usepackage{color,bm}

\def\colorson{0}

\usepackage{times}
\usepackage{color,xcolor}
\usepackage{amsfonts,amsmath}
\usepackage{url}
\usepackage{fullpage}

 \ifnum\llncs=0

\definecolor{orange}{RGB}{127,127,255}

\colorlet{green}{green!80!black}

\usepackage{amsthm}
\usepackage{latexsym,amssymb}
\usepackage{graphicx}

    \usepackage[letterpaper=true,pdftex,
    bookmarks,
    plainpages=false, 
 pdfpagelabels=true, 
 colorlinks=true,breaklinks=true,linkcolor=red, citecolor=green,filecolor=black,menucolor=black,pagecolor=black,
 citebordercolor=green,
    pdfauthor=Jalaj Upadhyay	
      ]{hyperref}{}

\fi 
\ifnum\mnotes=0
\newcommand{\mnote}[1]{}
\else
\newcounter{mynotes}
\newcommand{\mnote}[1]{\addtocounter{mynotes}{1}{{\bf !}}%
\marginpar{\scriptsize  {\arabic{mynotes}.\ {\sf \textcolor{red}{#1}}}}}
\fi

\ifnum\colorson=1
\newenvironment{todo}{\noindent
\sf \footnotesize \textcolor{blue}{To go here}:
\begin{CompactItemize}\color{blue}}
{\color{black}\end{CompactItemize}\rm \normalsize}
\else

\fi

\ifnum\llncs=0

\newenvironment{CompactEnumerate}{
  \vspace{-5pt}
  \begin{list}{\arabic{enumi}.}
  {%
      \usecounter{enumi} %
      \setlength{\leftmargin}{10pt}%
      \setlength{\itemsep}{0em}
      }}
  {\end{list}}

\newenvironment{CompactItemize}{
  \vspace{-10pt}
 \begin{list}{$\bullet$}{%
      \setlength{\leftmargin}{8pt}%
      \setlength{\itemsep}{0em}
      }}
{\end{list}}

\makeatletter
\def\thm@space@setup{\thm@preskip=3pt
\thm@postskip=2pt}
\makeatother
\newtheoremstyle{newstyle}      
{} 
{} 
{\mdseries} 
{} 
{\bfseries} 
{.} 
{ } 
{} 

\newtheorem{theorem}{Theorem}
\newtheorem{lemma}[theorem]{Lemma}

\newtheorem{claim}[theorem]{Claim}

\newtheorem{fact}[theorem]{Fact}
\newtheorem{corollary}[theorem]{Corollary}

\newtheorem{prob}{Problem}

\newtheorem*{theorem*}{Theorem}

\newtheorem{definition}[theorem]{Definition} 

\theoremstyle{definition}

\theoremstyle{remark}

\newtheorem{myremark}{Remark} 
\newenvironment{remark}{\begin{myremark}}{\end{myremark}}

\newtheorem{myexample}{Example}

\newtheorem{mygraph}{Graph}

\else
\spnewtheorem{protocol}{Protocol}{\bfseries}{\rmfamily}
\spnewtheorem{algm}{Algorithm}{\bfseries}{\rmfamily}
\spnewtheorem{fact}{Fact}{\bfseries}{\rmfamily}
\spnewtheorem{myclaim}{Claim}{\bfseries}{\itshape}
\fi

\newcommand{\thmref}[1]{\textcolor{red}{Theorem}~\ref{thm:#1}}

\newcommand{\lemref}[1]{\textcolor{red}{Lemma}~\ref{lem:#1}}

\newcommand{\factref}[1]{\textcolor{red}{Fact}~\ref{fact:#1}}

\newcommand{\claimref}[1]{\textcolor{red}{Claim}~\ref{claim:#1}}

\newcommand{\ldefref}[1]{\textcolor{red}{Definition}~\ref{def:#1}}
\newcommand{\secref}[1]{\textcolor{red}{Section}~\ref{sec:#1}}

\newcommand{\appref}[1]{\textcolor{red}{Appendix}~\ref{app:#1}}

\newcommand{\figref}[1]{\textcolor{red}{Figure}~\ref{fig:#1}}

\newcommand{\eqnref}[1]{\textcolor{red}{equation~(\ref{eq:#1})}}

\newcommand{\comment}[1]{}
\newcommand{\ignore}[1]{}

\newcommand{\h}[1]{\widehat{#1}}

\newcommand{\tr}{\mbox{Tr}}

\newcommand{\brak}[1]{{\langle {#1} \rangle}}
\newcommand{\set}[1]{\left\{ {#1} \right\}}
\newcommand{\paren}[1]{\left( {#1} \right)}
\newcommand{\sparen}[1]{\left[ {#1} \right]}

\newcommand{\eps}{\epsilon}

\DeclareMathOperator{\poly}{poly}

\renewcommand{\gets}{\leftarrow}

\newcommand{\E}{\mathbb{E}}

\newcommand{\I}{\mathbb{I}}

\newcommand{\bL}{\mathbb{L}}

\newcommand{\N}{\mathbb{N}}

\newcommand{\p}{\mathsf{Pr}}

\newcommand{\R}{\mathbb{R}}
\newcommand{\bS}{\mathbb{S}}

\newcommand{\bZ}{\textbf{Z}}
\newcommand{\bX}{\textbf{X}}
\newcommand{\by}{\textbf{y}}
\newcommand{\bx}{\textbf{x}}

\newcommand{\cA}{\mathcal{A}}

\newcommand{\cD}{\mathcal{D}}
\newcommand{\cE}{\mathcal{E}}

\newcommand{\cG}{\mathcal{G}}

\newcommand{\cL}{\mathcal{L}}

\newcommand{\cN}{\mathcal{N}}

\newcommand{\cP}{\mathcal{P}}

\newcommand{\cR}{\mathcal{R}}

\newcommand{\cV}{\mathcal{V}}

\newcommand{\beq}{\begin{equation}}
\newcommand{\eeq}{\end{equation}}
\newcommand{\bml}{{\begin{multline}}}
\newcommand{\eml}{{\end{multline}}}

\usepackage{multirow}
\usepackage{listings}

\usepackage[top=1in, bottom = 1in, left=1in, right=1in]
{geometry}

\makeatletter
\renewcommand{\paragraph}{%
 \@startsection{paragraph}{4}%
  {\z@}{1.5ex \@plus 1ex \@minus .2ex}{-.2em}%
  {\small\normalsize\bfseries}%
}
\makeatother

\usepackage{comment,enumitem}
\excludecomment{full}
\includecomment{extended}

\excludecomment{l1norm}
\includecomment{l2norm}
\excludecomment{fullversion}

\newcommand{\vertical}[1]{{\left\vert\kern-0.25ex\left\vert\kern-0.25ex\left\vert #1  \right\vert\kern-0.25ex\right\vert\kern-0.25ex\right\vert}}

\newcommand{\bLambda}{\mathbf{\Lambda}}
\newcommand{\bSigma}{\mathbf{\Sigma}}
\newcommand{\bPsi}{\mathbf{\Psi}}
\newcommand{\bPhi}{\mathbf{\Phi}}
\newcommand{\bPi}{\mathbf{\Pi}}
\newcommand{\bA}{\mathbf{A}}
\newcommand{\bB}{\mathbf{B}}
\newcommand{\bC}{\mathbf{C}}
\newcommand{\bD}{\mathbf{D}}
\newcommand{\bN}{\mathbf{N}}
\newcommand{\bU}{\mathbf{U}}
\newcommand{\bV}{\mathbf{V}}
\newcommand{\bW}{\mathbf{W}}
\newcommand{\bR}{\mathbf{R}}
\newcommand{\bT}{\mathbf{T}}
\newcommand{\bY}{\mathbf{Y}}
\renewcommand{\bS}{\mathbf{S}}
\newcommand{\bOmega}{\mathbf{\Omega}}

\newcommand{\aind}{\mathsf{AIND}}
\newcommand{\ind}{\mathsf{ind}}
\newcommand{\bM}{\mathbf{M}}
\newcommand{\bv}{\mathbf{v}}
\newcommand{\cov}{\mathsf{COV}}
\newcommand{\lra}{\mathsf{LRA}}
\newcommand{\PDF}{\mathsf{PDF}}
\newcommand{\range}{\mathsf{range}}

\newcommand{\rad}{\mathsf{Rad}}
\newcommand{\lrf}{\mathsf{LRF}}
\newcommand{\priv}{\mathsf{Priv}}

\newcommand{\argmin}{\operatornamewithlimits{argmin}}

\usepackage{tabularx} 

\usepackage[bottom]{footmisc}

\renewcommand{\h}[1]{\widehat #1}
\renewcommand{\t}[1]{\widetilde #1}

\begin{document}
\title{The Price of Differential Privacy for Low-Rank Factorization}
\author{
Jalaj Upadhyay \thanks{Work was done when the author was at Pennsylvania State University. Research supported by NSF award IIS-1447700. } 
\\
Department of Computer Science\\
 Whiting School of Engineering\\
Johns Hopkins University \\
\small{$\mathsf{jalaj@jhu.edu}$}
}
\date{}
\maketitle

\pagenumbering{gobble}
\begin{abstract}
Computing a {\em low-rank factorization} of an $m \times n$ matrix $\mathbf{A}$ is a fundamental component in many applications, such as clustering, data mining, recommendation system, etc. It requires outputting three matrices: an orthonormal $m \times k$ matrix $\mathbf{U}$, an orthonormal $n \times k$ matrix $\mathbf{V}$, and a $k \times k$ positive semidefinite diagonal matrix $\mathbf{\Sigma}$, where  $\mathbf{U} \mathbf{\Sigma} \mathbf{V}^{\mathsf T}$ approximates the matrix $\mathbf{A}$ in the Frobenius norm in the sense that 
$$  { \| \mathbf{A} - \mathbf{U} \mathbf{\Sigma} \mathbf{\bV}^{\mathsf T} \|_F \leq (1+\alpha)  \|\mathbf{A} - [\mathbf{A}]_k\|_F +\gamma }$$ with constant probability. 
Here $[\mathbf{A}]_k$ is the best rank-$k$ approximation of $\mathbf{A}$.

In this paper, we study what  price one has to pay to release {\em differentially private low-rank factorization} of a matrix. We consider various settings that are close to the real world applications of low-rank factorization: (i) the manner in which matrices are updated (row by row or in an arbitrary manner), (ii) whether matrices are distributed or not, and (iii) how the output is produced (once at the end of all updates, also known as {\em one-shot algorithms}  or continually). Even though these settings are well studied without privacy, surprisingly, there are no private algorithm for these settings (except when a matrix is updated row by row). We present the first set of differentially private algorithms for all these settings.  

Our algorithms when private matrix is updated in an arbitrary manner promise differential privacy with respect to two stronger privacy guarantees than previously studied, use space and time {\em  comparable} to the non-private algorithm, and achieve {\em optimal accuracy}. To complement our positive results, we also prove that the space required by our algorithms is optimal up to logarithmic factors. When data matrices are distributed over multiple servers, we give a non-interactive differentially private algorithm  with communication cost independent of dimension. In concise, we give algorithms that incur optimal cost. 
We also perform experiments  to verify that all our algorithms  perform well in practice and outperform the best known algorithms until now for large range of parameters. We give  experimental results for total approximation error and additive error for varying dimensions, $\alpha$ and $k$. 

There are few key take aways from this paper: (i) maintaining a differentially private sketches of row space and column space of a matrix already give a sub-optimal accuracy, but this can be significantly improved by careful inexpensive post-processing, 
 and (ii) even though we can only use linear sketches when matrices are updated arbitrarily, the structural properties of linear sketches can be carefully exploited to get tight bound on the error. While most of the focus has been on careful analysis, the power of post-processing  has been largely ignored in the literature of differentially private low rank factorization. 
\end{abstract}

\paragraph{Keywords} 
Low-rank approximation,
differential privacy,
local differential privacy

\newpage

\pagenumbering{arabic}
\setlength{\textfloatsep}{-2pt plus 0pt minus 0pt}

\section{Introduction} \label{sec:introduction}
Low-rank factorization ($\lrf$) of matrices is a fundamental component used in many applications, 
such as
chemometrics~\cite{Stegeman08,WAHFK97}, 
clustering~\cite{CEMMP15,DFKVV04,McSherry01}, 
data mining~\cite{AFKMS01}, recommendation systems~\cite{DKR02}, 
information retrieval~\cite{PRTV00,SC04}, 
learning  distributions~\cite{AM05,KSV05}, 
system-control~\cite{Markovsky}, and
web search~\cite{AFKM01,Kleinberg99}.
 In these applications, given an $m \times n$ matrix $\bA$, a common approach is to first compute three matrices: a diagonal positive semidefinite matrix $ \widetilde{\bSigma}_k \in \R^{k \times k}$ and two  matrices with orthonormal columns, $ \widetilde{\bU}_k \in \R^{m \times k}$ and $ \widetilde{\bV}_k \in \R^{n \times k}$. The requirement then  is that the product $\bB:=\widetilde \bU \widetilde \bSigma \widetilde \bV^{\mathsf T}$ is as close to $\bA$ as possible.  
More formally, 
\begin{prob} 
\label{prob:private_factor} $(\alpha,\beta,\gamma,k)\mbox{-}{\lrf}$.
Given parameters $0 <\alpha,\beta <1$ and $\gamma$, an $m \times n$ matrix $\bA$  and the target rank $k$, compute a rank-$k$ matrix factorization $\widetilde{\bU}_k, \widetilde{\bSigma}_k$, and $\widetilde{\bV}_k$ such that with probability at least $1-\beta$
$$  { \| \bA - \widetilde{\bU}_k \widetilde{\bSigma}_k \widetilde{\bV}_k^{\mathsf T} \|_F \leq (1+\alpha)  \Delta_k +\gamma },$$ 
{where}~$\| \cdot \|_F$ denotes the Frobenius norm, $\Delta_k:= \|\bA - [\bA]_k\|_F$ with $[\bA]_k$ being the best rank-$k$ approximation of $\bA$. 
We refer to the parameter $\gamma$ as the {\em additive error} and the parameter $\alpha$ as the {\em multiplicative error}. 
\end{prob}

Unfortunately, most of the applications listed above either does not use any privacy mechanism or use ad hoc privacy mechanism. This has led to serious privacy leaks as exemplified by denonymization of Netflix datasets~\cite{NS08}. To assuage such attacks, McSherry and Mironov~\cite{MM09} suggested that one should use a  robust privacy guarantee, such as $(\eps,\delta)$-{\em differential privacy}, defined by Dwork {\it et al.}~\cite{DMNS06}:
\begin{definition} \label{def:approxdpstreaming} 
	A randomized algorithm $\mathfrak{M}$ gives {\em $(\varepsilon, \delta)$-differential privacy}, if for all neighboring datasets $\bA$ and ${\bA}'$, and all subsets $S$ in the range of $\mathfrak{M}$, 	$ \p[\mathfrak{M}(\bA) \in S] \leq \exp (\varepsilon) \p[\mathfrak{M}({\bA}') \in S] + \delta, $ where the probability is over the coin tosses of $\mathfrak{M}$. 
\end{definition}

All currently known differentially private algorithms~\cite{BDMN05,CSS12,DTTZ14,HP14,HR12,JXZ15,KT13,Upadhyay14} output a low-rank matrix  if a single entity stores a static matrix. One can then factorized the output in $O(mn^2)$ time and $O(mn)$ space. 
However, practical matrices used in above applications are often large, distributed over many servers, and are dynamically updated~\cite{MM11,SC}. 
Therefore, for any practical deployment~\cite{apple,rappor}, one would like to simultaneously maintain strong privacy guarantee and minimize space requirements, communication between servers and computational time, in order to 
reduce bandwidth, latency, synchronization issues and resource overhead. 

We initialize the study of space, time, and communication cost efficient differentially private algorithm for solving \textcolor{red}{Problem}~\ref{prob:private_factor} under various settings: (i) the manner in which matrices are updated (row by row or in arbitrary manner), (ii) whether matrices are distributed (local model of privacy) or not (central model of privacy), and (iii) how the output is produced (once at the end of all updates, also known as {\em one-shot algorithms}  or continually). Without privacy, there are many algorithms for solving \textcolor{red}{Problem}~\ref{prob:private_factor} under all these settings~\cite{BCL05, BWZ16, CMM17, CW13, CW17, DV06, GSS17, MZ11,MM13,NDT09,Sarlos06}. 

Surprisingly, there are no private algorithm for these settings (except in central model when matrices are received row by row~\cite{DTTZ14}). For example,  known algorithms either use multiple pass over the data matrix~\cite{HP14,HR12, HR13, KT13} or cannot handle arbitrary updates~\cite{DTTZ14}. Similarly,  known algorithms that continually release output are for {\em monotonic functions}~\cite{DNPR10}, thereby excluding \textcolor{red}{Problem}~\ref{prob:private_factor}. Private algorithms like~\cite{DTTZ14,HR12} that can be extended to distributed setting use multiple rounds of interaction, large communication cost, and/or result in trivial error bounds. Moreover,  known private algorithms are inefficient compared to  non-private algorithms: $O(mnk)$ time and $O(mn)$ space compared to time linear in the sparsity of the matrix and $O((m+n)k/\alpha)$ space~\cite{BWZ16,CW13}. In fact, for rank-$k$ matrices, Musco and Woodruff~\cite{MW17} state that \textcolor{red}{Problem}~\ref{prob:private_factor} is equivalent to the well studied {\em matrix completion} problem for which  one can have $\widetilde O(n\cdot \poly(k))$ time non private algorithm~\cite{JNS13}. Under same assumptions, private algorithm takes $O(mnk)$ time~\cite{HR12}. This motivates the central thesis of this paper: can we have non-trivial private algorithms under all these settings, and if yes, do we  have to pay for privacy? 

\paragraph{Our Results.} We give a unified approach and  first set of algorithms  for solving \textcolor{red}{Problem}~\ref{prob:private_factor} in all the settings mentioned above while showing that {\em one does not have to pay the price of privacy} (more than what is required in terms of additive error and space, see~\secref{optimality}). Our algorithms  are also practical as shown by our empirical evaluations (see~\appref{empirical}).   
 On a high level, we show the following:
\begin{CompactEnumerate}
	\item When a private matrix is streamed, we propose  differentially private algorithms with respect to two stronger privacy guarantees than previously studied. 
	Our algorithms are space and time efficient and also achieve {\em optimal accuracy}. 
	We complement our positive results with a matching lower bound on the space required  (\thmref{informallower}). \secref{meta} covers this in more details.
	\item When a private matrix is distributed over multiple servers, we give a non-interactive differentially private algorithm  with communication cost independent of dimension.  \secref{informallocal} covers this in more details. 
\end{CompactEnumerate}
\begin{table}[t]
\small{
 \begin{center}
\begin{tabular}{|c|c|c|c|c|c|}
\hline
Our Results  & Updates & Output Produced & Privacy  & Additive Error & Local? \\ \hline

\thmref{low_space_private} & Turnstile & At the end &  $\bA - {\bA}' =\mathbf{u} \mathbf{v}^{\mathsf T}$ & $\widetilde{O} \paren{   \paren{ \sqrt{{mk} \alpha^{-1}} + \sqrt{kn}}{\varepsilon}^{-1}}$  & $\times$
 \\ \cline{1-6}

\thmref{streaming_naive} & Turnstile  & At the end& $\| \bA - {\bA}'\|_F = 1$ & $\widetilde{O} \paren{   \paren{ \sqrt{{mk} \alpha^{-2}} + \sqrt{kn}}{\varepsilon}^{-1}}$   & $\times$
\\  \cline{1-6}

\thmref{continualfrob} & Turnstile & Continually &  $\bA - {\bA}' =\mathbf{u} \mathbf{v}^{\mathsf T}$ & $\widetilde{O} \paren{   \paren{ \sqrt{{mk} \alpha^{-1}} + \sqrt{kn}}{\varepsilon}^{-1} \log T}$   & $\times$
\\ \cline{1-6}

\thmref{continuallowspace} & Turnstile & Continually & $\| \bA - {\bA}'\|_F = 1$ & $\widetilde{O} \paren{   \paren{ \sqrt{mk} \alpha^{-1} + \sqrt{kn}}{\varepsilon}^{-1} \log T}$   & $\times$
\\ \hline

\thmref{streaming_covariance} & Row by row & At the end & $\| \bA - {\bA}'\|_F = 1$ & $\widetilde{O} \paren{   \paren{\alpha \varepsilon}^{-1} \sqrt{nk} }$   & $\times$
\\ \cline{1-6}

\thmref{local} & Row by row & At the end &  $\| \bA - {\bA}'\|_F = 1$ & $\widetilde{O}\paren{ k \alpha^{-2} \eps^{-1}  \sqrt{m} }$   & $\surd$

 \\ \hline

 \end{tabular}
 \caption{\small{Our Results for $(\varepsilon,\Theta(n^{-c}))$-Differentially Private Algorithms When $k>1/\alpha, c \geq 2$, stream length $T$)}.} \label{tab:results}
\end{center}
}
\end{table}
 
 Our results relies crucially on  {\em careful analysis} and our algorithms heavily exploit advance yet inexpensive {\em post-processing}. The power of careful analysis was shown by Dwork {\it et al.}~\cite{DTTZ14} who improved an earlier result of Blum {\it et al.}~\cite{BDMN05}. However, only two known private algorithms for \textcolor{red}{Problem}~\ref{prob:private_factor}  use any form of post-processing: 
Hardt and Roth~\cite{HR12}  uses simple pruning of entries of a matrix formed in the intermediate step, while Dwork {\it et al.}~\cite{DTTZ14} uses  best rank-$k$ approximation of the privatized matrix.  These post-processing either  make the algorithm suited only for static matrices or are expensive.

In what follows, we give more details of our result when matrices are streamed and compare it with previous works. All our results are summarized  in \textcolor{red}{Table}~\ref{tab:results}. 
Unless specified, for the ease of presentation, we assume that  $k \geq 1/\alpha, \delta = \Theta(n^{-c})$ for $c \geq 2$, and $\widetilde O(\cdot)$ hides a $\poly \log n$ factor for the rest of this section.  
  
\subsection{Differentially Private Low-Rank Factorization When Matrices are Streamed} \label{sec:meta}
All currently known differentially private algorithms to solve \textcolor{red}{Problem}~\ref{prob:private_factor} assume that a single central server holds the entire private matrix which is not updated over a period of time (see~\figref{streaming}). However, practical matrices are constantly updated. In this section, we  consider when a private matrix is updated arbitrarily, known as the {\em turnstile update model} to capture many practical scenarios (see 
the survey~\cite{Muthu05} for further motivations).   
Formally, in a turnstile update model,
 a matrix $\bA \in \R^{m \times n}$ is initialized  to an all zero-matrix  and is
  updated by a sequence of triples $\set{i,j,\Delta}$, where $1 \leq i \leq m, 1\leq j \leq n,$ and $ \Delta \in \R$.
 \begin{wrapfigure}{r}{0.65\textwidth}
  \includegraphics[height=0.65in,width=1.8in]{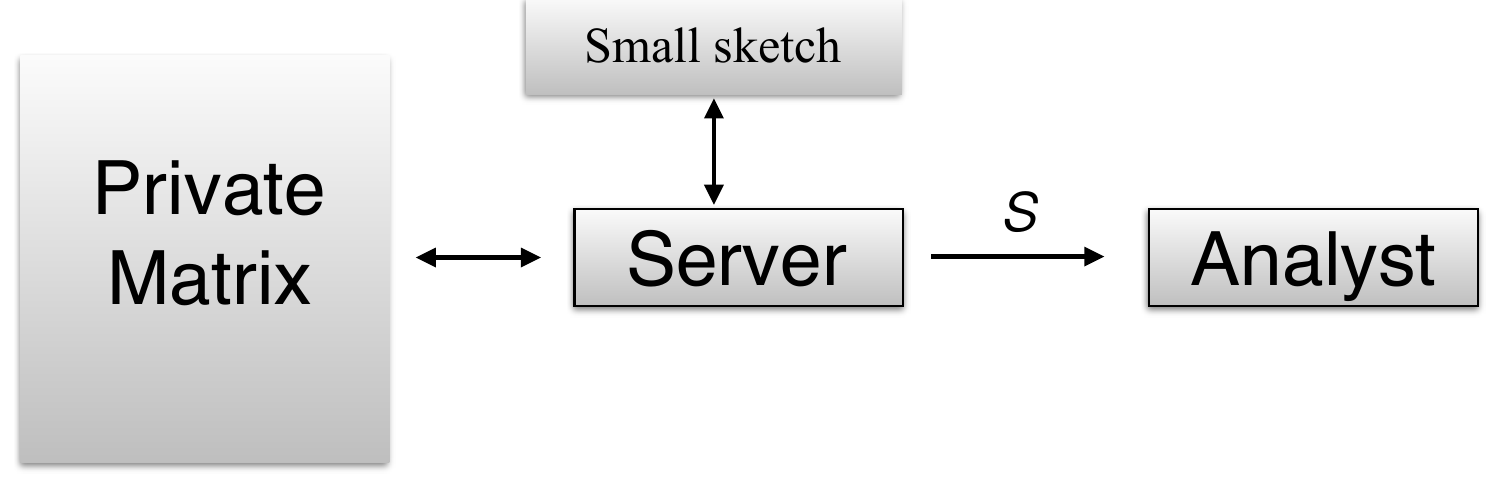}
  \hspace{6pt}  \rule{0.4pt}{.65in}
  \hspace{6pt}  \includegraphics[height=0.65in,width=1.8in]{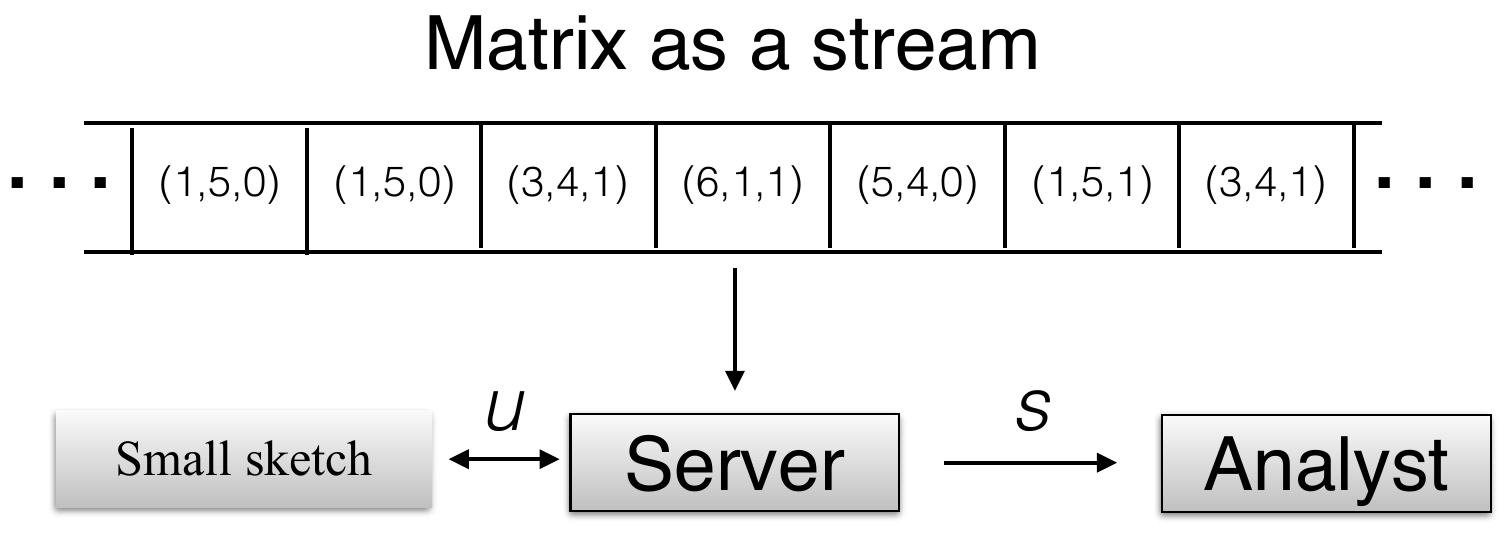}  
  \caption{The central model (left) and the turnstile model (right).} 
  \label{fig:streaming}
\end{wrapfigure} 
   Each update results in a change in the $(i,j)$-th entry of $\bA$ as follows:
  $\bA_{i,j} \gets \bA_{i,j}+\Delta$.  
 An algorithm is 
{\em differentially private under  turnstile update model}
 if, for all possible matrices updated in the turnstile 
 update model and runs of the
 algorithm, the output of the server is $(\varepsilon,\delta)$-differentially private. 
  In the non-private setting, Boutsidis {\it et al.}~\cite{BWZ16} gave a space-optimal algorithm for low-rank factorization in the {turnstile update model}. 
  Unfortunately, a straightforward application of known privacy techniques to make their algorithm   differentially private  incurs a large additive error (see~\appref{BWZ16} for details).
This raises the following natural question:
\begin{quote}
 {\bf Question 1.} Can we solve \textcolor{red}{Problem}~\ref{prob:private_factor} with good accuracy while preserving differential privacy and receiving the matrix in the turnstile update model? 
\end{quote}
We answer this question positively. To better comprehend the privacy aspect of our results presented in this section, we first need to define our notion of neighboring datasets (streaming matrices in our context). 

We consider two  stronger privacy guarantees than  previously studied: $\priv_1$ and $\priv_2$ (see~\secref{compare} for comparison and~\secref{discussion} for discussion). 
 In $\priv_1$, we call two  matrices $\bA$ and $\bA'$  neighboring if $\bA - \bA' = \mathbf{u}\mathbf{v}^{\mathsf T}$ for some unit vectors $\mathbf{u}$ and $\mathbf{v}$. In $\priv_2$, we consider two  matrices $\bA$ and $\bA'$ neighboring if $\| \bA - \bA' \|_F \leq 1$.
\begin{l1norm}
 In $\priv_3$, we consider two  matrices $\bA$ and $\bA'$ neighboring if $\| \bA - \bA' \|_1 \leq 1$.
\end{l1norm}
 We say two data streams are neighboring if they are formed by neighboring matrices. For these two definitions, we show the following:
\begin{theorem} (\thmref{low_space_private} and \thmref{streaming_naive}, informal).  \label{thm:informalspace}Let $\bA$ be an $m \times n$ matrix (with $\mathsf{nn}(\bA)$ non-zero entries and $m \leq n$) streamed in a  turnstile update model. 
Then 
\begin{CompactEnumerate}
	\item There is an $(\varepsilon,\delta)$-differentially private algorithm under $\priv_1$ that  uses $ \widetilde{O}((m+n)k  \alpha^{-1})$ space and outputs a rank-$k$ factorization $\t{\bU}_k, \t{\bSigma}_k, {\t{\bV}}_k$ in time $\widetilde{O}(\mathsf{nn}(\bA) + m^2k  \alpha^{-1}+ nk^2 \alpha^{-2} )$ such that, 
	 \begin{align*}
   \| \bA - \t{\bU}_k \t{\bSigma}_k \t{\bV}_k^{\mathsf T} \|_F \leq (1+\alpha)  \Delta_k +
  \widetilde{O} ({   ({ \sqrt{{mk}\alpha^{-1}} + \sqrt{kn}})/{\varepsilon}}). 
   \end{align*} 
   \item There  is an $(\varepsilon,\delta)$-differentially private algorithm under $\priv_2$ that  uses $\widetilde O( (m\alpha^{-1}+n)k \alpha^{-1})$ space, and outputs a rank-$k$ factorization $\t{\bU}_k, \t{\bSigma}_k, \t{\bV}_k$ in $\widetilde O(\mathsf{nn}(\bA)+(m\alpha^{-2} + n)k^2 \alpha^{-2} )$ time such that
	$$ \| \bA - \t{\bU}_k \t{\bSigma}_k \t{\bV}_k^{\mathsf T} \|_F \leq  (1+\alpha)  \Delta_k +
	  \widetilde O ({ { ( \sqrt{ m k \alpha^{-2}} + \sqrt{ kn  } )}/{\varepsilon}   }). 
	  $$
\begin{l1norm}
	   \item There  is an $(\varepsilon,\delta)$-differentially private algorithm under $\priv_3$ that  uses $\widetilde O( (m\alpha^{-1}+n)k \alpha^{-1})$ space, and outputs a rank-$k$ factorization $\bU_k, \t{\bSigma}_k, \bV_k$ in $\widetilde O(\mathsf{nn}(\bA)+(m\alpha^{-2} + n)k^2 \alpha^{-2} )$ time such that
	$$ \| \bA - \bU_k \t{\bSigma}_k \bV_k^{\mathsf T} \|_1 \leq  (1+\alpha)  \min_{\text{rank-k matrices }\bX} \| \bA - \bX \|_1 + 
	  \widetilde O ({ { ( \sqrt{ m k \alpha^{-2}} + \sqrt{ kn  } )}/{\varepsilon}   }). 
	  $$
\end{l1norm}
 \end{CompactEnumerate}
\end{theorem}

\subsubsection{Comparison With Previous Works}
\label{sec:compare}
All the previous private algorithms compute a low rank approximation of either  the matrix $\bA$ or its covariance $\bA^{\mathsf T}\bA$. One can compute a factorization from their output at the expense of an extra $O(mn^2)$ time and $O(mn)$ space (Dwork {\it et al.}~\cite{DTTZ14} requires an extra $O(n^3)$ time and $O(n^2)$ space to output an $\lrf$ of $\bA^{\mathsf T}\bA$).  Some  works like~\cite{CSS12,KT13, HR12, HP14} compute $\lrf$ under the spectral norm instead of Frobenius norm. 
We state the previously studied problems    and their differences with \textcolor{red}{Problem}~\ref{prob:private_factor} in~\appref{problems}. 
To summarize the discussion in~\appref{problems}, only Hardt and Roth~\cite{HR12} and Upadhyay~\cite{Upadhyay14} study a problem closest to ours (the differences being that  they do not consider turnstile updates and output a low rank approximation). Therefore, we compare~\thmref{informalspace} only with these two results. 

We do not make any assumptions on the input matrix. 
This allows us to cover matrices of all form in an unified manner. We cover this in more detail in~\secref{discussion}.  In what follows, we compare the additive error, space and time required, and privacy guarantees in more detail. \textcolor{red}{Table}~\ref{table:basicresults} enumerates these comparisons. 
Both Hardt and Roth~\cite{HR12} and Upadhyay~\cite{Upadhyay14} incur a multiplicative error of  $\sqrt{1 + k/p}$, where $p$ is an oversampling parameter (typically, $p=\Theta( k)$), and $m \leq n$. Therefore, for a reasonable comparison, we consider~\thmref{informalspace} when $\alpha = \Theta(1)$ and $m \leq n$. 

\begin{table} [t]
{
\small{
 \begin{center}
\begin{tabular}{|c|c|c|c|c|}
\hline
	  		&	Privacy Notion		  & Additive Error  & Space Required & Streaming  \\ \hline
\multirow{1}{*}{\thmref{low_space_private}}						
& $\bA - {\bA}' =\mathbf{u} \mathbf{v}^{\mathsf T}$	  
 & \textcolor{blue}{$\widetilde{O} \paren{   \paren{ \sqrt{{km}\alpha^{-1}} + \sqrt{kn}}{\varepsilon}^{-1}} $} & \textcolor{blue}{$\widetilde O((m+n)k \alpha^{-1})$} & Turnstile

\\ \hline

\multirow{1}{*}{\thmref{streaming_naive}}						
& $\| \bA - {\bA}'\|_F = 1$	  
 & \textcolor{blue}{$ \widetilde O \paren{ { ( \sqrt{ km  \alpha^{-2}} + \sqrt{ kn  } )}{\varepsilon}^{-1}   }$} & \textcolor{blue}{$\widetilde O((m+n)k \alpha^{-2})$} & Turnstile
\\ \hline 

\multirow{1}{*}{Hardt-Roth~\cite{HR12} }				
& $\bA - {\bA}' = \mathbf{e}_s \mathbf{v}^{\mathsf T}$		 	 	
 & $\widetilde O \paren{ \paren{ \sqrt{km}  + k c \sqrt{n}} {\varepsilon}^{-1}  }$  & $O(mn)$ & $\times$
  \\ \hline 
  
\multirow{1}{*}{Upadhyay~\cite{Upadhyay14}}						
& $\bA - {\bA}' = \mathbf{e}_s \mathbf{v}^{\mathsf T}$	  
 & $\widetilde O\paren{ \paren{ k^2 \sqrt{n+m} }{\varepsilon}^{-1}  }$ & $\widetilde O((m+n)k\alpha^{-1})$ & Row-wise
\\ \hline

\multirow{1}{*}{Lower Bounds}						
& All of the above	  
 & {$  \Omega \paren{  \sqrt{  km} + \sqrt{ kn  } }   $}~\cite{HR12} & \textcolor{blue}{$ \Omega((m+n)k \alpha^{-1})$} (\textcolor{red}{Thm.}~\ref{thm:lower}) & Turnstile
\\ \hline 
\end{tabular}
\caption{\small{Comparison of Private $k$-Rank Approximation ($\|\mathbf{u}\|_2, \|\mathbf{v}\|_2=1$, $\mathbf{e}_s$ is a standard basis, $k \leq 1/\alpha$).}} \label{table:basicresults}
\label{table}
 \end{center}
}}
\end{table}

Our additive error is smaller than Upadhyay~\cite{Upadhyay14} by a factor of $\widetilde{O}(k^{3/2})$. To make a reasonable comparison with Hardt and Roth~\cite{HR12}, we consider their result without coherence assumption:  which roughly says that no single row of the matrix is  significantly correlated with any of the right singular vectors of the matrix.  Then combining Theorem 4.2 and 4.7 in Hardt and Roth~\cite{HR12} results in an additive error $\widetilde{O}((\sqrt{km} + ck  \sqrt{n})\varepsilon^{-1})$, where $c$ is the maximum entry in their projection matrix. 
	In other words, we improve the result of Hardt and Roth~\cite{HR12} by an $\widetilde{O}(c\sqrt{k})$ factor. 

Our algorithms are more efficient than previous algorithms in terms of space and time even though earlier algorithms  output a rank-$k$ matrix and cannot handle updates in the turnstile model. Upadhyay~\cite{Upadhyay14} takes more time than Hardt and Roth~\cite{HR12}. The algorithm of Hardt and Roth~\cite{HR12} uses $O(mn)$ space since it is a private version of Halko {\it et al.}~\cite{HMT11} and has to store the entire matrix: both the  stages of Halko {\it et al.}~\cite{HMT11}  require the matrix explicitly.
One of the motivations mentioned in Hardt and Roth~\cite{HR12} is sparse private incoherent matrices (see the discussion in Hardt and Roth~\cite[Sec 1.1]{HR12}), but their algorithm uses this only to reduce the additive error and not the running time. On the other hand, 
our algorithms use sublinear space. In fact, when $m \cdot \poly(k,1/\alpha) \leq \mathsf{nn}(\bA)$,  our algorithm under $\priv_2$ runs in time linear in $\mathsf{nn}(\bA)$ and almost matches the running time of most efficient 
non-private algorithm 
in the turnstile model
 ~\cite{CW13, BWZ16}, which run in time $O(\mathsf{nn}(\bA) + nk^2\alpha^{-2} + k^3\alpha^{-5})$. 


Previous works consider two matrices neighboring either if they differ in one entry~\cite{HP14,HR13,JXZ15}  (i.e., $\bA - \bA' = \mathbf{e}_i \mathbf{e}_j^{\mathsf T}$) or when they differ in one row by a unit norm~\cite{DTTZ14,HR12,Upadhyay14} (i.e., $\bA - \bA' = \mathbf{e}_i \mathbf{v}^{\mathsf T}$ for some unit vector $\mathbf{v}$), depending on whether a user's data is an  entry of the matrix or a row of the matrix. 
Both these privacy guarantees are special case of $\priv_1$ by setting either $\mathbf{u}$ (in the case of~\cite{DTTZ14,HR12}) or both $\mathbf{v}$ and $\mathbf{u}$ (in the case of~\cite{HR13,HP14}) as a standard basis vector. 

 \subsubsection{Optimality of Our Results} \label{sec:optimality}
\noindent {\bf Tightness of  Additive Error.} 
Hardt and Roth~\cite{HR12} showed 
 a lower bound of $\Omega(\sqrt{kn} + \sqrt{km})$ on additive error  by showing a reduction to the linear reconstruction attack~\cite{DN03}. In other words, any algorithm that outputs a low rank matrix with additive error  $o(\sqrt{kn} + \sqrt{km})$ cannot be {differentially private}! 
This  lower bound  holds even when the  private algorithm can access the private matrix any number of times. Our results show that one can match the lower bound for constant $\alpha$, a setting considered in Hardt and Roth~\cite{HR12}, up to a small logarithmic factor, while allowing  access to the private matrix only  in the turnstile model. 

\begin{l1norm}
We use the reduction to the linear reconstruction attack in order to prove the lower bound. 
The reduction follows from the fact that we can always encode a bit-valued database $D \in \{0,1\}^{N}$ for $N =kn$ as k rows of an $n \times d$ matrix, simply by zeroing out all additional  $n- k$ rows. Note that the resulting matrix only has rank k, and so the optimal rank k approximation to this matrix has zero error. If we could recover a matrix $A'$ such that $\|A - A' \|_p = o(\paren{kn})$, this would mean that for a typical nonzero row $A_i$ of the matrix with $i \in [k]$, we would have $\| A_i - A_i' \|_p = o(n) $ or  
$\| A_i '-A'_i\|_2 \leq n^{-1/2} \| A_i - A_i' \|_p =  o(\sqrt{n}),$ and $ \| A_i - A'_i \|_1 =  o(n)$. Then, by simply rounding the entries, we could reconstruct the original database $D$ in almost all of its entries, giving blatant non-privacy. Since it is known that $(\eps,\delta)$-differentially private algorithms do not admit such reconstruction attacks,
\end{l1norm}

\paragraph{Space Lower Bound and Optimality of the Algorithm Under $\priv_1$.}  \label{sec:informallower}
Our algorithm under $\priv_1$ uses the same space as non-private algorithm up to a logarithmic factor, which is known to be optimal  for $\gamma=0$~\cite{CW09}.  However, we incur a non-zero additive error, $\gamma$, which is inevitable~\cite{HR12}, and it is not clear if we can achieve better space algorithm when $\gamma \neq 0$.  
 This raises the following question:
\begin{quote}
 {\bf Question 2.} What is the space complexity for computing $(\alpha,\beta,\gamma,k)\mbox{-}\lrf$ when $\gamma \neq 0 $? 
\end{quote}
We complement our positive results by showing a lower bound on the space complexity of computing low-rank approximation with non-trivial additive error. The result holds for any randomized algorithm; therefore, also hold for differentially private algorithm. This we believe makes our result of independent interest. 

\begin{theorem} (\thmref{lower}, informal).  \label{thm:informallower}
The space required by any randomized algorithm to solve $(\alpha,1/6,O(m+n),k)\mbox{-}{\lrf}$ in the turnstile update model is  $\Omega((n+m)k\alpha^{-1})$. 
\end{theorem}
As  mentioned above, any  differentially private incurs an additive error $\Omega(\sqrt{km}+\sqrt{kn})$~\cite{HR12}. 
Moreover, earlier literature in differentially private low-rank approximation (i.e., Hardt and Roth~\cite{HR12})  set $\alpha=\sqrt{2}-1$. \thmref{informallower} thus answers Question 3  for all $k \geq 3$  for the parameters studied in the literature, and prove the space optimality of our algorithm.
Our proof uses a reduction from the two-party communication complexity of  
$\mathsf{AIND}$ problem. 
In the $\mathsf{AIND}$ problem, Alice is given  $\bx \in \set{0,1}^N$ and Bob is given an index $i$ together with $\bx_{i+1}, \cdots, \bx_N$. The goal of Bob is to output $\bx_{i}$. 
We use the input of Alice and Bob to construct an instance of low-rank approximation
 and show that any  algorithm for $\lrf$ in the turnstile update model for this instance yields a single round communication protocol for $\mathsf{AIND}$. 


\subsubsection{Some Discussion on Privacy Guarantees, Assumptions, and Relaxation} \label{sec:discussion} 
The two privacy guarantees considered in this paper have natural reasons to be considered. 
$\priv_1$ generalizes the earlier privacy guarantees and captures the setting where any two matrices differ in only one spectrum. Since $\priv_1$ is defined in  terms of the spectrum of  matrices, $\priv_1$ captures one of the natural privacy requirements in all the applications of $\lrf$. 
$\priv_2$ is more stronger than $\priv_1$. To motivate the definition of $\priv_2$, consider a graph, $\cG :=(\cV,\cE)$ that stores career information of people in a set $\cP$ since their graduation. The vertex set $\cV$  is the set of all companies. An edge $e = (u,v) \in \cE$ has weight  $\sum_{p \in \cP} ({t_{p,e}}/t_p)$, where $t_{p,e}$ is the time for which the person $p$ held a job at $v$ after leaving his/her job at $u$,
and $t_p$ is the total time lapsed since his/her graduation.
Graphs like $\cG$ are useful because  the weight on every edge $e=(u,v)$  depends on the number of people  who changed their job status from $u$ to $v$ (and the time they spent at $v$).  Therefore, data analysts might want to  mine such graphs for various statistics. 
 In the past, graph statistics have been extensively studied in the literature of differential privacy for static graph under {\em edge-level privacy} (see, for e.g.,~\cite{DTTZ14,HR12,HP14,Upadhyay13,Upadhyay14}):  the presence or absence of a person corresponds to  a change in a single edge. On the other hand, in graphs like $\cG$,  presence or absence of a person would be reflected on many edges.  If we use earlier results on edge-level privacy to such graphs, it would  lead to either a large additive error or a loss in  privacy parameters $\varepsilon,\delta$. $\priv_2$ is an attempt to understand whether 
 we can achieve any non-trivial guarantee on the additive error without depreciating  the privacy parameters. 

Our choice to  {\em not make any assumptions} such as  symmetry, {incoherence}, or a bound on the Frobenius norm of the input matrix, as made in previous works~\cite{DTTZ14, HR12, HP14, KT13},  allows us to give  results as general as possible and cover  many practical scenarios not covered by previous works. 
For example, this allows us to have a unified algorithmic approach for all the settings mentioned earlier in the introduction (see Corollaries~\ref{cor:streaming_covariance} and~\ref{cor:informalcontinualspace}).  \label{graph}
Moreover, by not making any assumption on the Frobenius norm or coherence of the matrix, we are able to cover many general cases not covered by previous results. For example, adjacency matrices corresponding to graph $\cG$ can have arbitrary large Frobenius norm. Such graphs can be distributed over all the companies and updated every time a person changes his job, and the effect is (possibly) reflected on many edges.  
On the other hand, one would like to protect private information about an individual in $\cG$ or data packets in communication network.   

There are practical learning application such as $\mathsf{word2vec}$ where private data is used to build a relationship graph on objects (such as words-to-words, queries-to-urls, etc) also share all these properties. Other applications   like recommendation systems, clustering, and network analysis, have their input datasets store dynamically changing private information in the form of  $m \times n$ matrices.  For  example,  communication networks  that stores the routing information  passing through various nodes of  the network and online video streaming websites that stores users' streaming history, both store private information in the form of matrix which is updated in a turnstile manner and has arbitrary large Frobenius norm. These data are useful for their respective businesses since they can be used for various analyses~\cite{DFKVV04,LKLS13,Newman06}. 

On the other hand, in some practical applications, matrices are more structured. One such special case is when  the rows of  private matrix $\bA$ have bounded norm and one would like to approximate $\bA^{\mathsf T}\bA$. This problem was studied by Dwork {\it et al.}~\cite{DTTZ14} (see~\appref{problems_DTTZ14} for their formal problem statement). We consider the matrices are updated by inserting one row at a time: all the updates at time $\tau \leq T$ are of the form $\set{i_\tau,\bA^{(\tau)}}$, where $1 \leq i_\tau \leq m$, $\bA^{(\tau)} \in \R^{n}$,
 and $i_\tau \neq i_{\tau'}$ for all $\tau \neq \tau'$. 
We get the following corollary by using $\bA^\mathsf{T} \bA$ as the input matrix (see~\appref{normalized} for the formal description of the algorithm): 
\begin{corollary} (\thmref{streaming_covariance}, informal).  \label{cor:streaming_covariance}Let $\bA$ be an $m \times n$ matrix updated by inserting one row at a time such that every row has a bounded norm $1$ and $m >n$. 
Then 
    there  is an $(\varepsilon,\delta)$-differentially private algorithm under $\priv_2$ that  uses $\widetilde O( nk \alpha^{-2})$ space, and outputs a rank-$k$ matrix $\bB$  such that with probability at least $9/10$, 
	$ \| \bA^{\mathsf T} \bA - \bB \|_F \leq  (1+\alpha)   \| \bA^{\mathsf T} \bA - [\bA^{\mathsf T} \bA]_k \|_F +
	  \widetilde O \paren{ \sqrt{ n k  }/{(\alpha \varepsilon)} } . 
	  $
\end{corollary}

We do not violate the lower bound of Dwork {\it et al.}~\cite{DTTZ14} because their lower bound is valid when $\alpha=0$, which is not possible for low space algorithms due to~\thmref{informallower}. Dwork {\it et al.}~\cite{DTTZ14} bypassed their lower bounds under a stronger assumption known as {\em singular value separation}: the difference between $k$-th singular value and all $k'$-th singular values for $k'>k$ is at least $\omega(\sqrt{n})$.  
In other words, our result shows that we do not need singular value separation while using significantly less space---$\widetilde O(nk/\alpha^2)$ as compared to $O(n^2)$---if we are ready to pay for a small multiplicative error (this also happens for both the problems studied by Blocki {\it et al.}~\cite{BBDS12} though the authors do not mention it explicitly). There are  scenarios where our algorithm perform better, such as when $ \| \bA^{\mathsf T} \bA - [\bA^{\mathsf T} \bA]_k \|_F = O(\sqrt{kn})$ with small singular value separation. For example, a class of rank-$ck$ matrices (for $c \geq 2$) with singular values $\sigma_1, \cdots, \sigma_{ck}$, such that   $ \sigma_k^2 = \sigma_{k'}^2 + O(\sqrt{n/k})$ for $k \leq k' \leq ck$, does not satisfy singular value separation. Moreover, $\| \bA^{\mathsf T} \bA - [\bA^{\mathsf T} \bA]_k \|_F = O( \sqrt{kn})$ and $\| \bA^{\mathsf T}\bA \|_F = O(\sqrt{m})$. In this case, we have  $ \| \bA^{\mathsf T} \bA - \bB \|_F \leq  O ((\alpha \varepsilon)^{-1}  \sqrt{ kn  }) $.  On the other hand, Dwork {\it et al.}~\cite{DTTZ14} guarantees $ \| \bA^{\mathsf T} \bA - \bB \|_F \leq  O ( \varepsilon^{-1}  k\sqrt{ n  }) $. That is, we improve on Dwork {\it et al.}~\cite{DTTZ14} when $\alpha=\Theta(1)$ for a large class of matrices. Some examples include the matrices considered in the literature of matrix completion (see~\cite{MW17} for details).



\section{A Meta Low Space Differentially Private Algorithm Under $\priv_1$} \label{sec:metapriv1}
  \begin{figure}[t!]
 \begin{center}
\fbox{
\begin{minipage}[l]{6.3in}
{
{\bf Initialization.} Set $\rho_1={\sqrt{(1+\alpha)\ln(1/\delta)}}/{\varepsilon}$ and $\rho_2=\sqrt{(1+\alpha)} \rho_1$. Set $\eta=\max\set{k,\alpha^{-1}}$, $t=O(\eta \alpha^{-1}  \log (k/\delta))$, $v=O(\eta \alpha^{-2} \log(k/\delta))$ and $\sigma_\mathsf{min}={16 \log(1/\delta) \sqrt{t (1+\alpha)(1-\alpha)^{-1} \ln(1/\delta)}}/{\varepsilon}$. 
	 Sample $\bOmega \sim \cN(0,1/t)^{m \times t}$, and  $\bPhi \in \R^{(m+n) \times m}, \bPsi  \in \R^{t \times m}$ such that $\bPhi^{\mathsf T} \sim \cD_R$, $ \bPsi \sim \cD_R$  satisfies~\lemref{CW13}. 
	 Sample  $\bS\in \R^{v \times m}, \bT \in \R^{v \times (m+n)}$  such that $\bS \sim \cD_A$,  $\bT^{\mathsf T} \sim \cD_A$  satisfies~\lemref{S}. 
	 Define $\widehat{\bPhi} = \bPhi \bOmega
	 $.  Sample $\bN_1 \sim \cN(0,\rho_1^2)^{t \times (m+n)}$ and $\bN_2 \sim \cN(0,\rho_2^2)^{v \times v}$. 
	Keep $\bN_1,\bN_2,\widehat \bPhi$ private.

\medskip \noindent \textbf{Computing the factorization.} 	  
 On input a matrix $\bA \in \R^{m \times n}$,
\begin{CompactEnumerate}
	\item Set $\widehat{\bA} = \begin{pmatrix} \bA & \sigma_\mathsf{min} \I_m \end{pmatrix}$ by padding a scaled identity matrix $\sigma_\mathsf{min} \I_m$ to the columns of $\bA$, where $\I_m$ denotes an $m \times m$ identity matrix. Compute $\bY_c = \widehat{\bA} \widehat{\bPhi}$,  $\bY_r=  \bPsi \widehat{\bA} + \bN_1,$ and $\bZ= \bS\widehat{\bA} \bT^{\mathsf T} + \bN_2.$
	\item Compute a matrix $\bU \in \R^{m \times t}$ whose columns are orthonormal basis  for the column space of $\bY_c$ and matrix $\bV \in \R^{t \times n}$ whose rows are the orthonormal basis for the row space of $\bY_r$.
	\item Compute a SVD of $\mathbf{S} \bU := \widetilde{\bU}_s \widetilde{\bSigma}_s \widetilde{\bV}_s^{\mathsf T} \in \R^{v \times t}$ 
and a SVD of $ \bV\mathbf{T}^{\mathsf T}:=\widetilde{\bU}_t \widetilde{\bSigma}_t \widetilde{\bV}_t^{\mathsf T} \in \R^{t \times v}.$ 
	\item Compute a SVD of $ \widetilde{\bV}_s \widetilde{\bSigma}_s^{\dagger} [\widetilde{\bU}_s^{\mathsf T} \mathbf{Z} \widetilde{\bV}_t]_k \widetilde{\bSigma}_t^{\dagger} \widetilde{\bU}_t^{\mathsf T} $. Let it is be $\bU' \bSigma' \bV'^{\mathsf T}$. 
	\item Output the matrix  $\widetilde{\bU}=\bU \bU'$ comprising of left singular vectors, diagonal matrix $\widetilde{\bSigma}=\bSigma'$, and the matrix $\widetilde{\bV}=\bV^{\mathsf T} \bV' $ with right-singular vectors. 
\end{CompactEnumerate}
}
\end{minipage}
} \caption{Space Optimal $(\varepsilon,\delta)$-Differentially Private $\lrf$ Under $\priv_1$ ({\scshape Private-Space-Optimal-$\lrf$})} \label{fig:spacelowprivate}
\end{center}
\end{figure}
Our aim in this section is to present a unified algorithmic approach in the form of meta algorithm (see, \figref{spacelowprivate}) for all the settings mentioned in the introduction. This restricts a lot in terms of algorithmic design especially when privacy is the main focus. For example, since one of our goals is private algorithms under turnstile update model, we are restricted to only use linear sketches due to the result of Li {\it et al.}~\cite{LNW14}. 

Even though we are restricted by what approach we can take, we show that advance yet inexpensive post-processing combined with careful analysis can lead to a  unified algorithmic approach for outputting $\lrf$ with significantly small error in all the settings mentioned earlier in the introduction. 
In particular, we make three key observations: (i) there is a way to maintain differentially private sketches of $\bA$ (henceforth, we call such sketches {\em noisy sketches}) that  give sub-optimal additive error,  (ii) one can apply post-processing to these noisy sketches to obtain optimal additive error, and (iii) even though we can only use linear sketches when matrices are updated arbitrarily, we can judiciously exploit the properties of these sketches in our analysis to get tight bounds on the additive error for all the settings considered in this paper.  

To illustrate point (i) and why we need post-processing, consider the following vanilla algorithm, which is reminiscent of Dwork {\it et al.}~\cite{DTTZ14}, for approximating the right singular vector: compute $\bB = \bPhi \bA + \bN_1$, where $\bPhi$ satisfies~\cite[Theorem 1]{MM13} and $\bN_1 \sim \cN(0,\rho_1^2)^{\widetilde O(n^2) \times n}$ for $\rho_1$ as defined in~\figref{spacelowprivate}. The output is $[\bB]_k$, the best rank $k$ approximation of $\bB$. This already gives a good approximation. Let $m \gg n^2$ and let $[\t{\bU}]_k [\t{\bSigma}]_k [\t{\bV}]_k^{\mathsf T}$ be the singular value decomposition of $[\bB]_k$. 
Then by~\cite[Theorem 1]{MM13},
\begin{align*}
	\| \bA - \bA [\t{\bV}]_k [\t{\bV}]_k^{\mathsf T} \|_F & \leq \| (\bA + \bPhi^\dagger \bN_1) - (\bA + \bPhi^\dagger \bN_1) [\t{\bV}]_k [\t{\bV}]_k^{\mathsf T} \|_F+ \| \bPhi^\dagger \bN_1  +  \bPhi^\dagger \bN_1 [\t{\bV}]_k [\t{\bV}]_k^{\mathsf T} \|_F \\
	& \leq \paren{1 - \alpha}^{-1}{\| \bB (\I -  [\t{\bV}]_k [\t{\bV}]_k^{\mathsf T}) \|_F} + O({\|  \bPhi^\dagger \bN_1 + \bPhi^\dagger \bN_1 [\t{\bV}]_k [\t{\bV}]_k^{\mathsf T} \|_F}) \\
	& = \paren{1 - \alpha}^{-1}{\| \bB - [\bB]_k \|_F}  + O({\|  \bPhi^\dagger \bN_1 + \bPhi^\dagger \bN_1 [\t{\bV}]_k [\t{\bV}]_k^{\mathsf T} \|_F})  \\
			& \leq \paren{1 - \alpha}{ \| \bB(\I -  [\bV]_k [\bV]_k^{\mathsf T} )\|_F}  + O({\|  \bPhi^\dagger \bN_1 + \bPhi^\dagger \bN_1 [\t{\bV}]_k [\t{\bV}]_k^{\mathsf T} \|_F}) \\
			& \leq {\paren{1 + \alpha }\paren{1 - \alpha}^{-1}} \Delta_k + O({\|  \bPhi^\dagger \bN_1 \|_F + \| \bPhi^\dagger \bN_1 [\t{\bV}]_k [\t{\bV}]_k^{\mathsf T} \|_F}) .
\end{align*}

The term in $O(\cdot)$ can be bounded using the projection property of $\bPhi$ to get a sub-optimal additive error. The question is whether we can further improve it to get  optimal additive error. We show that with careful algorithmic design and analysis as pointed out in point (ii) above, we can indeed achieve optimal additive error. That is, we can extract top-$k$ singular components of the input matrix $\bA$ from  sketches that are appropriately perturbed to preserve differential privacy. The underlying idea is as follows:
suppose we know the singular value decomposition of $[\bA]_k:=[\bU]_k [\bSigma]_k [\bV]_k^{\mathsf T}$, where $[\bA]_k$ is a best rank-$k$ approximation to $\bA$~\cite{EY36}. Then for finding a matrix $\bB$ such that $\bB \approx \bA$, it suffices to compute $\widetilde{\bU}$ that approximates $[\bU]_k$, $\widetilde{\bSigma}$ that approximates $[\bSigma]_k$, and $\widetilde{\bV}$ that approximates $[\bV]_k$, and set   $\bB := \widetilde{\bU} \widetilde{\bSigma} \widetilde{\bV}^{\mathsf T}$. 
However, this over simplistic overview does not guarantee privacy. The question next is which privacy preserving technique one should use to  ensure  small additive error (and if possible, optimal additive error).

   The two  traditional methods to preserve privacy---input perturbation and output perturbation---do not provide both privacy and small additive error.  For example, if we use  output perturbation to compute the sketches $\bY_c= {\bA} {\bPhi} + \bN$ and $\bY_r = \bPsi {\bA} + \bN'$ for appropriate sketching matrices $\bPhi$ and $\bPsi$ and noise matrices $\bN$ and $\bN'$, then we get an additive error term that can be arbitrarily large (more specifically, depends on the Frobenius norm of $\bA$ and has the form
$\| \bN \cL_\bA \mathbf{N}' \|_F$ for some  linear function $\cL_\bA$ of $\bA$). On the other hand,  input perturbation of the linear sketches $\bY_r = \bPsi \bA$ and  $\bY_c = \bA \bPhi$ followed by a multiplication by Gaussian matrices $\bOmega_1$ and $\bOmega_2$ as in~\cite{BBDS12, Upadhyay14,Upadhyay14b} can leak private data due to a subtle reason. Every row of $\bOmega_1 \bY_c$ (and columns of $\bY_r \bOmega
_2$) has a multivariate Gaussian distribution if the determinant of $\bY_c^{\mathsf T} \bY_c$ ($ \bY_r \bY_r^{\mathsf T}$, respectively) is non zero. If $m < n$, one can prove that computing $\bY_c \bOmega_1$ preserves privacy, but, since, $\bY_r$ is not a full-column rank matrix, the multivariate Gaussian distribution is not defined. The  trick to consider the subspace orthogonal to the kernel space of $\bY_r$~\cite{BBDS12} does not work because span of $\bPsi {\bA}$ and $\bPsi {\bA}'$  may not coincide for neighboring matrices $\bA$ and $\bA'$. If the span do not coincide, then one can easily differentiate the two cases with high probability, violating differential privacy. In fact, until this work, it was not even clear whether using input perturbation yields low rank approximation (see the comment after Theorem IV.2 and discussion in Section V in Blocki {\it et al.}~\cite{BBDS12}). 

All of this presents a pessimistic picture. However, we  show that we incur optimal additive error and preserve privacy if we use input perturbation with a careful choice of parameter to one of the sketches (depending on whether $m \leq n$ or not) and output perturbation to the other two sketches. The intuitive reason why this incurs small  additive error   is the fact that only one of the sketches, $\bY_r$ or $\bY_c$, undergoes output perturbation, so there is no term like $\| \bN \cL_\bA \mathbf{N}' \|_F$  as above. 
 While helpful to the intuition, this does not directly yield optimal additive error. This is because, 
even if we do not get an additive error term with large value like $\| \bN \cL_\bA \mathbf{N}' \|_F$, if not analyzed precisely,  one can either get a non-analytic expression for the error terms or one that is difficult to analyze. 
To get  analytic expressions for all the error terms that are also easier to analyze, we introduce two carefully chosen optimization problems (see~\eqnref{optimizationproblems} below). We carefully tune our analysis so that the intermediate terms satisfy certain properties (see the proof sketch below for exact requirements). 
This allows us to show that $\bY_c$ and $\bY_r$ (or equivalently, their orthonormal bases ${\bU}$ and ${\bV}$),  as formed in~\figref{spacelowprivate}, approximates the span of $[\bU]_{k}$ and $[\bV]_{k}$ up to a small additive error. 

Once we have extracted a ``good enough" ${\bU}$ and ${\bV}$, our problem reduces to computing $\argmin_{\mathsf{rk}(\bX)\leq k } \| \bA -  \bU \bX  \bV \|_F $. This would require storing the whole matrix $\bA$, something that we wish to avoid. To avoid storing the whole $\bA$, we use the fact that $\bS$ and $\bT$ are sampled from a distribution of random matrices with a property that, for all appropriate $\bX$, $ \| \bA -  \bU \bX  \bV \|_F \approx  \| \bS( \bA -  \bU \bX  \bV) \bT^{\mathsf T} \|_F$. In other words, without privacy, $\argmin_{\mathsf{rk}(\bX)\leq k} \|\bS(\bA  -  \bU \bX  \bV)\bT^{\mathsf T}  \|_F$ can be used to get a ``good" approximation of $[\bSigma]_k$. 

The exact method to perform and analyze the approximation of $[\bSigma]_k$ is slightly more involved because we only have access to the noisy version of $\bS\bA \bT$. We show that careful post processing 
  allows one to output an approximation to $\bSigma_k$ under a rotation and a small additive error. 
Finally, we arrive at the following result stated for the case when $m \leq n$. The result when $m >n$ can be derived by just swapping $m$ and $n$ in the theorem below.


  \begin{theorem} 
 \label{thm:meta} 
Let $m,n,k \in \N$ and $\alpha,\varepsilon,\delta$ be the input parameters (with $m \leq n$).  Let 
$\kappa=(1+\alpha)/(1-\alpha)$, $\eta=\max\set{k,\alpha^{-1}}$, and $\sigma_\mathsf{min}={16 \log(1/\delta) \sqrt{t \kappa \ln(1/\delta)}}/{\varepsilon}$. Given an $m \times n$ matrix $\bA$,   {\scshape Private-Space-Optimal-$\lrf$}, described in~\figref{spacelowprivate}, outputs a factorization $\widetilde{\bU}, \widetilde{\bSigma}, \widetilde{\bV}$ such that 
\begin{CompactEnumerate}
 	\item {\scshape Private-Space-Optimal-$\lrf$} is $(3\varepsilon,3\delta)$ differentially private under $\priv_1$.  \label{privacy}
	\item   With probability $9/10$ over the random coins of  {\scshape Private-Space-Optimal-$\lrf$},  \label{privatecorrectness1}
 \begin{align*} 
\| \begin{pmatrix} \bA & \mathbf{0} \end{pmatrix} - \widetilde{\bU} \widetilde{\bSigma} \widetilde{\bV}^{\mathsf T} 
\|_F \leq (1+\alpha) \Delta_k  + O ( \sigma_\mathsf{min} \sqrt{m} + \varepsilon^{-1}  {\sqrt{kn  \ln(1/\delta)}} ),~\text{where}~\mathbf{M}_k=\widetilde{\bU} \widetilde{\bSigma} \widetilde{\bV}^{\mathsf T}.
\end{align*}
	\item The space used by {\scshape Private-Space-Optimal-$\lrf$} is $O((m+n)\eta \alpha^{-1}\log(k/\delta))$.  \label{privatespace}
	\item The  total computational time is 
	$O\paren{  {\mathsf{nn}(\bA) \log (1/\delta) +  {(m^2 +n\eta)\eta {\alpha^{-2}} \log^2(k/\delta) }  + {\eta^3 {\alpha^{-5}}  \log^3 (k/\delta)} } }.$ 
\end{CompactEnumerate}
Here $\begin{pmatrix} \bA & \mathbf{0} \end{pmatrix}$ is the matrix formed by appending an all zero $m \times m$ matrix to the columns of $\bA$.
 \end{theorem}
\begin{proof}[Proof Sketch.] 
The proof of~\thmref{meta} is presented in~\appref{lowspaceprivate}. Here, we give a brief sketch of \textcolor{red}{part}~\ref{privatecorrectness1} (for $m \leq n$) to illustrate the key points that helps us in the discussion of other results in this paper. Let $\h{\bA}$ be as defined in~\figref{spacelowprivate}. \textcolor{red}{Part}~\ref{privatecorrectness1} follows from  the following 
 chain of inequalities and bounding $\| \h{\bA} \bPhi \cL_\bA  \bN_1  \|_F$:
\begin{align}
 \| \mathbf{M}_k - \begin{pmatrix} \bA & \mathbf{0} \end{pmatrix} \|_F  &\leq  \| \mathbf{M}_k - \widehat{\bA} \|_F + O(\sigma_{\min} \sqrt{m}) \nonumber \\
 	& \leq (1+\alpha) \|\widehat{\bA} - [\widehat{\bA}]_k\|_F  +  \| \h{\bA} \bPhi \cL_\bA  \bN_1  \|_F + O(\sigma_{\min} \sqrt{m}) \nonumber \\
& \leq  (1+\alpha)  \| \bA - [\bA]_k \|_F   +  \| \h{\bA} \bPhi \cL_\bA  \bN_1  \|_F + O( \sigma_{\min} \sqrt{m}),  \label{eq:equation3}
\end{align}
where the matrix $\cL_\bA$ satisfies the following properties: (a) $
   	\|  \h{\bA} \bPhi \cL_\bA \bPsi \h{\bA} - \h{\bA} \|_F \leq (1+\alpha) \| \h{\bA} - [\h{\bA}]_k \|_F$, (b) $\cL_\bA$ has rank at most $k$, and (c) $ \bPsi \bA \bPhi \cL_\bA$ is a rank-$k$ projection matrix. We use subadditivity of norm to prove the first inequality and Weyl's perturbation theorem (\thmref{Weyl}) to prove the third inequality. Proving the second inequality is the  technically involved part. For this, we need to find a candidate $\cL_\bA$. Let suppose we have such a candidate $\cL_\bA$ with all the three properties. Then we can show that
\begin{align}
 \min_{\mathsf{rk}(\bX) \leq k} \| \bU \bX \bV - \bB \|_F &\leq \|  \h{\bA} \bPhi \cL_\bA \bPsi \h{\bA} - \h{\bA}\|_F + \| \h{\bA} \bPhi \cL_\bA  \bN_1  \|_F + \| \bS^\dagger \bN_1 (\bT^\dagger)^\mathsf T \|_F \nonumber \\
 	&\leq (1+\alpha) \|\widehat{\bA} - [\widehat{\bA}]_k\|_F +  \| \h{\bA} \bPhi \cL_\bA  \bN_1  \|_F  +  \|  \bS^\dagger \bN_2 (\bT^{\mathsf T})^\dagger \|_F,  \label{eq:overallbound}   
\end{align}
where $\bB= \bA + \bS^\dagger \bN_1 (\bT^\dagger)^\mathsf T$. The first inequality follows from the subadditivity of Frobenius norm, the fact that $\bU$ and $\bV$ are orthonormal bases of $\bY_c$ and $\bY_r$, and property (b) to exploit that minimum on the left hand side is over rank-$k$ matrices. We then use the approximation guarantee of property (a) to get the second inequality. 
Using~\lemref{orthonormal} and the fact that $\bS$ and $\bT$ satisfies~\lemref{S}, we can lower bound the left hand side of~\eqnref{overallbound} up to an additive term as follows:
\[ \| \begin{pmatrix} \bA & \mathbf{0} \end{pmatrix} - \widetilde{\bU} \widetilde{\bSigma} \widetilde{\bV}^{\mathsf T} 
\|_F -  \|  \bS^\dagger \bN_2 (\bT^{\mathsf T})^\dagger \|_F \leq (1+\alpha)^3 \min_{\mathsf{rk}(\bX) \leq k} \| \bU \bX \bV - \bB \|_F , \]
where $\t{\bU},\t{\bSigma}, $ and $\t{\bV}$ are as in~\figref{spacelowprivate}. 
We upper bound the right hand side of~\eqnref{overallbound} by using Markov's inequality combined with the fact that both $\bS$ and $\bT$ satisfy~\lemref{S} and $\cL_\bA$ satisfies property (c). 
Scaling the value of $\alpha$ by a constant gives \textcolor{red}{part}~\ref{privatecorrectness1}. So all that remains is to find a candidate matrix $\cL_\bA$. 
To get such an $\cL_\bA$,  we construct and solve the following two closely related  optimization problems: 
 \begin{align}
& \mathsf{Prob}_1: \min_{\bX} \| \bPsi( \h{\bA} \bPhi ([\h{\bA}]_k \bPhi)^{\dagger} \bX - \h{\bA}) \|_F \quad \text{and} & \mathsf{Prob}_2: \min_{\bX} \| \h{\bA} \bPhi ([\h{\bA}]_k \bPhi)^{\dagger} \bX - \h{\bA} \|_F.  \label{eq:optimizationproblems}
\end{align} 
We prove that a solution to $\mathsf{Prob}_1$ gives us a candidate  $\cL_\bA$. This completes the proof.  
\end{proof}

\paragraph{Adapting to Continual Release Model.} 
Until now, we gave algorithms that produce the output only at the end of the stream. 
There is a related model called {\em $(\varepsilon,\delta)$-differential privacy  under $T$-continual release}~\cite{DNPR10}. In this model, the server receives a stream of length $T$ and produces an output  after every  update,
  such that every output 
is $(\varepsilon,\delta)$-differentially private. 
Our meta algorithm can be easily modified to work in this model while incurring an extra $\log T$ factor. 
We utilize the fact that we only store noisy linear sketches  of the private matrix during the updates 
and  low-rank factorization is computed through post-processing on only the noisy sketches. So, we can use the generic transformation~\cite{CSS12,DNPR10} to  maintain the sketch form of the updates. A factorization for any range of time can be then done by aggregating all the sketches for the specified range using range queries. 
This gives the {\em first instance} of algorithm that provides differentially private continual release of $\lrf$. We show the following.
\begin{corollary}  \label{cor:informalcontinualspace} (\thmref{continuallowspace}, informal) Let $\bA$ be an $m \times n$ private matrix  with $m \leq n$ streamed over $T$ time epochs. 
Then 
there is an $(\varepsilon,\delta)$-differentially private algorithm under $\priv_1$ that outputs a rank-$k$ factorization  under the continual release for  $T$ time epochs such that 
 $ \gamma = \widetilde O ({  {\varepsilon}^{-1} ({ \sqrt{{mk} \alpha^{-1}} + \sqrt{kn}}) \log T  })$ with probability $9/10$.
\end{corollary}

\paragraph{From $\priv_1$ to $\priv_2$.} If we try to use the idea described above to prove differential privacy under $\priv_2$, we end up with an additive error  that depends linearly on $\min \set{{m},{n}}$. This is because we need to perturb the input matrix by a noise proportional to $\min \{\sqrt{km},\sqrt{kn} \}$ to preserve differential privacy under $\priv_2$.  We show that by maintaining noisy sketches $\bY= \bA \bPhi + \bN_1$ and $\bZ=\bS \bA + \bN_2$ for appropriately chosen noise matrices $\bN_1$ and $\bN_2$ and sketching matrices $\bPhi$ and $\bS$, followed by some post processing, we can have  an optimal error differentially private algorithm under $\priv_2$.  Here, we require  $\bS$ to satisfy the same property as in the case of $\priv_1$.  However, the lack of symmetry between $\bS$ and $\bPhi$ requires us to decouple the effects of  noise matrices to get a tight bound on the additive error.

\section{Local Differentially Private Algorithm} \label{sec:informallocal}
Till now, we have considered a single server that receives the private matrix in a streamed manner. We next consider a stronger
 variant of differential privacy known as {\em local differential privacy} (LDP) \cite{DJW13,DMNS06,EGS03,W65}. In the \emph{local} model, each individual applies a differentially
private algorithm locally to their data and shares only the output of
the algorithm---called a report---with a server that
aggregates users' reports. A multi-player protocol is $(\varepsilon,\delta)$-LDP if for all possible inputs and runs of the
protocol, the transcript of player $i$'s interactions with the server
is $(\varepsilon,\delta)$-LDP.    

One can study two variants of local differential privacy depending on whether the server and the users interact more than once or not (see  \figref{LDP}).    
In the interactive
     variant, the server sends several messages, each to a subset of
    users.  In
    the noninteractive variant, the server sends
 a single message
   to all the users at the start of the protocol and sends no message after that.         
    Smith, Thakurta, and Upadhyay~\cite{STU16}   
 \begin{wrapfigure}{r}{0.75\textwidth}
  \includegraphics[height=1in]{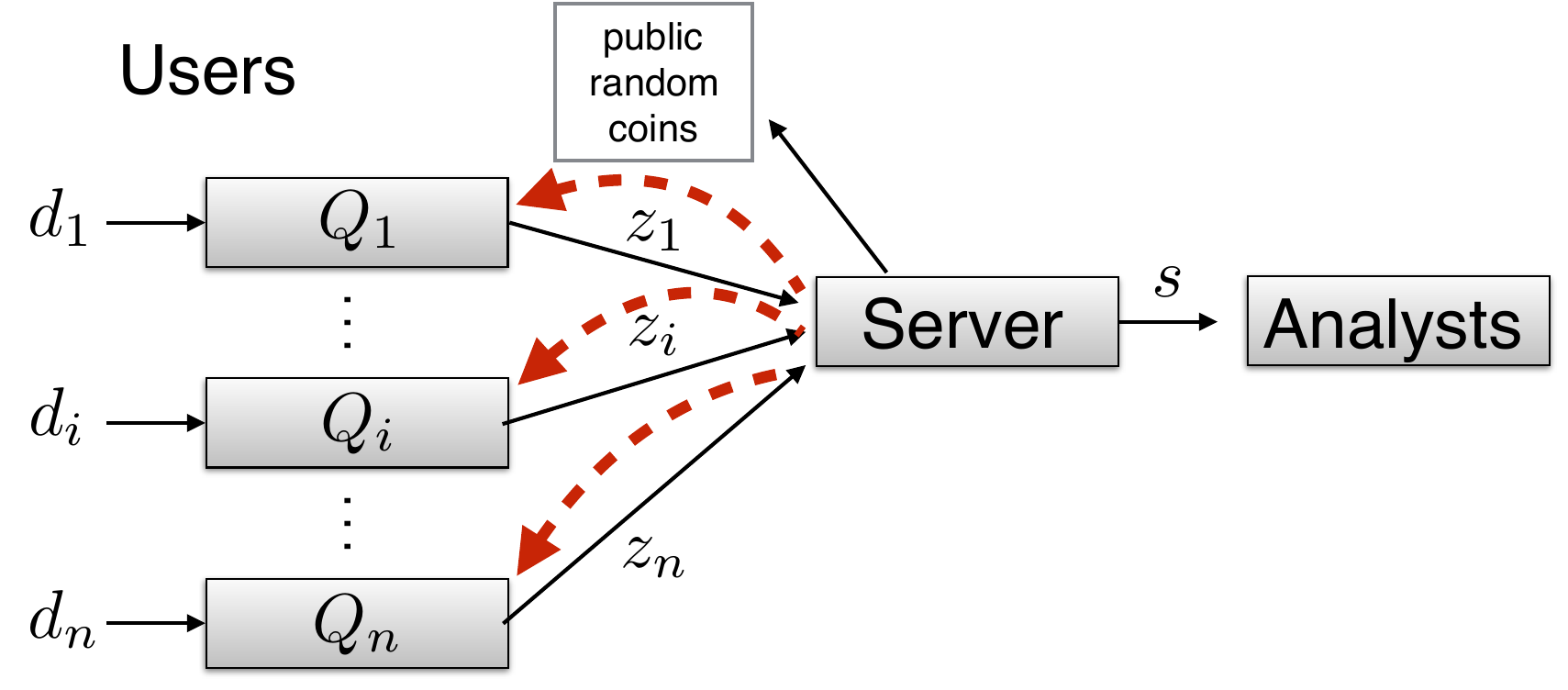}
  \hspace{6pt}  \rule{0.4pt}{1in}
  \hspace{6pt}  \includegraphics[height=1in]{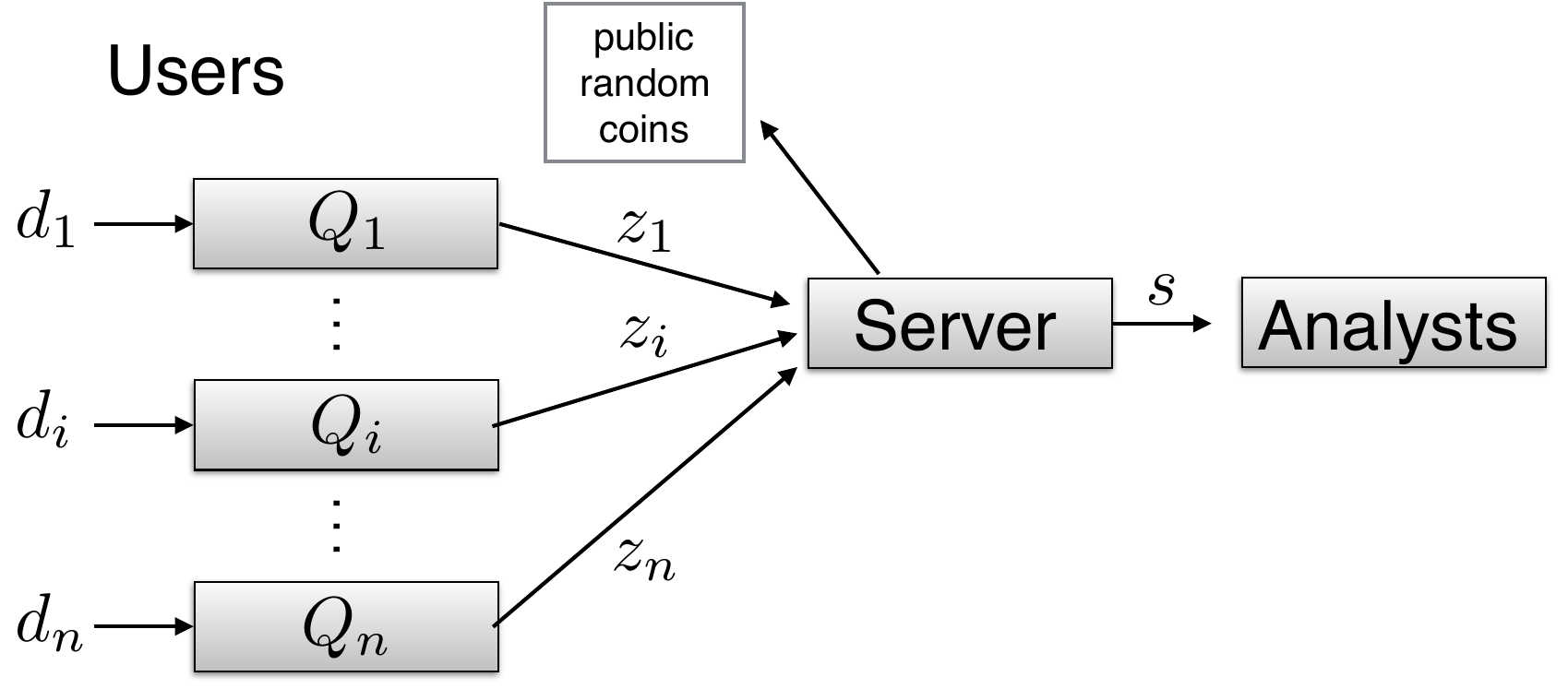}  
  \caption{LDP with (left) and
    without (right) interaction.}
  \label{fig:LDP}
\end{wrapfigure} 
argued that noninteractive locally private algorithms are ideal for  implementation.  

The natural extension of \textcolor{red}{Problem}~\ref{prob:private_factor} in the local model    is when the matrix  is distributed among the users
     such that every user has one row of the matrix and users are responsible for the privacy of their row vector. 
Unfortunately, known private algorithms (including the results presented till now) do not yield non trivial additive error in the local model.  For example, if we convert~\thmref{informalspace} to the local model, we end up with an additive error $\widetilde O(\sqrt{kmn})$. This is worse than the trivial bound of $O(\sqrt{mn})$, for example, when $\bA \in \set{0,1}^{m \times n}$, a trivial output of all zero matrix incurs an error at most $O(\sqrt{mn})$. In fact, existing lower bounds in the local model suggests that one is likely to incur an error which is $O(\sqrt{m})$ factor worse than in the central model, where $m$ is the number of users. However, owing to the result of Dwork {\it et al.}~\cite{DTTZ14}, we can hope to achieve non-trivial result for differentially private principal component analysis (see,~\textcolor{red}{Definition} \ref{def:dppca}) 
leading us to ask
\begin{quote}
 {\bf Question 3.} Is there a  locally  private algorithm for low rank principal component analysis?
\end{quote}
This problem has been studied without privacy under the {\em row-partition model}~\cite{BCL05, BKLW14, BRB08, BWZ16,GSS17,elemental,QOSG02,TD99}). 
 Even though there is a rich literature on local differentially private  algorithms~\cite{agrawal2009frapp,BS15,DJW13,rappor,EGS03,HKR,KLNRS11,MS06,W65}, the known approaches to convert existing (private or distributed non-private) algorithms  to locally private algorithms either leads to a large additive error or require interaction. 
 We exploit the fact that our meta algorithm only stores differentially private sketches of the input matrix  to give a noninteractive algorithm for low-rank principal component analysis (PCA) under local differential privacy. 

\begin{theorem} (\thmref{local}, informal).  \label{thm:informallocal}Let $\bA$ be an $m \times n$ matrix (where $m > n$) such that every user holds one row of $\bA$. 
Then 
   there  is a noninteractive $(\varepsilon,\delta)$-LDP under $\priv_2$ that  
   outputs a rank-$k$ orthonormal matrix  $\bU$ 
   such that
	$  	  \| \bA - \bU \bU^{\mathsf T}\bA \|_F \leq (1+ \alpha)\Delta_k + {O}{( k \alpha^{-2} \eps^{-1}  \sqrt{m} \log(k/\delta) \sqrt{\log(1/\delta)} )}$ with probability at least $9/10$. Moreover, every user sends $ \widetilde O (k^2 \alpha^{-4} \log^2(k/\delta))$ bits of communication. 
\end{theorem}
The above theorem gives the first instance of non-interactive algorithm that computes  low-rank principal component analysis in the model of  local differential privacy. The best known lower bound on additive error for differentially private PCA is by Dwork {\it et al.}~\cite{DTTZ14} in the central model in the static data setting. Their lower bound on the additive error is $\widetilde \Omega(k \sqrt{n})$ for squared Frobenius norm and has no multiplicative error. 
An interesting question from our result is to investigate how close or far we are from optimal error. 

	
\paragraph{Roadmap to Appendices.} All our proofs and algorithms are in appendices.	
 We prove~\thmref{meta} in \appref{spaceprivate}, and, also provide and analyze a space efficient algorithm under $\priv_2$. 
	 \appref{streaming} covers the differentially private $\lrf$  in the  turnstile update model under both $\priv_1$ and $\priv_2$ answering Question 1. We extend these algorithms to continual release model in
	 \appref{continual}. 
	 \appref{lower} gives the space lower bound when $\gamma \neq 0$ answering Question 2 and discussed in~\secref{informallower}.
	 \appref{local} gives noninteractive local differentially private $\lrf$ algorithm under $\priv_2$ answering Question 3 and discussed in~\secref{informallocal}.
	 \appref{empirical} gives the empirical evaluation of our results. 

\paragraph{Acknowledgement.} The author wishes to thank Adam Smith for many useful discussions. Even though he declined to be coauthor in this paper, this paper would have not existed without his insightful inputs. 	

\pagebreak   

\tableofcontents
\newpage
\listoffigures
 
\listoftables




 





\newpage

\begin{appendix}

\section{Related Work} \label{sec:related} \label{app:related}
Low-rank approximation ($\lra$), where the goal is to output a matrix $\bB$ such that $\|\bA - \bB\|_F$ is close to the optimal $\lra$, of large data-matrices  has received a lot of attention in the recent past in the private as well as the non-private setting.  In what follows, we give a brief exposition of works  most relevant to this work.

In the non-private setting, previous works have either used random projection or random sampling (at a cost of  a small additive error) to give low-rank approximation~\cite{AM07,CW09,DKM06,DM05,DRVW06,FKV04,KN14,PTRV98,RV07,Sarlos06}. 
Subsequent works~\cite{CMM17, CW13, CW17, DV06,MZ11,MM13,NDT09,Sarlos06} achieved a run-time that depends linearly on the input sparsity of the matrix. 
In a series of works, Clarkson and Woodruff~\cite{CW09,CW13} showed space lower bounds and almost matching space algorithms. 
Recently, Boutsidis {\it et al.}~\cite{BWZ16} gave the first space-optimal algorithm for low-rank approximation under turnstile update mode, but they do not optimize for run-time. A optimal space algorithm was recently proposed by Yurtsever {\it et al.}~\cite{YUTC17}. 
Distributed PCA algorithms in the row partition model and arbitrary partition model has been long studied~\cite{BCL05, BKLW14, BRB08, BWZ16,GSS17,elemental,QOSG02,TD99}.

In the private setting, low-rank approximation $\lra$ has been studied under a privacy guarantee called differential privacy. Differential privacy was  introduced by Dwork {\it et al.}~\cite{DMNS06}. 
Since then, many algorithms for preserving differential privacy have been proposed in the literature~\cite{DR14}.  All these mechanisms have a common theme: they perturb the output before responding to  queries. Recently, Blocki {\it et al.}~\cite{BBDS12} and Upadhyay~\cite{Upadhyay13} took a complementary approach. They perturb the input reversibly and then perform a random projection of the perturbed matrix.

Blum  {\it et al.}~\cite{BDMN05} first studied the problem of differentially private $\lra$ in the Frobenius norm. This was improved by Hardt and Roth~\cite{HR12}  under the low coherence assumption. Upadhyay~\cite{Upadhyay14} later made it a single-pass. Differentially private $\lra$ has been studied in the spectral norm as well by many works~\cite{CSS12,KT13,HR13,HP14}.
Kapralov and Talwar~\cite{KT13} and Chaudhary  {\it et al.}~\cite{CSS12} studied the spectral $\lra$ of a matrix by giving a matching upper and lower bounds for privately computing the top $k$ eigenvectors of a matrix with pure differential privacy (i.e., $\delta=0$). In subsequent works Hardt and Roth~\cite{HR13} and Hardt and Price~\cite{HP14} improved the approximation guarantee with respect to the spectral norm by using {\em robust private subspace iteration} algorithm. 
Recently, Dwork  {\it et al.}~\cite{DTTZ14} gave a tighter analysis of Blum {\it et al.}~\cite{BDMN05} to give an optimal approximation to the right singular space, i.e., they gave a $\lra$ for the covariance matrix.  
Dwork  {\it et al.}~\cite{DNPR10} first considered streaming algorithms with privacy under the model of {\em pan-privacy}, where the internal state is known to the adversary. They gave private analogues of known sampling based streaming algorithms to answer various counting tasks.
This was followed by results on online private learning~\cite{DTTZ14,JKT12,TS13}.


\section{Notations and Preliminaries} \label{sec:prelims}  \label{app:prelims}
We give a brief exposition of notations and linear algebra to the level required to understand this paper. We refer the readers to standard textbook on this topic for more details~\cite{Bhatia}. We let $\N$ to denote the set of natural numbers and $\R$ to denote the set of real numbers.  For a real number $x \in \R$, we denote by $|x|$ the absolute value of $x$.  We use boldface lowercase letters to denote vectors, for example, $\bx$, and $\bx_1, \ldots, \bx_n$ to denote the entries of $\bx$, and  bold-face capital letters to denote matrices, for example, $\bA$. For two vectors $\bx$ and $\by$, we denote by $\brak{\bx,\by} = \sum_i \bx_i \by_i$ the inner product of $\bx$ and $\by$. 
We let $\mathbf{e}_1, \ldots , \mathbf{e}_n$ denote the standard basis vectors in $\R^n$, i.e., $\mathbf{e}_i$ has entries $0$ everywhere except for the position $i$ where the entry is $1$. We denote by $\begin{pmatrix} \mathbf{A} & \mathbf{b} \end{pmatrix}$ the matrix formed by appending the matrix $\mathbf{A}$ with the vector $\mathbf{b}$. We use the notation $\I_n$ to denote the identity matrix of order $n$ and $\mathbf{0}^{m \times n}$ the all-zero $m \times n$ matrix. Where it is clear from the context, we drop the subscript.

For a $m \times n$ matrix $\bA$, we denote by $\bA_{ij}$  the $(i,j)$-th entry of $\bA$. We denote by $\mathsf{vec}(\bA)$  the vector of length $mn$ formed by the entries of the matrix $\bA$, i.e., for an $m \times n$ matrix $\bA$, the $((i-1)n+j)$-th entry of $\mathsf{vec}(\bA)$ is $\bA_{ij}$, where $1\leq i \leq m, 1\leq j \leq n$. The {\em transpose} of a matrix $\bA$ is a matrix $\mathbf{B}$ such that $\mathbf{B}_{ij}=\bA_{ji}$. We use the notation $\bA^{\mathsf T}$\label{defn:transpose} to denote the transpose of a matrix $\bA$.  For a matrix $\bA$, we denote the best $k$-rank approximation of $\bA$ by $[\bA]_k$ and its Frobenius norm by $\| \bA \|_F$. For a matrix $\bA$, we use the symbol $\mathsf{r}(\bA)$ to denote its {\em rank} and $\det(\bA)$ to denote its {\em determinant}.  A matrix $\bA$ is  a {\em non-singular matrix} if $\det(\bA) \neq 0.$

\subsection{Definitons}

\begin{definition} 
(Differentially Private Low Rank Factorizataion in the General Turnstile Update Model). 
\label{def:dpturnstile}
Let the private matrix $\bA \in \R^{m \times n}$ initially be all zeroes. At every time epoch, the curator receives an update in the form of the tuple $\set{i,j,\Delta}$, where $1 \leq i \leq m, 1\leq j \leq n,$ and $ \Delta \in \R$---each update results in a change in the $(i,j)$-th entry of the matrix $\bA$ as follows: $\bA_{i,j} \gets \bA_{i,j}+\Delta$.
Given parameters $0 <\alpha,\beta <1$ and $\gamma$,   and the target rank $k$, the curator   is required  to output 
a differentially-private  rank-$k$ matrix factorization $\widetilde{\bU}_k, \widetilde{\bSigma}_k$, and $\widetilde{\bV}_k$ of $\bA$ at the end of the stream, such that, with probability at least $1-\beta$,
$$  { \| \bA - \widetilde{\bU}_k \widetilde{\bSigma}_k \widetilde{\bV}_k^{\mathsf T} \|_F \leq (1+\alpha)  \Delta_k +\gamma },$$ 
{where}~$\Delta_k:= \|\bA - [\bA]_k\|_F$ with $[\bA]_k$ being the best rank-$k$ approximation of $\bA$. 
\end{definition}

\begin{definition}
(Differentially Private Low Rank Factorizataion in the Continual Release Model).
\label{def:dpcontinual}
Let the private matrix $\bA \in \R^{m \times n}$ initially be all zeroes. At every time epoch $t$, the curator receives an update in the form of the tuple $\set{i,j,\Delta}$, where $1 \leq i \leq m, 1\leq j \leq n,$ and $ \Delta \in \R$---each update results in a change in the $(i,j)$-th entry of the matrix $\bA$ as follows: $\bA_{i,j}^{(t)} \gets \bA_{i,j}^{(t-1)}+\Delta$.
Given parameters $0 <\alpha,\beta <1$ and $\gamma$,   and the target rank $k$, the curator   is required  to output 
a differentially-private  rank-$k$ matrix factorization $\widetilde{\bU}_k^{(t)}, \widetilde{\bSigma}_k^{(t)}$, and $\widetilde{\bV}_k^{(t)}$ of $\bA$ at every time epoch, such that, with probability at least $1-\beta$,
$$  { \| \bA^{(t)} - \widetilde{\bU}_k^{(t)} \widetilde{\bSigma}_k^{(t)} (\widetilde{\bV}_k^{(t)})^{\mathsf T} \|_F \leq (1+\alpha)  \Delta_k^{(t)} +\gamma },$$ 
{where}~$\Delta_k^{(t)}:= \|\bA^{(t)} - [\bA^{(t)}]_k\|_F$ with $[\bA^{(t)}]_k$ being the best rank-$k$ approximation of $\bA^{(t)}$. 
\end{definition}

\begin{definition}
(Local Differentially Private Principal Component Analysis Problem). 
\label{def:dppca}
Given an $m \times n$ matrix $\bA$ distributed among $m$ users where each user holds a row of the matrix $\bA$, a rank parameter $k$, and an accuracy parameter $0 < \alpha < 1$, design a local differentially private algorithm  which, upon termination, outputs an orthonormal matrix $\bU$ such that
$$  { \| \bA - \bU \bU^{\mathsf T}\bA \|_F \leq (1+\alpha) \| \bA - [\bA]_k  \|_F +\gamma }$$ with probability at least $1-\beta$,  
 and the communication cost of the algorithm is as small as possible. Furthermore, the algorithm should satisfies $(\varepsilon,\delta)$-LDP. 
\end{definition}

In this paper, we use various concepts and results from the theory of random projections, more specifically the Johnson-Lindenstrauss transform and its variants. 

\begin{definition} \label{thm:JL} 
Let $\alpha, \delta >0$. A distribution $\cD$ over ${t \times n}$ random matrices satisfies {\em $(\alpha,\beta)$-Johnson Lindenstrauss property} {(\em {\sf JLP})} if, for   any unit vector $\bx \in \R^n$, we have $$\p_{\bPhi \sim \cD} [\|\bPhi \bx \|_2^2  \in (1 \pm \alpha)] \geq 1 - \beta.$$
\end{definition}

\begin{definition} \label{def:genreg}
 A distribution $\cD_R$ of $t \times m$ matrices satisfies {\em $(\alpha,\delta)$-subspace embedding for generalized regression} if it has  the following property:  for any matrices $\mathbf{P} \in \R^{m \times n}$ and $\mathbf{Q} \in \R^{m \times n'}$ such that $\mathsf{r}(\mathbf{P}) \leq r$, with probability $1-\delta$ over $\bPhi \sim \cD_R$, if 
 $$\widetilde{\bX} = \argmin_{\bX} \| \bPhi (\mathbf{P} \bX - \mathbf{Q}) \|_F \quad \text{and} \quad \widehat{\bX}= \argmin_{\bX \in \R^{n \times n'}}  \| \mathbf{P} {\bX} - \mathbf{Q} \|_F ,$$ then 
	$ \| \mathbf{P} \widetilde{\bX} - \mathbf{Q} \|_F \leq (1+\alpha)  \| \mathbf{P} \widehat{\bX} - \mathbf{Q} \|_F.$
 \end{definition}

 \begin{definition} \label{def:affine}
A distribution $\cD_A$ over $v \times m$  matrices satisfies {\em $(\alpha,\delta)$-affine subspace  embedding}  if it has  the following property: for any matrices $\mathbf{D} \in \R^{m \times n}$ and $\mathbf{E} \in \R^{m \times n'}$ such that $\mathsf{r}(\mathbf{D}) \leq r$, with probability $1-\delta$ over $\bS \sim \cD_A$, simultaneously for all $\bX \in \R^{n \times n'}$, 
	$ \| \bS( \mathbf{D} {\bX} - \mathbf{E}) \|_F^2 = (1\pm \alpha)  \| \mathbf{D} \bX - \mathbf{E} \|_F^2.$
 \end{definition}

We use the symbol $\cD_R$ to denote a distribution that satisfies $(\alpha,\delta)$-subspace embedding for generalized regression and $\cD_A$ to denote a distribution that satisfies $(\alpha,\delta)$-affine subspace embedding. 
\subsection{Linear Algebra}
A matrix is called a {\em diagonal matrix} if the  non-zero entries are all along the principal diagonal.   An $m \times m$ matrix $\bA$ is a {\em unitary matrix}  if $\bA^{\mathsf T} \bA = \bA \bA^\mathsf{T} = \I_m$. Additionally, if the entries of the matrix $\bA$ are real, then such a matrix is called an {\em orthogonal matrix}. For an $m \times m$ matrix $\bA$, the {\em trace} of $\bA$ is the sum of its diagonal elements.  We use the symbol $\tr(\bA)$  to denote the trace of matrix $\bA$.\label{defn:trace} We use the symbol $\det(\bA) $ to denote the determinant of matrix $\bA$.
	 
Let $\bA$ be an $m \times m $ matrix. Its {\em singular values} are the eigenvalues of the matrix $\sqrt{\bA^{\mathsf T} \bA}$. The eigenvalues of the matrix $\sqrt{\bA^{\mathsf T} \bA}$ are real because ${\bA^{\mathsf T} \bA}$ is a symmetric matrix and has a well-defined {\em spectral decomposition}\footnote{A spectral decomposition of a symmetric matrix $\bA$ is the representation of a matrix in form of its eigenvalues and eigenvectors: $\sum_{i} \lambda_i \mathbf{v}_i \mathbf{v}_i^{\mathsf T}$, where $\lambda_i$ are the eigenvalues of $\bA$ and $\mathbf{v}_i$ is the eigenvector corresponding to $\lambda_i$.}~\cite{Bhatia}. 

The {\em singular-value decomposition} (SVD) of an $m \times n$ rank-$r$ matrix $\bA$ is a decomposition of $\bA$ as a product of three matrices, $\bA = \bU \bSigma \bV^{\mathsf T}$ such that $\bU \in \R^{m \times r}$ and $\bV \in \R^{n \times r}$ have orthonormal columns and $\bSigma \in \R^{r \times r}$ is a diagonal matrix with singular values of $\bA$ on its diagonal. One can equivalently write it in the following form:
\[ \bA = \sum_{i=1}^{\mathsf{r}(\bA)} \sigma_i\mathbf{u}_i  \bv^{\mathsf T}_i, \]
where $ \mathbf{u}_i $ is the $i$-th column of $\bU$, $\bv_i$ is the $i$-th column of $\bV$, and $\sigma_i$ is the $i$-th diagonal entry of $\Sigma$. 

One can derive a lot of things from the singular value decomposition of a matrix. For example, 
\begin{enumerate}
	\item The {\em Moore-Penrose pseudo-inverse} of a matrix $\bA =  \bU \bSigma \bV^{\mathsf T}$ is denoted by $\bA^\dagger$ and has a SVD $\bA^\dagger =  \bV \bSigma^\dagger \bU^{\mathsf T}$, where $\bSigma^\dagger$ consists of inverses of only non-zero singular values of $\bA$. In other words, 
	\[ \bA^\dagger =  \sum_{i=1}^{\mathsf{r}(\bA)} \frac{1}{\sigma_i} \mathbf{u}_i  \bv^{\mathsf T}_i, \]
	where $\mathsf{r}(\bA)$ is the number of non-zero singular values of $\bA$.
	\item Let $\sigma_1\geq \cdots \geq \sigma_k \geq \cdots \geq \sigma_{\mathsf{r}(\bA)}$ be the singular values of $\bA$. Then 
	\[ [\bA]_k = \sum_{i=1}^{k} \sigma_i\mathbf{u}_i  \bv^{\mathsf T}_i, \]
	\item The trace of a matrix $\bA$ can be represented in form of the singular values of $\bA$ as follows: $\tr(\bA) =  \sum_i \sigma_i$. Similarly, the determinant of a matrix $\bA$ is $\det(\bA)= \prod_i \sigma_i$. Moreover, the Frobenius norm of $\bA$ is $\sum_i \sigma_i^2$.
\end{enumerate}


We use few unitary matrices, which we define next. Let $N$ be a power of $2$. A Walsh-Hadamard matrix of order $N$ is an $N \times N$ matrix formed recursively as follows:
\[ \bW_N = \frac{1}{\sqrt{2}} \begin{pmatrix} \bW_{N/2} & \bW_{N/2} \\ \bW_{N/2} & -\bW_{N/2} \end{pmatrix}  \qquad \bW_1 :=(1).  \]
	A Walsh-Hadamard matrix is a unitary matrix. We often drop the subscript wherever it is clear from the context. A {\em randomized Walsh-Hadamard matrix} is a matrix product of a Walsh-Hadamard matrix and a random diagonal matrix with entries $\pm 1$ picked according to the following probability distribution:
	\[ \p [X=1] = \p[X=-1]=1/2.  \]

A {\em discrete Fourier matrix} of order $n$ is an $n \times n$ matrix such that the $(i,j)$-th entry is $\omega^{(i-1)(j-1)}$, where $\omega$ is the $n$-th root of unity, i.e., $\omega = e^{-2\pi \iota/n}$.

\subsection{Gaussian Distribution}

Given a random variable $x$, we denote by $\cN(\mu, \rho^2)$ the fact that  $x$ has a normal Gaussian distribution with mean $\mu$ and variance $\rho^2$. The Gaussian distribution is invariant under affine transformation, i.e., if $X \sim \cN(\mu_x, \sigma_x)$ and $Y \sim \cN(\mu_y, \sigma_y)$, then $Z=aX +bY$ has the distribution $Z \sim \cN(a \mu_x + b \mu_y, a \sigma_x^2 + b\sigma_y^2)$. This is also called the {\em rotational invariance} of Gaussian distribution. By simple computation, one can verify that the tail of a standard Gaussian variable decays exponentially. More specifically, for a random variable $X \sim \cN(0,1)$, we have  $ \p \sparen{|X| >t}  \leq 2 e^{-t^2/2}. $

Our proof uses an analysis of multivariate Gaussian distribution. The multivariate Gaussian distribution is a generalization of univariate Gaussian distribution. Let $\mu$ be an $N$-dimensional vector. An $N$-dimensional multivariate random variable, $\bx \sim \cN(\mathbf{\mu}, \bLambda)$, where $\bLambda = \E [(\bx - \mu)(\bx -\mu)^{\mathsf T}]$ is the $N \times N$ covariance matrix, has the probability density function given by 
$\PDF_\mathbf{X} (\bx) := \frac{e^{- \bx^{\mathsf T}    \bLambda^\dagger   \bx/2}}{\sqrt{(2 \pi)^{\mathsf{r}(\bSigma)} \det(\bLambda)}}   \label{eq:pdfmultivariate}.$ 
 If $\bLambda$ has a non-trivial kernel space, then the multivariate distribution is undefined. However, in this paper, all our covariance matrices have only trivial kernel.
Multivariate Gaussian distributions is invariant under affine transformation, 
i.e., if $\by = \bA\bx+\mathbf{b}$, where $\bA \in \R^{M \times N}$ is a rank-$M$ matrix and $\mathbf{b} \in \R^M$, then $\by \sim \cN (\bA\mu+\mathbf{b},\bA \bLambda \bA^{\mathsf T}   )$.


\section{Basic Results Used in This Paper} \label{app:auxiliary}
Our proofs uses various concepts and known results about random projections, pseudo-inverse of matrices and gaussian distribution. In this section, we cover them up to the level of exposition required to understand this paper. We refer to the excellent book by Bhatia~\cite{Bhatia} for more exposition on pseudo-inverses, and Woodruff~\cite{Woodruff14}.
\subsection{Random Projections.}
Random projection has been used in computer science for a really long time. Some partial application includes metric and graph embeddings~\cite{Bou85, LLR94}, computational speedups~\cite{Sarlos06, Vem05}, machine learning~\cite{BBV06, Sch00},  nearest-neighbor search~\cite{IM98, AC06}, and compressed sensing~\cite{BDDW08}.

In this paper, we use random projections that satisfy~\ldefref{genreg} and~\ldefref{affine}. An example distribution $\cD_R$ with $t = O(\alpha^{-2} \log (1/\delta))$ is the distribution of random matrices whose entries are sampled i.i.d. from $\cN(0,1/t)$. 
Recently, Clarkson and Woodruff~\cite{CW13full} proposed other distribution of random matrices that satisfies {\em $(\alpha,\delta)$-subspace embedding for generalized regression} and {\em $(\alpha,\delta)$-affine subspace  embedding}. 
They showed the following: 

\begin{lemma} \label{lem:phi}  \label{lem:CW13} (\cite[Lem 41, Lem 46]{CW13full})
	There is a distribution $\cD_R$  over $\R^{t \times m}$ such that satisfies 
	\begin{description}
		\item [(i)] if $t=O(\alpha^{-2} \log^2 m)$, then for $\bPhi \sim \cD_R$ and any $m \times n$ matrix $\mathbf{Q}$, $  \| \bPhi\mathbf{Q}  \|_F^2 = (1 \pm \alpha) \| \mathbf{Q} \|_F^2,$ and  
		\item [(ii)] if  $t = O(r/\alpha \log(r/\delta))$, then $\cD_R$ satisfies $(\alpha,\delta)$-subspace embedding for generalized regression for $\mathbf{P}$ and $\mathbf{Q}$. 
 	\end{description}
	Further, for any matrix  $\mathbf{L} \in \R^{m \times n}$, $\bPhi \mathbf{L}$ can be computed in $O(\mathsf{nn}(\mathbf{L}) + tn \log t )$ time. Here $\mathbf{P}$, $\mathbf{Q}$, and $r$ 
 are as in~\ldefref{genreg}.	
\end{lemma}

\begin{lemma} \label{lem:S} {\cite[Thm 39, Thm 42]{CW13full})}
	There exists a distribution $\mathcal{D}_A$  over $\R^{v \times m}$ such that 
	\begin{description}
		\item [(i)]  if $v=\Theta(\alpha^{-2})$, then for $\bS \sim \cD_A$ and any $m \times d$ matrix $\mathbf{D}$, $\| \bS \mathbf{D}  \|_F^2 = (1 \pm \alpha) \| \mathbf{D} \|_F^2$. 
		\item [(ii)]  if $v = O(p/\alpha^2 \log(p/\delta))$, then $\cD_A$  satisfies $(\alpha,\delta)$-affine embedding for $\mathbf{D}$ and $\mathbf{E}$.
	\end{description} 
	Further, for any matrix  $\mathbf{L} \in \R^{m \times n}$,  $\mathbf{S}\mathbf{L}$  can be computed in  $O(\mathsf{nn}(\mathbf{L}) + nv \log v)$ time. Here $\mathbf{E}$, $\mathbf{D}$, and $p$ 
 are as in~\ldefref{affine}.
\end{lemma}

In the theorems above, $\bPhi$ and $\bS$ are oblivious to the matrices $\mathbf{P},\mathbf{Q},\mathbf{D},$ and $\mathbf{E}$. That is, we design the distribution $\cD_A$ over linear maps such that for any fixed matrices $\mathbf{D},\mathbf{E}$, if we chose $\bS \sim \cD_A$, then $\bS$ is an $(\alpha,\beta)$-affine embedding for $\mathbf{D},\mathbf{E}$. Similarly, we design the distribution $\cD_R$ over linear maps such that for any fixed matrices $\mathbf{P},\mathbf{Q}$, if we chose $\bPhi \sim \cD_A$, then $\bPhi$ is an $(\alpha,\beta)$ embedding for $\mathbf{P},\mathbf{Q}$.

\subsection{Differential privacy}
Differential privacy is a very robust guarantee of privacy which makes confidential data available widely for accurate analysis while still preserving the privacy of individual data. Achieving these two requirements at the same time seems paradoxical. On one hand, we do not wish to leak information about an individual. On the other hand, we want to answer the query on the entire database as accurately as possible. This makes designing differentially private mechanisms challenging. 

\subsubsection{Robustness of Differential Privacy}
One of the key features of differential privacy is that it is preserved under arbitrary post-processing, i.e., an analyst, without additional information about the private database, cannot compute a function that makes an output less differentially private. In other words,
\begin{lemma} \label{lem:post}  {\em (Dwork {\it et al.}~\cite{DKMMN06}).}
Let $\mathfrak{M}(\mathbf{D})$ be an $(\alpha, \beta)$-differential private mechanism for a database $\mathbf{D}$ , and let $h$ be any function, then any mechanism $\mathfrak{M}':=h(\mathfrak{M}(\mathbf{D}))$ is also $(\alpha,\beta)$-differentially private for the same set of tasks.
\end{lemma}
\begin{proof}
Let $\mathfrak{M}$ be a differentially private mechanism. Let $\range(\mathfrak{M})$ denote the range the of $\mathfrak{M}$.Let $R$ be the range of the function $h(\cdot)$. Without loss of generality, we assume that $h(\cdot): \range(\mathfrak{M}) \rightarrow \cR$ is a deterministic function. This is because any randomized function can be decomposed into a convex combination of deterministic function, and a convex combination of differentially private mechanisms is differentially private. Fix any pair of neighbouring data-sets $\mathbf{DB}$ and $\widetilde{\mathbf{DB}}$ and an event $S \subseteq \cR$. Let $T=\set{y \in \range(\mathfrak{M}): f(r) \in S }$. Then
	\begin{align*}
		\p [f(\mathfrak{M}(\mathbf{DB})) \in S] &= \p [\mathfrak{M}(\mathbf{DB}) \in T] \\
					&\leq \exp(\alpha) \p [\mathfrak{M}(\widetilde{\mathbf{DB}}) \in T] + \beta \\
					&= \exp(\alpha) \p [f(\mathfrak{M}(\widetilde{\mathbf{DB}})) \in S] + \beta.
\end{align*}
\end{proof}

\subsubsection{Composition} 
Before we begin, we discuss what does it mean by the term ``composition" of differentially private mechanism. The composition that we consider covers the following two cases:
\begin{enumerate}
	\item Repeated use of differentially private mechanism on the same database.
	\item Repeated use of differentially private mechanism on different database that might contain information relating to a particular individual.
\end{enumerate} 
The first case covers the case when we wish to use the same mechanism multiple times while the second case covers the case of cumulative loss of privacy of a single individual whose data might be spread across many databases.

It is easy to see that the composition of pure differentially private mechanisms yields another pure differentially private mechanism, i.e., composition of an $(\alpha_1,0)$-differentially private and an $(\alpha_2,0)$-differentially private mechanism results in an $(\alpha_1+\alpha_2,0)$-differentially private mechanism. In other words, the privacy guarantee depreciates linearly with the number of compositions. In the case of approximate differential privacy, we can improve on the degradation of $\alpha$ parameter at the cost of slight depreciation of the $\beta$ factor. We use this strengthening in our proofs. In our proofs of differential privacy, we prove that each row of the published matrix preserves $(\alpha_0, \beta_0)$-differential privacy for some appropriate $\alpha_0,\beta_0$, and then invoke a composition theorem  by Dwork, Rothblum, and Vadhan~\cite{DRV10} to prove that the published matrix preserves $(\alpha,\beta)$-differential privacy. The following theorem is the composition theorem that we use.
\begin{theorem} \label{thm:DRV10}   {\em (Dwork {\it et al.}~\cite{DRV10}).}
	Let $\alpha_0, \beta_0 \in (0,1)$, and $\beta'>0$. If $\mathfrak{M}_1, \cdots , \mathfrak{M}_\ell$ are each $(\alpha, \beta)$-differential private mechanism, then the mechanism $\mathfrak{M}(\mathbf{D}):= (\mathfrak{M}_1(\mathbf{D}), \cdots , \mathfrak{M}_\ell(\mathbf{D}))$ releasing the concatenation of each algorithm is $(\alpha', \ell \beta+0+\beta')$-differentially private for $\alpha' < \sqrt{2\ell \ln (1/\beta')}\alpha_0 + 2\ell \alpha_0^2$.
\end{theorem}
A proof of this theorem could be found in~\cite[Chapter 3]{DR14}.

\paragraph{Gaussian Mechanism.} \label{gaussianmechanism} The Gaussian variant of the Laplace mechanism was proven to preserve differential privacy by Dwork {\it et al.}~\cite{DKMMN06} in a follow-up work. Let $f(\cdot)$ be a function from a class of $\Delta$-sensitive functions. The Gaussian mechanism is 
\[\mathfrak{M}(\mathbf{D}, f(\cdot), \alpha) := f(\mathbf{D}) + (X_1, \cdots, X_k), \text{ where } X_i \sim \cN\paren{0,\frac{\Delta^2}{ \eps^{2}} \log(1.25/\delta)}. \]
Dwork {\it et al.}~\cite{DKMMN06} proved the following.

\begin{theorem} (Gaussian mechanism~\cite{DKMMN06}.) \label{thm:gaussian}
Let $\bx, \by \in \R^n$ be any two vectors such that $\| \bx - \by \|_2 \leq c$. Let $\rho = c\varepsilon^{-1} \sqrt{\log(1/\delta)}$ and $\mathbf{g} \sim \cN(0,\rho^2)^n$ be a vector with each entries sampled i.i.d. Then for any $\mathbf{s} \subset \R^n$, $  \p [\bx + \mathbf{g} \in \mathbf{s}] \leq e^{\varepsilon} \p [\by + \mathbf{g} \in \mathbf{s}] + \delta. $
\end{theorem}

\subsection{Properties of Gaussian distribution.} We need the following property of a random Gaussian matrices. 

\begin{fact} \label{fact:gaussian} (\cite{JL84,Sarlos06})
Let $\mathbf{P} \in \R^{m \times n}$ be a matrix of rank $r$ and $\mathbf{Q} \in \R^{m \times n'}$ be an $m \times n'$ matrix. Let $\cD$ be a distribution of matrices over $\R^{t \times n}$ with entries  sampled i.i.d. from $\cN(0,1/t)$. Then there exists a $t = O(r/\alpha \log(r/\beta) )$ such that $\cD$ is an $(\alpha,\beta)$-subspace embedding for generalized regression.
\end{fact}

\begin{lemma} \label{lem:N}
Let $\bN \sim \cN(0,\rho^2)^{m \times n}$ Gaussian matrix. Then with probability $99/100$, $\| \bN \|_F = O( \rho \sqrt{mn})$.
\end{lemma}
\begin{proof}
The lemma follows from the following computation.
 	\[
	\E[ \| \bC \bPhi  \|_F^2] = \E \sparen{ \sum_{i,j}  (\widetilde{\bN}_1)_{ij}^2} =  \sum_{i,j} \E[  (\widetilde{\bN}_1)_{ij}^2] = mn \rho^2.
	\]
The result follows using Markov's inequality. 
 \end{proof}

\subsection{Properties of pseudo-inverse of a matrix.} We need the following results about product of pseudo-inverse in the proof of~\lemref{SRHTinverse} and~\lemref{efficiency}.
\begin{fact}
If $\bA$ has a left-inverse, then $\bA^\dagger = (\bA^{\mathsf T} \bA)^{-1} \bA^{\mathsf T}$ and if $\bA$ has right-inverse, then $\bA^\dagger =  \bA^{\mathsf T} (\bA \bA^{\mathsf T})^{-1}$. 
\end{fact}

\begin{theorem} \label{thm:pseudoinverse}
Let $\bA$ and $\bB$ be conforming matrices and either,
\begin{enumerate}
	\item $\bA$ has orthonormal columns (i.e., $\bA^{\mathsf T} \bA$ is an identity matrix) or,
	\item $\bB$ has orthonormal rows (i.e., $\bB \bB{^\mathsf T}$ is an identity matrix),
	\item $\bA $ has all columns linearly independent (full column rank) and $\bB$ has all rows linearly independent (full row rank) or,
	\item $\bB = \bA^{\mathsf T}$ (i.e., $\bB$ is the conjugate transpose of $\bA$),
\end{enumerate}
then $(\bA \bB)^\dagger = \bB^\dagger \bA^\dagger$.
\end{theorem}

We use the following variant of Pythagorean theorem in the proof of~\lemref{orthonormal}.
\begin{theorem} \label{thm:phyth} (Pythagorean theorem). 
	Let $\bA$ and $\bB$ be two matrices such that $\bA^{\mathsf T} \bB$ is an zero matrix. Then for any $\bC = \bA + \bB$, we have $\| \bC \|_F^2 = \| \bA \|_F^2 + \| \bB \|_F^2.$
\end{theorem}

\subsection{Linear Algebraic Results Used in This Paper}
We also need the following results for the privacy proof.
\begin{theorem} \label{thm:lidskii} (Lidskii Theorem~\cite{Bhatia}). 
Let $\bA, \mathbf{B}$ be $n \times n$ Hermittian matrices. Then for any choice of indices $1 \leq i_1 \leq \cdots \leq i_k \leq n$,  
\[  \sum_{j=1}^k  \lambda_{i_j}(\bA + \mathbf{B}) \leq \sum_{j=1}^k  \lambda_{i_j}(\bA) + \sum_{j=1}^k  \lambda_{i_j}( \mathbf{B}), \]
where $\set{\lambda_i(\bA)}_{i=1}^n$ are the eigen-values of $\bA$ in decreasing order.
\end{theorem}

\begin{theorem} \label{thm:Weyl} (Weyl's Pertubation Theorem~\cite{Bhatia}). 
For any $m \times n$ matrices $\mathbf{P}, \mathbf{Q}$, we have $| \sigma_i(\mathbf{P} + \mathbf{Q}) - \sigma_i(\mathbf{P})| \leq \| \mathbf{Q} \|_2$, where $\sigma_i(\cdot)$ denotes the $i$-th singular value and $\| \mathbf{Q} \|_2$ is the spectral norm of the matrix $\mathbf{Q}$.
\end{theorem}

We use the notation $\rad(p)$ to denote a distribution with support $\pm 1$ such that $+1$ is sampled with probability  $p$ and $-1$ is sampled with probability  $1-p$. 
An $n \times n$ Walsh-Hadamard matrix $\mathbf{H}_n$ is constructed recursively as follows:
$$  \mathbf{H}_n = \begin{pmatrix}  \mathbf{H}_{n/2} & \mathbf{H}_{n/2} \\  \mathbf{H}_{n/2} & -\mathbf{H}_{n/2} \end{pmatrix}~\text{and}~\mathbf{H}_1 :=1.$$
A randomized Walsh-Hadamard matrix $\bW_n$ is formed by multiplying $\mathbf{H}_n$ with a diagonal matrix whose diagonal entries are picked i.i.d. from $\rad(1/2).$ We drop the subscript $n$ where it is clear from the context. A subsampled randomized Hadamard matrix is construct by multiplying $\bPi_{1..r}$ from the left to a randomized Hadamard matrix, where $\bPi_{1..r}$ is the matrix formed by the first $r$ rows of a random permutation matrix. 

\begin{lemma} \label{lem:SRHTinverse}
Let $\bS$ be a $v \times m$ subsampled randomized Hadamard matrix, where $v\leq m$ and $\bN \in \R^{v \times n}$. Then we have, 
\[ \|  \bS^\dagger \bN_2 \|_F = \| \bN_2 \|_F.\]  
\end{lemma} 
 \begin{proof}
One way to look at the action of $\bS$ when it is a subsampled Hadamard transform is that it is a product of matrices $\bW$ and $\bPi_{1..r}$, where $\bPi_{1..r}$ is the matrix formed by the first $r$ rows of a random permutation matrix and $\bW$ is a randomized Walsh-Hadamard matrix formed by multiplying a Walsh-Hadamard matrix with a diagonal matrix whose non-zero entries are picked i.i.d. from $\rad(1/2)$. 

 Since $\bW\bD$ has orthonormal rows, $\bS^\dagger=(\bPi_{1..v} \bW \bD)^\dagger= (\bW \bD)^{\mathsf T} (\bPi_{1..v})^{\dagger}$.  
 This implies
\begin{align*}
 \|  \bS^\dagger \bN \|_F &= \| ({\bPi}_{1..v} \bW \bD)^\dagger \bN \|_F = \| (\bW \bD)^{\mathsf T} {\bPi}_{1..v}^{\dagger} \bN \|_F \\
 & = \| {\bPi}_{1..v}^{\dagger} \bN \|_F. \end{align*}

Using the fact that $\bPi_{1..v}$ is a full row rank matrix and $\widehat{\bPi}_{1..v}  \widehat{\bPi}_{1..v}^{\mathsf T}$ is an identity matrix, we have 
$\widehat{\bPi}_{1..v}^{\dagger} = \widehat{\bPi}_{1..v}^{\mathsf T} ( \widehat{\bPi}_{1..v}  \bPi_{1..v}^{\mathsf T})^{-1} =  \widehat{\bPi}_{1..v}^{\mathsf T}. $ The result follows. 
 \end{proof}

We reprove the following theorem of Boutsidis {\it et al.}~\cite{BWZ16}. Our proof allows us a tighter control on the intermediate results which helps us to get a tighter bound on the additive error. 
\begin{theorem} \label{thm:BWZ16}
Let $\cD_R$  be an $(\alpha,\delta)$-subspace embedding for generalized regression (\textcolor{red}{Definition}~\ref{def:genreg}). 
Then with probability $1-2\delta$ over $\bPhi^{\mathsf T} \sim \cD_R$ and $\bPsi \sim \cD_R$, for any arbitrary $m \times n$ matrix $\bA$,
	\begin{align} \min_{\bX, \mathsf{r}(\bX) \leq k} \| \bA \bPhi \bX \bPsi \bA  - \bA \|_F \leq (1+\alpha)^2 \| \bA - [\bA]_k \|_F. \label{eq:step14} \end{align}
\end{theorem}	
Our proof uses two optimization problems and uses the solution to those optimization problem in a clever way. We feel that our proof is simpler. It also has explicit solutions to the two optimization problems, which makes it easy to extend to the case of private low-rank factorization and get a tight bound.	
\begin{proof}
Let $[\bA]_k=\bU_k \bSigma_k \bV_k^{\mathsf T}$. 
We will use~\lemref{phi} to prove the theorem. Set $\bPhi= \bPhi^{\mathsf T}$, $\mathbf{P}=[\bA]_k^{\mathsf T}$, $\mathbf{Q}=\bA^{\mathsf T}$. Then for $\widetilde{\bX}=\argmin_{\bX} \|\bPhi^{\mathsf T} ([\bA]_k^{\mathsf T} \bX - \bA^{\mathsf T}) \|_F$, we have with probability $1-\delta$ over $\bPhi^{\mathsf T} \sim \cD_R$,
  \begin{align*} \| [\bA]_k^{\mathsf T} \widetilde{\bX} - \bA^{\mathsf T} \|_F &\leq (1+\alpha) \min_{\bX} \| [\bA]_k^{\mathsf T} \bX - \bA^{\mathsf T} \|_F \\
  		&\leq (1+\alpha)  \| [\bA]_k^{\mathsf T} - \bA^{\mathsf T} \|_F 
  \end{align*}
where the last inequality follows by setting $\bX = \bU_k \bU_k^\mathsf T$. 
 Here $\widetilde{\bX} = (\bPhi^{\mathsf T} [\bA]_k^{\mathsf T})^\dagger (\bA \bPhi)^{\mathsf T}$. Since Frobenius norm is preserved under transpose, we have by substituting the value of $\widetilde \bX$,
 \begin{align}
  \| \bA \bPhi ([\bA]_k \bPhi)^\dagger [\bA]_k - \bA \|_F \leq (1+\alpha) \| \bA - [\bA]_k \|_F. \label{eq:phipsi1}
 \end{align}
 
We now use~\lemref{phi} on the following regression problem: 
  \[ \min_{\bX} \|  \mathbf{W} \bX - \bA \|_F, \quad \text{where}\quad \mathbf{W} = \bA \bPhi ([\bA]_k \bPhi)^{\dagger} .\]

  Let $\widehat{\bX} = \argmin_{\bX} \| \bPsi (\mathbf{W} \bX - \bA) \|_F.$ Since $[\bA]_k$ has rank $k$,~\lemref{phi} and~\eqnref{phipsi1} gives with probability $1-\delta$ over $\bPsi \sim \cD_R$ 
  \begin{align*}
  	\|  \mathbf{W} \widehat{\bX} - \bA \|_F &= \|  \bA \bPhi ([\bA]_k \bPhi)^{\dagger} \widehat{\bX} - \bA \|_F \\ 
		&\leq  (1+\alpha)  \min_{\bX} \| (\bA \bPhi) ([\bA]_k \bPhi)^{\dagger} \bX - \bA \|_F \\
		&\leq (1+\alpha)  \| \bA \bPhi ([\bA]_k \bPhi)^{\dagger} [\bA]_k - \bA \|_F  \\
		&\leq (1+\alpha)^2 \| \bA - [\bA]_k \|_F.
  \end{align*}

 {Substituting the value of $\widehat{\bX} = (\bPsi \mathbf{W})^\dagger \bPsi \bA$, with probability $1-2\delta$ over $\bPhi^{\mathsf T}, \bPsi \sim \cD_R$, we have}
  \begin{align}
  &\| \bA \bPhi ([\bA]_k \bPhi)^{\dagger} (\bPsi \bA \bPhi ([\bA]_k \bPhi)^{\dagger})^\dagger \bPsi \bA - \bA \|_F 
    \quad \leq (1+\alpha)^2 \| \bA - [\bA]_k \|_F. \label{eq:BWZ16last}
  \end{align}
 
Since $ ([\bA]_k \bPhi)^{\dagger} (\bPsi \bA \bPhi ([\bA]_k \bPhi)^{\dagger})^\dagger$ has rank at most $k$, this completes the proof because $\bPhi ([\bA]_k \bPhi)^{\dagger})^\dagger \bPsi$ is a rank-$k$ matrix.
\end{proof}

We prove the  following key lemma, which can be seen as a generalization of one of the previous results of Clarkson and Woodruff~\cite{CW09}. This lemma would be required in proving all our results.
\begin{lemma} \label{lem:orthonormal}
Let $\mathbf{R}$ be a matrix with orthonormal rows and $\bC$ have orthonormal columns. Then 
\[ \min_{\bX, \mathsf{r}(\bX)=k} \|\bC \bX \mathbf{R} - \mathbf{F}  \|_F = \| \bC[\bC^{\mathsf T} \mathbf{F} \mathbf{R}^{\mathsf T}]_k \mathbf{R} - \mathbf{F} \|_F. \]  
 \end{lemma}
 \begin{proof}
For any matrix $\bY$ of appropriate dimension, we have $\brak{\mathbf{F}-\bC \bC^{\mathsf T} \mathbf{F}, \bC \bC^{\mathsf T} \mathbf{F} - \bC \bY \mathbf{R}}=0$. This is because $\mathbf{F}-\bC \bC^{\mathsf T} \mathbf{F} = ( \I -\bC \bC^{\mathsf T}) \mathbf{F}$ lies in space orthogonal to $\bC (\bC^{\mathsf T} \mathbf{F} -  \bY \mathbf{R})$. By~\thmref{phyth}, 
 \begin{align}
 \| \mathbf{F} - \bC \bY \mathbf{R} \|_F^2 
 	&= \| \mathbf{F} - \bC \bC^{\mathsf T} \mathbf{F} \|_F^2 + \| \bC \bC^{\mathsf T} \mathbf{F} - \bC \bY \mathbf{R} \|_F^2   \nonumber \\
  	&= \| \mathbf{F} - \bC \bC^{\mathsf T} \mathbf{F} \|_F^2 + \|  \bC^{\mathsf T} \mathbf{F} -  \bY \mathbf{R} \|_F^2, \label{eq:orthonormal1}
\end{align}	 
where the second equality follows from the properties of unitary matrices. 

{Again, for any matrix $\bY$ of appropriate dimensions, we have $\brak{\bC^{\mathsf T} \mathbf{F} \mathbf{R}^{\mathsf T} \mathbf{R} - \bY \mathbf{R}, \bC^{\mathsf T} \mathbf{F} -  \bC^{\mathsf T} \mathbf{F} \mathbf{R}^{\mathsf T} \mathbf{R}  } =0$. This is because $\bC^{\mathsf T} \mathbf{F} \mathbf{R}^{\mathsf T} \mathbf{R} - \bY \mathbf{R}= (\bC^{\mathsf T} \mathbf{F} \mathbf{R}^{\mathsf T}  - \bY) \mathbf{R}$ lies in the space spanned by $\mathbf{R}$, and $\bC^{\mathsf T} \mathbf{F} -  \bC^{\mathsf T} \mathbf{F} \mathbf{R}^{\mathsf T} \mathbf{R} = \bC^{\mathsf T} \mathbf{F} (\I  - \mathbf{R}^{\mathsf T} \mathbf{R})$ lies in the orthogonal space. By~\thmref{phyth}, we have }
\begin{align}
 \|  \bC^{\mathsf T} \mathbf{F} -  \bY \mathbf{R} \|_F^2
    	&=  \| \bC^{\mathsf T} \mathbf{F} - \bC^{\mathsf T} \mathbf{F} \mathbf{R}^{\mathsf T} \mathbf{R} \|_F^2 
	+ \| \bC^{\mathsf T} \mathbf{F} \mathbf{R}^{\mathsf T} \mathbf{R} -\bY \mathbf{R} \|_F^2 \label{eq:orthonormal2}
\end{align}	 

{Since $\| \bC^{\mathsf T} \mathbf{F} - \bC^{\mathsf T} \mathbf{F} \mathbf{R}^{\mathsf T} \mathbf{R} \|_F^2 $ is independent of $\bY$, we just bound the term $\| \bC^{\mathsf T} \mathbf{F} \mathbf{R}^{\mathsf T} \mathbf{R} -\bY \mathbf{R} \|_F^2$. Substituting $\bY = [ \bC \mathbf{F} \mathbf{R}]_k$ and using the fact that multiplying $\mathbf{R}$ from the right does not change the Frobenius norm and $[\bC^{\mathsf T} \mathbf{F} \mathbf{R}^{\mathsf T}]_k $ is the best $k$-rank approximation to the matrix $\bC^{\mathsf T} \mathbf{F} \mathbf{R}^{\mathsf T}$, for all rank-k matrices $\bZ$, we have } 
 \begin{align}
\| \bC^{\mathsf T} \mathbf{F} \mathbf{R}^{\mathsf T} \mathbf{R} -[\bC^{\mathsf T} \mathbf{F} \mathbf{R}^{\mathsf T}]_k \mathbf{R} \|_F^2 &\leq \| \bC^{\mathsf T} \mathbf{F} \mathbf{R}^{\mathsf T} \mathbf{R} -\bZ \mathbf{R} \|_F^2. \label{eq:orthonormal3}
\end{align}
{Combining~\eqnref{orthonormal3} with~\eqnref{orthonormal2} and~\thmref{phyth}, we have}
 \begin{align}
\|  \bC^{\mathsf T} \mathbf{F} -   [ \bC \mathbf{F} \mathbf{R}]_k \mathbf{R} \|_F^2 & \leq \| \bC^{\mathsf T} \mathbf{F} - \bC^{\mathsf T} \mathbf{F} \mathbf{R}^{\mathsf T} \mathbf{R} \|_F^2  + \| \bC^{\mathsf T} \mathbf{F} \mathbf{R}^{\mathsf T} \mathbf{R} -   \bZ \mathbf{R} \|_F^2  = \| \bC^{\mathsf T} \mathbf{F} -  \bZ \mathbf{R} \|_F^2. \label{eq:orthonormal4}
\end{align}
{Combining~\eqnref{orthonormal4} with~\eqnref{orthonormal1}, the fact that $\bC$ has orthonormal columns, and~\thmref{phyth}, we have}
\begin{align} \| \mathbf{F} - \bC [ \bC \mathbf{F} \mathbf{R}]_k \mathbf{R} \|_F^2 
 	&\leq \| \mathbf{F} - \bC \bC^{\mathsf T} \mathbf{F} \|_F^2 + \| \bC^{\mathsf T} \mathbf{F} -  \bZ \mathbf{R} \|_F^2 \nonumber \\
 	&= \| \mathbf{F} - \bC \bC^{\mathsf T} \mathbf{F} \|_F^2 + \| \bC \bC^{\mathsf T} \mathbf{F} - \bC \bZ \mathbf{R} \|_F^2 \nonumber \\
 	&= \| \mathbf{F}  - \bC \bZ \mathbf{R} \|_F^2. \nonumber 
 \end{align}
 This completes the proof of~\lemref{orthonormal}.
 \end{proof}


\section{Low Space Differentially private  Low-rank Factorization} \label{app:spaceprivate}
In this section, we give our basic low space algorithms for various granularity of privacy. These algorithms serve as a meta algorithm on which we built algorithms in various model of computations, like streaming model under turnstile update, continual release model, and local model. In~\appref{lowspaceprivate}, we analyze the algorithm presented earlier in~\secref{meta}. In~\appref{meta_time}, we give our low space differentially private algorithm under $\priv_2$, a stronger privacy guarantee. 

\subsection{Low Space Differentially Private Low-rank Factorization Under $\priv_1$} \label{app:lowspaceprivate}
In this section, we analyze the algorithm presented earlier in~\secref{meta}. For the ease of the readers, we first restate~\thmref{meta} here.

\medskip \noindent {\bf Restatement of~\thmref{meta}.} 
Let $m,n,k \in \N$ and $\alpha,\varepsilon,\delta$ be the input parameters.  Let $s=\max\{m,n\}, u=\min \set{m,n}$, $\kappa=(1+\alpha)/(1-\alpha)$, $\eta=\max\set{k,\alpha^{-1}}$, and $\sigma_\mathsf{min}={16 \log(1/\delta) \sqrt{t \kappa \ln(1/\delta)}}/{\varepsilon}$. Given an $m \times n$ matrix $\bA$,  {\scshape Private-Space-Optimal-$\lrf$}, described in~\figref{spacelowprivate}, outputs a $k$-rank factorization $\widetilde{\bU}, \widetilde{\bSigma}$, and $ \widetilde{\bV}$, such that 
\begin{enumerate}
	\item {\scshape Private-Space-Optimal-$\lrf$} is $(3\varepsilon,3\delta)$ differentially private under $\priv_1$.  \label{privacy}
	\item  Let $\mathbf{M}_k=\widetilde{\bU} \widetilde{\bSigma} \widetilde{\bV}^{\mathsf T}. $ Then with probability $9/10$ over the random coins of  {\scshape Private-Space-Optimal-$\lrf$},  \label{privatecorrectness}
 \begin{align*} 
 &\| \begin{pmatrix} \bA & \mathbf{0} \end{pmatrix} - \mathbf{M}_k \|_F \leq (1+\alpha) \| \bA - [\bA]_k \|_F  + O ( \sigma_\mathsf{min} \sqrt{u} + \varepsilon^{-1}  {\sqrt{ks  \ln(1/\delta)}} ),
\end{align*}
 	\item The space used by {\scshape Private-Space-Optimal-$\lrf$} is $O((m+n)\eta \alpha^{-1}\log(k/\delta))$.  \label{privatespace}
	\item The  total computational time is 
	$O\paren{  \mathsf{nn}(\bA) \log(1/\delta) +  \frac{(u^2+s\eta)\eta \log^2(k/\delta) }{\alpha^{2}}  + \frac{\eta^3  \log^3 (k/\delta)}{\alpha^{5}} }.$ \label{privatetime}
\end{enumerate}

First note that if we do not aim for run time efficiency, then we can simply use Gaussian random matrices instead of sampling $\bPhi, \Psi \sim \cD_R$ and $\bS,\bT \sim \cD_A$ as per~\lemref{phi} and~\lemref{S}. This would simplify the privacy proof as we will see later. Secondly, the probability of success can be amplified to get a high probability bound by standard techniques. We leave these details as they are standard arguments.
\begin{proof} [Proof  of~\thmref{meta}]
Part~\ref{privatespace} follows immediately by setting the values of $t$ and $v$.
Part~\ref{privatetime} of~\thmref{meta} requires some computation.
More precisely, we have the following.
\begin{enumerate}
	 \item 
	Computing $\bY_c$ requires $ O(\mathsf{nn}(\bA) \log(1/\delta))+ m(m+n)t$ time and computing $\bY_r$ requires $O(\mathsf{nn}(\bA)\log(1/\delta)) + (m+n)t^2$ time. 
	 \item 
	 Computing 
	 $\bU$ and $\bV$ requires $O(nt^2 + mt^2) = 
	 O((m+n)\eta^2 \alpha^{-2} \log^2( k/\delta))$ time.
	 \item 
	 Computing a SVD of  matrices $\bS \bU$  and $\mathbf{T} \bV^{\mathsf T}$ requires $vt^2 + tv^2=O(k^3 \alpha^{-5} \log^2( k/\delta))$.
	 \item 
	 Computing $\bZ$ requires $O(\mathsf{nn}(\bA) \log( 1/\delta)  +  \eta^2/\alpha^2 \log^2( k/\delta) + n\eta^2\alpha^{-2} \log^2( k/\delta))$
	 Computation of $[\widetilde{\bU}_s^{\mathsf T}\bZ \widetilde{\bV}]_k$ requires $O(\mathsf{nn}(\bA) \log (1/\delta)) + tv^2= O(\mathsf{nn}(\bA) \log(1/\delta) + k^{3} \alpha^{-5} \log^3 ( k/\delta))$ time.
	 \item 
	 Computation of the last 
	 SVD 
	 requires 
	 $O((m+n) \eta^2 \alpha^{-2} \log^2( k/\delta))$ time.
\end{enumerate}
Combining all these terms, we have our claim on the running time. 
\subsubsection{Correctness Proof of {\scshape Private-Space-Optimal-$\lrf$}} \label{app:correctness_space}
We now prove the correctness guarantee of {\scshape Private-Space-Optimal-$\lrf$}. In what follows, we analyze the case when $m \leq n$. The case when $n \leq m$ follows analogously due to the symmetry of {\scshape Private-Space-Optimal-$\lrf$}.  First note that appending $\bA$ with an all zero matrix $\mathbf{0}^{m \times m}$ has no effect on its $k$-rank approximation, i.e., we can analyze the approximation guarantee for $\widehat \bA := \begin{pmatrix} \bA & \mathbf{0} \end{pmatrix}$ instead of $\bA$. 
Let $\mathbf{M}_k$ be as defined in~\thmref{meta}. We break our proof in three main steps.
\begin{description}
	\item [(i)] Lower bound $ \| \mathbf{M}_k - \widehat{\bA} \|_F$ by $\| \mathbf{M}_k - \begin{pmatrix} \bA & \mathbf{0} \end{pmatrix} \|_F$ up to an additive term (\lemref{first}). \label{item:lower}
	\item [(ii)] Relate  $ \| \mathbf{M}_k - \widehat{\bA} \|_F$ and $ \| \widehat{\bA} - [\widehat{\bA}]_k \|_F$ (\lemref{third}). \label{item:relate}
	\item [(iii)] Upper bound $ \| \widehat{\bA} - [\widehat{\bA}]_k \|_F$ by a term linear in $\| \bA - [\bA]_k \|_F$ up to an additive term (\lemref{second}). \label{item:upper} 
\end{description}
Part~\ref{privatecorrectness} of~\thmref{meta} follows by combining these three items together. 

\medskip \noindent \textbf{Performing step~(i).} We start by proving a bound on  $ \| \mathbf{M}_k - \widehat{\bA} \|_F$ by $\| \mathbf{M}_k - \bA \|_F$ and  a small additive term. The following lemma provides such a bound.
\begin{lemma} \label{lem:first}
 Let $\bA$ be an $m \times n$ input matrix, and let $\widehat{\bA}=  \begin{pmatrix} \bA & \sigma_\mathsf{min} \I_m \end{pmatrix}$ for $\sigma_\mathsf{min}$ defined in~\thmref{meta}. Denote by $\mathbf{M}_k:= \widetilde{\bU}  \widetilde{\bSigma} \widetilde{\bV}^{\mathsf T} $ the output of {\scshape Private-Optimal-Space}-$\lrf$. Then 
$$ \| \mathbf{M}_k - \begin{pmatrix} \bA & \mathbf{0} \end{pmatrix}\|_F  \leq  \| \mathbf{M}_k - \widehat{\bA} \|_F + \sigma_{\min} \sqrt{m}.$$
\end{lemma}
\begin{proof}
The lemma is immediate from the following observation:
$ \| \mathbf{M}_k - \begin{pmatrix} \bA & \mathbf{0} \end{pmatrix} \|_F - \sigma_{\min} \| \I_m  \|_F   \leq  \| \mathbf{M}_k - \begin{pmatrix} \bA &  \mathbf{0} \end{pmatrix} - \begin{pmatrix}  \mathbf{0} &  \sigma_{\min} \I \end{pmatrix} \|_F =  \| \mathbf{M}_k - \widehat{\bA} \|_F, 
$ where the first inequality follows from the sub-additivity of the Frobenius norm.
\end{proof}


\medskip \noindent \textbf{Performing step~(ii).} This is the most involved part of the proof and uses multiple lemmas as follows.
\begin{lemma} \label{lem:third} 
 Let $\widehat{\bA}=  \begin{pmatrix} \bA & \sigma_{\min} \I \end{pmatrix}$  and denote by $\mathbf{M}_k:= \widetilde{\bU}  \widetilde{\bSigma} \widetilde{\bV}^{\mathsf T} $. Let $\widehat{\bPhi} = t^{-1} \bOmega \bPhi$. Then  with probability $1-O(\delta)$ over the random coins of the algorithm {\scshape Private-Space-Optimal-$\lrf$}, 
  \begin{align*}
 \|\mathbf{M}_k - \widehat{\bA} \|_F  &\leq (1+\alpha) \|\widehat{\bA} - [\widehat{\bA}]_k\|_F + 2 \|  \bS^\dagger \bN (\bT^{\mathsf T})^\dagger \|_F +  \| { \widehat{\bA} \widehat{\bPhi}  ([\widehat{\bA}]_k \widehat{\bPhi})^{\dagger} (\bPsi \widehat{\bA} \widehat{\bPhi} ([\widehat{\bA}]_k \widehat{\bPhi})^{\dagger})^\dagger \bN_1}\|_F .
   \end{align*}
\end{lemma}
\begin{proof}
Let $\bB = \widehat{\bA} + \bS^\dagger \bN (\bT^\dagger)^{\mathsf T}$ and $\widehat{\bPhi} =  \bOmega \bPhi $. 
	We first use the relation between $\min_{\bX, \mathsf{r}(\bX) \leq k} \| \widehat{\bA} \widehat{\bPhi} \bX \bPsi \widehat{\bA}  - \widehat{\bA} \| $ and $(1+\alpha) \| \widehat{\bA} - [\widehat{\bA}]_k \|_F$ from the proof of~\thmref{BWZ16}.
	Using~\lemref{phi} and ~\factref{gaussian}, if we set $\bA = \widehat{\bA}$ in~\eqnref{BWZ16last}, then with probability $1-3\delta$ over $\widehat{\bPhi}^{\mathsf T},\bPsi \sim \cD_R$, 
	$   \| \widehat{\bA} \widehat{\bPhi} ([\widehat{\bA}]_k \widehat{\bPhi})^{\dagger} (\bPsi \widehat{\bA} \widehat{\bPhi} ([\widehat{\bA}]_k \widehat{\bPhi})^{\dagger})^\dagger \bPsi \widehat{\bA} - \widehat{\bA} \|_F \leq (1+\alpha)^2 \| \widehat{\bA} - [\widehat{\bA}]_k \|_F. $
Now define a rank-$k$ matrix $\mathbf{P}_k := ([\widehat{\bA}]_k \widehat{\bPhi})^{\dagger} (\bPsi \widehat{\bA} \widehat{\bPhi} ([\widehat{\bA}]_k \widehat{\bPhi})^{\dagger})^\dagger$. 
Let us consider the following optimization problem: 
$$ \min_{\bX \atop \mathsf{r}(\bX) \leq k} \|\bY_c \bX \bY_r - \bB \|_F.$$
Since $\mathbf{P}_k$ is a rank-$k$ matrix, using the subadditivity of the Frobenius norm, we have
\begin{align}
  \min_{\bX, \atop \mathsf{r}(\bX) \leq k}  \|\bY_c \bX \bY_r - \bB \|_F & \leq \| \bY_c \mathbf{P}_k\bY_r - \bB \|_F \nonumber \\
  &   = \| \bY_c \mathbf{P}_k\bY_r - (\widehat{\bA} + \bS^\dagger \bN (\bT^\dagger)^{\mathsf T}) \|_F  \nonumber \\
   & \leq \| \bY_c \mathbf{P}_k\bY_r - \widehat{\bA} \|_F   + \| \bS^\dagger \bN (\bT^\dagger)^{\mathsf T} \|_F \nonumber \\
  	&= \| \widehat{\bA} \widehat{\bPhi} \mathbf{P}_k (\bPsi \widehat{\bA} + \bN_1) - \widehat{\bA} \|_F  + \| \bS^\dagger \bN (\bT^\dagger)^{\mathsf T} \|_F\nonumber \\
  	&\leq \| \widehat{\bA} \widehat{\bPhi} \mathbf{P}_k \bPsi \widehat{\bA} - \widehat{\bA} \|_F  + \| \bS^\dagger \bN_2 (\bT^\dagger)^{\mathsf T} \|_F
	+ \| { \widehat{\bA} \widehat{\bPhi}  \mathbf{P}_k \bN_1}\|_F\nonumber  \\	
	&\leq (1+\alpha)^2 \| \widehat{\bA} - [\widehat{\bA}]_k \|_F  + \| \bS^\dagger \bN_2 (\bT^\dagger)^{\mathsf T} \|_F  + \| { \widehat{\bA} \widehat{\bPhi}  \mathbf{P}_k \bN_1}\|_F  \label{eq:step14lowspace}
 \end{align}

	Let $\bS_2=\widehat{\bA} \widehat{\bPhi} \mathbf{P}_k \bN_1 $. 
	By definition, $\bV$ is a matrix whose rows are an orthonormal basis for the row space of $\bY_r$  and $\bU$ is a matrix whose columns are an orthonormal basis for the column space of $\bY_c$. Therefore,
	\begin{align} 
		\min_{\bY \atop r(\bY) \leq k} \| \bU \bY \bV - \bB \|_F & \leq \min_{\bX \atop \mathsf{r}(\bX) \leq k} \| \bY_c \bX \bV - \bB \|_F  \leq \min_{\bX \atop \mathsf{r}(\bX) \leq k} \| \bY_c \bX \bY_r - \bB \|_F. \label{eq:basis6}  
	\end{align}
	
	Combining~\eqnref{step14lowspace} and~\eqnref{basis6}, we have 
	\begin{align} 
	 \min_{\bY \atop r(\bY) \leq k} \| \bU \bY \bV &- \bB \|_F \leq (1+\alpha)^2 \| \widehat{\bA} - [\widehat{\bA}]_k \|_F + \| \bS^\dagger \bN_2 (\bT^\dagger)^{\mathsf T} \|_F +  \| \bS_2 \|_F. \label{eq:step26}
	\end{align}

\begin{claim} \label{claim:thirdspace}
Let $\bU, \bV, \bB, \bA, \bS, \bT$ and $\bN_2$ be as above. Let $\cD_A$ be a distribution that satisfies  $(\alpha,\delta)$-affine embedding.  Let $\widetilde{\bX}= \argmin_{\bX, \mathsf{r}(\bX)=k} \| \bS(\bU \bX -\bB) \|_F$. Then  with probability $1-O(\delta)$ over $\bS, \bT^{\mathsf T} \sim \cD_A$,
\begin{align*} \
	\|( \bU \widetilde{\bX} \bV  - \bB) \|_F &\leq {(1+\alpha)^4} \| \widehat{\bA} - [\widehat{\bA}]_k \|_F+ 4 \| \bS^\dagger \bN_2 (\bT^\dagger)^{\mathsf T} \|_F + 4\| \bS_2 \|_F . \label{eq:step36} 
\end{align*}
\end{claim}
\begin{proof}
	Set $p=t$, $\mathbf{D}=\bU$ and $\mathbf{E}=\bB$ in the statement of~\lemref{S}. Let us restrict our attention to matrices $\bX$ with rank at most $k$ and denote by $$\widehat{\bX}= \argmin_{\bX, \mathsf{r}(\bX) \leq k} \| \bU \bX\bV -\bB \|_F \quad \text{and} \quad \widetilde{\bX}= \argmin_{\bX, \mathsf{r}(\bX) \leq k} \| \bS(\bU \bX \bV  -\bB)\bT^{\mathsf T} \|_F.$$ Then we have with probability $1-3\delta$ over $\bS \sim \cD_A$, 
	\begin{align}  
		\min_{\bX \atop \mathsf{r}(\bX)=k} \| \bU \bX \bV -\bB \|_F  = \| \bU \widehat{\bX}  \bV- \bB \|_F    \geq (1 + \alpha)^{-1/2} \| \bS(\bU \widehat{\bX}  \bV - \bB ) \|_F.    
	\end{align}
		{Substituting $\mathbf{D}=\bV^{\mathsf T}$, $\bX = (\bS \bU \widehat{\bX})^{\mathsf T}$ and $\mathbf{E}=(\bS\bB)^{\mathsf T}$ in the statement of~\lemref{S}, with probability $1-4\delta$,}
	\begin{align}  
		 (1 + \alpha)^{-1/2} \| \bS(\bU \widehat{\bX}  \bV - \bB ) \|_F 
		 &  = (1 + \alpha)^{-1/2} \| \bV^{\mathsf T} (\bS\bU \widehat{\bX}) ^{\mathsf T} - (\bS\bB)^{\mathsf T} ) \|_F  \nonumber \\
		 	& \geq (1 + \alpha)^{-1} \| \bT(\bV^{\mathsf T} (\bS\bU \widehat{\bX}) ^{\mathsf T} - (\bS\bB)^{\mathsf T} ) \|_F  \nonumber \\
		 	& = (1 + \alpha)^{-1}  \| \bS(\bU \widehat{\bX}  \bV - \bB ) \bT^{\mathsf T}\|_F \nonumber \\
		 	& \geq (1 + \alpha)^{-1}  \min_{\bX \atop \mathsf{r}(\bX) \leq k} \| \bS(\bU {\bX}  \bV - \bB ) \bT^{\mathsf T}\|_F \nonumber \\
			& = (1 + \alpha)^{-1}  \| \bS(\bU \widetilde{\bX}  \bV - \bB  )  \bT^{\mathsf T} \|_F   \nonumber \\
			&\geq (1 + \alpha)^{-2}  \| (\bU \widetilde{\bX}  \bV  - \bB  ) \|_F.	 		  \label{eq:S6}  
	\end{align}
	
	Combining~\eqnref{S6} with~\eqnref{step26},  with probability $1-O(\delta)$ over the random coins of {\scshape Private-Optimal-Space}-$\lrf$,
	\begin{align}
	\|( \bU \widetilde{\bX} \bV  - \bB) \|_F &\leq {(1+\alpha)^4} \| \widehat{\bA} - [\widehat{\bA}]_k \|_F  + 4 (\| \bS^\dagger \bN_2 (\bT^\dagger)^{\mathsf T} \|_F + \| \bS_2 \|_F)  \label{eq:step36} 
	\end{align}
	as $\alpha \in (0,1)$. This completes the proof of~\claimref{thirdspace}.
\end{proof}

	To finalize the proof, we need to compute $$\widetilde{\bX}=\argmin_{\bX \atop \mathsf{r}(\bX) \leq k} \| \bS(\bU \bX \bV - \bB) \bT^{\mathsf T} \|_F.$$ 
	We use~\lemref{orthonormal} to compute $\widetilde{\bX}.$ 
	Recall $\bS \bU = \bU_s \bSigma_s \bV_s^{\mathsf T}$ and $\bT \bV^{\mathsf T} = \bU_t \bSigma_t \bV_t^{\mathsf T}$. Using~\lemref{orthonormal} with $\mathbf{C} = \widetilde{\bU}_s, \mathbf{R} = \widetilde{\bV}_t^{\mathsf T}$ and $\mathbf{F} = \bZ= \bS \widehat{\bA} \bT^{\mathsf T} + \bN_2$, we  get 
	
	 \begin{align}
  [ \widetilde{\bU}_s^{\mathsf T} \bS \bB \bT^{\mathsf T} \widetilde{\bV}_t]_k &=  \argmin_{\bX \atop \mathsf{r}(\bX) \leq k} \| \widetilde{\bU}_s \bX \widetilde{\bV}_t^{\mathsf T}- \bS \bB \bT^{\mathsf T} \|_F \nonumber  
 \end{align}
 {This implies that $\argmin_{\bX , \mathsf{r}(\bX) \leq k} \| \bS(\bU \bX \bV - \bB) \bT^{\mathsf T} \|_F$ has closed form}
\begin{align}	  	\widetilde{\bX} =\widetilde{\bV}_s \widetilde{\bSigma}_s^\dagger [\widetilde{\bU}_s^{\mathsf T} \bZ \widetilde{\bV}_t ]_k \widetilde{\bSigma}_t^\dagger \widetilde{\bU}_t^{\mathsf T} \label{eq:solution6}
	 \end{align}
	  
Recall, $\widetilde{\bX} =\widetilde{\bV}_s \widetilde{\bSigma}_s^\dagger [\widetilde{\bU}_s^{\mathsf T} \bZ \widetilde{\bV}_t ]_k \widetilde{\bSigma}_t^\dagger \widetilde{\bU}_t^{\mathsf T} = \bU'\bSigma'\bV'^{\mathsf T}$.  Substituting~\eqnref{solution6} in~\eqnref{step36} and the fact that $\bB= \widehat{\bA} + \bS^\dagger \bN_2 (\bT^\dagger)^{\mathsf T}$, we have
 \begin{align*}
 \| \bU \bU' \bSigma' (\bV^{\mathsf T} \bV')^{\mathsf T} - \widehat{\bA} \|_F - \| \bS^\dagger \bN_2 (\bT^{\mathsf T})^\dagger \|_F 
	&  \leq  \| \bU \bU' \bSigma' (\bV^{\mathsf T} \bV')^{\mathsf T} - {\bB} \|_F \\
 	&  \leq (1+\alpha)^6 \|\widehat{\bA} - [\widehat{\bA}]_k\|_F +  O( \|  \bS^\dagger \bN_2 (\bT^{\mathsf T})^\dagger \|_F +  \| \bS_2 \|_F ).
\end{align*}
This in particular implies that
\begin{align*} &   \| \bU \bU' \bSigma' (\bV^{\mathsf T} \bV')^{\mathsf T} - \widehat{\bA} \|_F \leq 	(1+\alpha)^6 \|\widehat{\bA} - [\widehat{\bA}]_k\|_F  + O( \|  \bS^\dagger \bN_2 (\bT^{\mathsf T})^\dagger \|_F +  \| \bS_2 \|_F).
   \end{align*}
 Scaling the value of $\alpha$ by a constant completes the proof of~\lemref{third}.
\end{proof}

\medskip \noindent \textbf{Performing step~(iii).} In order to complete the proof, we compute an upper bound on $\| \widehat{\bA} - [\widehat{\bA}]_k \|_F$. For this, we  need the Weyl's perturbation theorem (\thmref{Weyl}).

\begin{lemma} \label{lem:second}
Let $d$ be the maximum of the rank of $\bA$ and $\widehat{\bA}$. Let $\sigma_1,\cdots, \sigma_d$ be the singular values of $\bA$ and $\sigma_1',\cdots, \sigma_d'$ be the singular values of $\widehat{\bA}$. Then $| \sigma_i - \sigma_i' | \leq \sigma$ for all $1 \leq i \leq d$.
\end{lemma}
\begin{proof}
The lemma follows from the basic application of~\thmref{Weyl}. We can write $\widehat{\bA} =\begin{pmatrix} \bA & \mathbf{0} \end{pmatrix} + \begin{pmatrix}  \mathbf{0} & \sigma_\mathsf{min} \I_m \end{pmatrix}.$ The lemma follows since, by construction, all the singular values of $ \begin{pmatrix}  \mathbf{0} &  \sigma_\mathsf{min} \I_m \end{pmatrix}$ are $ \sigma_\mathsf{min}$. 
\end{proof}

To compute the additive error, we need to bound $ \|  \bS^\dagger \bN_2 (\bT^{\mathsf T})^\dagger \|_F$ and $\| \bS_2 \|_F$. This is done by the following two lemmas.

\begin{claim}  \label{claim:S_2}
Let $\cD_R$ be  a distribution that satisfies $(\alpha,\delta)$-subspace embedding for generalized regression. Let $\bS_2$ be as defined above. Then with probability $99/100$ over $\widehat{\bPhi} \sim \cD_R$, $\| \bS_2 \|_F =  \rho_1 \sqrt{kn(1+ \alpha)}.$
\end{claim}
\begin{proof}
Let $\mathbf{G}= \bPsi \bA \bPhi ([\widehat{\bA}]_k \bPhi)^\dagger$. $\mathbf{G}$ is an $m \times k$ matrix. When $\alpha \leq 1$,  $\mathbf{G}$ has rank $k$. This implies that there exist a $t \times k$ matrix $\widehat{\bU}$ with orthonormal columns such that $\mathbf{G} \mathbf{G}^\dagger = \widehat{\bU} \widehat{\bU}^{\mathsf T}$. Therefore, 
	$\bPsi \bS_2 =  \mathbf{G}  \mathbf{G}^\dagger  \bN_2 = \widehat{\bU} \widehat{\bU}^{\mathsf T} \bN_1.$ 
 From the second claim of~\lemref{phi} and the choice of the parameter $t$, $\| \bS_2 \|_F^2 \leq (1+ \alpha) \| \widehat{\bU} \widehat{\bU}^{\mathsf T} \bN_1 \|_F^2$.
 Since every entries of $\bN_1$ are picked i.i.d. and $\widehat{\bU} \widehat{\bU}^{\mathsf T}$  is an orthonormal projection onto a $k$-dimensional subspace, we have $\| \bS_2 \|_F = O( \rho_1 \sqrt{kn(1+ \alpha)}).$
\end{proof}

The following claim follows from~\lemref{SRHTinverse}.
\begin{claim} \label{claim:SRHTinversetwice}
Let $\bS, \bT \sim \cD_A$. Then for any matrix $\bN_2$ of appropriate dimension,
 $ \|  \bS^\dagger \bN_2 (\bT^{\mathsf T})^\dagger \|_F = \| \bN_2 \|_F.$
\end{claim}
 \begin{proof} 
   Let $\bC=\bS^\dagger \bN_2 (\bT^{\mathsf T})^\dagger  $. Then $\bS \bC \bT^\mathsf T = \bS \bS^\dagger \bN_2 (\bT \bT^\dagger)^\mathsf T$. Now $\bS \bS^\dagger$ (similarly, $\bT \bT^\dagger$) is a projection unto a random subspace of dimension $k$.  Since every entries of $\bN_2$ is picked i.i.d. from $\cN(0,\rho^2)$, $\bS \bC  \bT^\mathsf T =  \widetilde{\bN}_2$, where $\widetilde{\bN}_1$ is an $v \times v$ matrix with every entries picked i.i.d. from $\cN(0,\rho_2^2)$.  Using~\lemref{N}, this implies that 
 	\[
	\E[ \| \bS \bC \bT^\mathsf T   \|_F^2] = \E \sparen{ \| \widetilde{\bN}_2\|_F ^2} =  \sum_{i,j} \E[  (\widetilde{\bN}_2)_{ij}^2] = v^2 \rho_2^2.
	\]
The result follows using Markov's inequality and the fact that $\|\bS \bC \bT^\mathsf T \|_F^2 = (1+\alpha)^2 \| \bC \|_F^2$ and $\alpha \leq 1$.
 \end{proof}
The above claim implies  that $ \| \bN_2 \|_F = O(\rho_2 v)$ with  probability $99/100$. 

Since $\| \widehat{\bA} - [\widehat{\bA}]_k \|_F^2 \leq \sum_{i >k} \sigma_i'^2$ and $\| {\bA} - {\bA}_k \|_F^2 \leq \sum_{i >k} \sigma_i^2$, combining  \lemref{first}, \lemref{third}, \claimref{S_2}, \claimref{SRHTinversetwice},   \lemref{second}, and~\lemref{N}, we have the final utility bound.

\begin{lemma} \label{lem:low_space_utility}
Let $\rho_1,\rho_2$, and $\sigma_{\mathsf{min}}$ be as defined in~\thmref{meta}. With probability $99/100$ over the coins of the algorithm {\scshape Private-Optimal-Space}-$\lrf$,  the output of {\scshape Private-Optimal-Space}-$\lrf$ satisfies
 \begin{align*}
 & \| \begin{pmatrix} \bA & \mathbf{0} \end{pmatrix} - \mathbf{M}_k \|_F \leq (1+\alpha) \| \bA - [\bA]_k \|_F + O(\sigma_\mathsf{min} \sqrt{m}+  \rho_1 \sqrt{kn(1+ \alpha)}  + \rho_2 v). \end{align*}
\end{lemma}

Now observe that $v=O((\eta/\alpha^2) \log(k/\delta)) \ll \min \set{m,n}$ and $1+\alpha \leq 2$. If former is not the case, then there is no reason to do a random projection. Therefore, the term $\rho_2v$ is subsumed by the rest of the term. The result follows by setting the values of $\rho_1$ and $\sigma_\mathsf{min}$.
\end{proof}


\subsubsection{Privacy Proof of {\scshape Private-Space-Optimal-$\lrf$}} \label{app:privacyJL}
Our privacy result can be restated as the following lemma.
 \begin{lemma} \label{lem:low_space_private}
 If $\sigma_\mathsf{min}, \rho_1$ and $\rho_2$ be as in~\thmref{meta}, then the algorithm presented in~\figref{spacelowprivate}, {\scshape Algorithm 2}, is $(3\varepsilon,3\delta)$-differentially private.
 \end{lemma}
We prove the lemma when $m \leq n$. The case for $m \geq n$ is analogous after inverting the roles of $\widehat{\bPhi}$ and $\bPsi$. Let $\bA$ and $\bA'$ be two neighboring matrices, i.e., $\mathbf{E}= \bA - \bA' = \mathbf{u} \mathbf{v}^{\mathsf T}$. Then $\widehat{\bA}$ and $\widehat{\bA}'$, constructed by {\scshape Optimal-Space-Private-}$\lrf$, has the following property: $\widehat{\bA}' = \widehat{\bA} + \begin{pmatrix}   \mathbf{E} & \mathbf{0}  \end{pmatrix}$.

\begin{claim} 
If $\rho_1 =\frac{\sqrt{(1+\alpha)\ln(1/\delta)}}{\varepsilon}$ and $\rho_2 =\frac{(1+\alpha)\sqrt{\ln(1/\delta)}}{\varepsilon}$, then  publishing $\bY_r$ and $\bZ$ preserves $(2\varepsilon,2\delta)$-differential privacy. 
\end{claim}
\begin{proof}
We use the second claims of~\lemref{phi} and \lemref{S}, i.e., $\| \bS \mathbf{D} \|_F^2 = (1 \pm \alpha) \|\mathbf{D}\|_F^2$ and $\| \bPsi \mathbf{D}  \|_F^2 = (1 \pm \alpha) \|\mathbf{D}\|_F^2$ for all $\mathbf{D}$, where $\bS \sim \cD_A$ and $\bPsi \sim \cD_R$. Let $\bA$ and $\bA'$ be two neighboring matrices such that $\mathbf{E}= \bA - \bA' = \mathbf{u} \mathbf{v}^{\mathsf T}$. Then $\| \bS \begin{pmatrix} \mathbf{E} & \mathbf{0} \end{pmatrix} \bT^{\mathsf T} \|_F^2 \leq (1+\alpha) \|  \begin{pmatrix} \mathbf{E} & \mathbf{0} \end{pmatrix}  \bT^{\mathsf T}\|_F^2 \leq (1+\alpha)^2$. Publishing $\bZ$ preserves $(\varepsilon,\delta)$-differential privacy follows from considering the vector form of the matrix $\bS \widehat{\bA} \bT^{\mathsf T}$ and $\bN_2$ and applying~\thmref{gaussian}. Similarly, we use~\thmref{gaussian} and the fact that,  for any matrix $\bC$ of appropriate dimension, $\| \bPsi \bC  \|^2 \leq (1 + \alpha) \| \bC \|_F^2$, to prove that publishing $\bPsi \widehat{\bA}  + \bN_1$ preserves differential privacy. 
\end{proof}

We next prove that $\bY_c$ is $(\varepsilon,\delta)$-differentially private. This would complete the proof of~\lemref{low_space_private}  by combining~\lemref{post} and~\thmref{DRV10} with the above claim. 
Let $\bA - \bA' = \mathbf{E}=\mathbf{u} \mathbf{v}^{\mathsf T}$ and let $\widehat{\mathbf{v}} = \begin{pmatrix} \mathbf{v} & \mathbf{0}^m \end{pmatrix}$. Then $\widehat{\bA} - \widehat{\bA}' =  \mathbf{u}\widehat{\mathbf{v}}^{\mathsf T}$. Since $\bPhi^{\mathsf T}$ is sampled from $\cD_R$, we have $\| \bPhi^{\mathsf T} \mathbf{W} \|_F^2 = (1+\alpha) \| \mathbf{W} \|_F^2$ for any matrix $\bW$ with probability $1-\delta$ (second claim of~\lemref{phi}). Therefore, $\mathbf{u} \mathbf{v}^{\mathsf T}\bPhi =  (1+\alpha)^{1/2}\mathbf{u} \widetilde{\mathbf{v}}^{\mathsf T} = \widetilde{\mathbf{u}} \widetilde{\mathbf{v}}^{\mathsf T}$ for some unit vectors $\mathbf{u}$, $\widetilde{\mathbf{v}}$ and $ \widetilde{\mathbf{u}} = (1+\alpha)^{1/2}\mathbf{u}$. 
We now show that   $ \widehat{\bA}\bPhi \bOmega_1$ preserves privacy. 
We prove that each row of the published matrix preserves $(\varepsilon_0, \delta_0)$-differential privacy for some appropriate $\varepsilon_0,\delta_0$, and then invoke~\thmref{DRV10} to prove that the published matrix preserves $(\varepsilon,\delta)$-differential privacy. 

It may seem that the privacy of $\bY_c$ follows from the result of Blocki {\it et al.}~\cite{BBDS12}, but this is not the case because of the following reasons.
\begin{enumerate}
	\item The definition of neighboring matrices considered in this paper is different from that of Blocki {\it et al.}~\cite{BBDS12}. To recall, Blocki {\it et al.}~\cite{BBDS12} considered two matrices neighboring if they differ in at most one row by a unit norm. In our case, we consider two matrices are neighboring if they have the form $\mathbf{u} \mathbf{v}^{\mathsf T}$ for unit vectors $\mathbf{u},\mathbf{v}$. 
	\item We multiply the Gaussian matrix to a random projection of $\widehat{\bA}$ and not to $\bA$ as in the case of Blocki {\it et al.}~\cite{BBDS12}, i.e., to $\widehat{\bA} \bPhi$ and not to $\widehat{\bA}$. 
\end{enumerate}
If we do not care about the run time efficiency of the algorithm, then we can set $\widehat{\bPhi}:=\Omega$ instead of $\widehat{\bPhi}:=\bOmega \bPhi$. In this case, we would not need to deal with the second issue mentioned above. 

We first  give a brief overview of how to deal with these issues here. The first issue is resolved by analyzing $(\widehat{\bA} - \widehat{\bA}')\bPhi.$ We observe that this expression can be represented in the form of $\widetilde{\mathbf{u}} \widetilde{\mathbf{v}}^{\mathsf T}$, where $ \widetilde{\mathbf{u}} = (1+\alpha)^{1/2}\mathbf{u}$  for some  $\|\mathbf{u}\|_2=1$, $\|\widetilde{\mathbf{v}}\|_2=1$. 
The second issue can be resolved by observing that $\bPhi$ satisfies $(\alpha,\delta)$-{\scshape JLP} because of the choice of $t$. Since the rank of $\widehat{\bA}$ and $\widehat{\bA} \bPhi$ are the same, the singular values of $\widehat{\bA} \bPhi$ are within a multiplicative factor of $(1 \pm \alpha)^{1/2}$ of the singular values of $\bPhi$ with probability $1-\delta$  due to Sarlos~\cite{Sarlos06}. Therefore, 
we scale the singular values of $\widehat{\bA}$ appropriately.

We now return to the proof. Denote by $\widehat{\bA}=  \begin{pmatrix} \bA & \sigma_\mathsf{min} \I_m \end{pmatrix}$ and by $\widehat{\bA}'= \begin{pmatrix} {\bA}' &  \sigma_\mathsf{min} \I_m \end{pmatrix}$, where $\bA - \bA' = \mathbf{u} \mathbf{v}^{\mathsf T}$. Then $\widehat{\bA}' - \widehat{\bA} = \begin{pmatrix} \mathbf{u} \mathbf{v}^{\mathsf T} & \mathbf{0} \end{pmatrix}$.  Let $\bU_\bC \bSigma_\bC \bV_\bC^{\mathsf{T}}$ be the $\mathsf{SVD}$ of $\bC= \widehat{\bA} \bPhi$ and $\widetilde{\bU}_\bC \widetilde{\bSigma}_\bC \widetilde{\bV}_\bC^{\mathsf{T}}$ be the $\mathsf{SVD}$ of $\widetilde{\bC} =  \widehat{\bA}' \bPhi$. From above discussion, we know that if $\bA - \bA' = \mathbf{u} \mathbf{v}^{\mathsf T}$, then $\bC - \widetilde{\bC} = (1+\alpha)^{1/2}\widetilde{\mathbf{u}} \widetilde{\mathbf{v}}^{\mathsf T}$ for some unit vectors $ \widetilde{\mathbf{u}} $ and $ \widetilde{\mathbf{v}}.$ 
 For notational brevity, in what follows we write $\mathbf{u}$ for $ \widetilde{\mathbf{u}} $ and $\mathbf{v}$ for $ \widetilde{\mathbf{v}}.$

Note that both $\bC$ and $\widetilde{\bC}$ are  full rank matrices because of the construction; therefore $\bC \bC^{\mathsf T}$ (respectively, $\widetilde \bC \widetilde \bC^{\mathsf T}$) is a full dimensional $m \times m$ matrix. This implies that the affine transformation of the multi-variate Gaussian is well-defined (both the covariance $(\bC \bC^{\mathsf T})^{-1}$ has full rank and $\det(\bC \bC^{\mathsf T})$ is non-zero). That is, the {\scshape PDF} of the  distributions of the  rows, corresponding to $\bC$ and $\widetilde{\bC}$, is just a linear transformation of $\cN(\mathbf{0},\I_{m \times m})$. Let $\by \sim \cN(0,1)^t$.  
\begin{align*}
	\PDF_{\bC Y} (\bx) & = \frac{1}{\sqrt{(2\pi)^t \det(\bC  \bC^{\mathsf T})}} e^{(- \frac{1}{2} \bx (\bC  \bC^{\mathsf T})^{-1} \bx^{\mathsf T})} \\
	\PDF_{\widetilde{\bC} Y} (\bx) & = \frac{1}{\sqrt{(2\pi)^t \det(\widetilde{\bC}  \widetilde{\bC}^{\mathsf T})}} e^{(- \frac{1}{2} \bx (\widetilde{\bC}  \widetilde{\bC}^{\mathsf T})^{-1} \bx^{\mathsf T}) }
\end{align*}

 Let  $\varepsilon_0 = \frac{\varepsilon}{\sqrt{4 t \ln (1/\delta)} \log(1/\delta)}$ and  $\delta_0 = {\delta}/{2t},$ 
We prove that every row of the published matrix is $(\varepsilon_0,\delta_0)$ differentially private. 
Let $\bx$ be sampled either from $\cN(\mathbf{0},\bC  \bC^{\mathsf T})$ or  $\cN(\mathbf{0},\widetilde{\bC} \widetilde{\bC}^{\mathsf T})$.
It is straightforward to see that the combination of~\claimref{claim1} and~\claimref{claim2} below proves differential privacy for a row of published matrix. The lemma then follows by an application of~\thmref{DRV10} and our choice of $\varepsilon_0$ and $\delta_0$.

\begin{claim} \label{claim:claim1}
 Let $\bC$ and $\varepsilon_0$ be as defined above. Then 
\begin{align*} e^{- \varepsilon_0}  \leq  \sqrt{\frac{\det(\bC  \bC^{\mathsf T})}{\det(\widetilde{\bC}  \widetilde{\bC}^{\mathsf T})}} \leq  e^{\varepsilon_0}. 
\end{align*}
\end{claim}

\begin{claim} \label{claim:claim2}
Let $\bC, \varepsilon_0$, and $\delta_0$ be as defined earlier. Let $\by \sim \cN(0,1)^m$. If $\bx$ is sampled either from $\bC \by$ or $\widetilde{\bC} \by$, then we have
\begin{align*}
\p \sparen{ \left| \bx^{\mathsf T}  (\bC  \bC^{\mathsf T})^{-1} \bx - \bx^{\mathsf T}  (\widetilde{\bC}  \widetilde{\bC}^{\mathsf T})^{-1} \bx\right|   \leq \varepsilon_0} \geq 1 -\delta_0.\end{align*} 
\end{claim}

\begin{proof}[Proof of~\claimref{claim1}]
The claim follows simply as in~\cite{BBDS12} after a slight modification. More concretely, we have $\det(\bC  \bC^{\mathsf T}) = \prod_i \sigma_i^2$, where $\sigma_1 \geq \cdots \geq \sigma_m \geq  \sigma_{\mathsf{min}}(\bC)$ are the singular values of $\bC$. 
 Let $\widetilde{\sigma}_1 \geq \cdots \geq \widetilde{\sigma}_m \geq  \sigma_{\mathsf{min}}(\widetilde{\bC})$ be its singular value  for $\widetilde{\bC}$. The matrix $\mathbf{E}$ has only one singular value $\sqrt{1+\alpha}$. This is because $ \mathbf{E}\mathbf{E}^{\mathsf T} = (1+\alpha) \mathbf{v} \mathbf{v}^{\mathsf T}$. To finish the proof of this claim, we use~\thmref{lidskii}. 

Since the singular values of $\bC - \widetilde{\bC}$ and $\widetilde{\bC} -\bC$ are the same,  Lidskii's theorem (\thmref{lidskii}) gives $\sum_i(\sigma_i - \widetilde{\sigma}_i) \leq \sqrt{1+\alpha}$. 
 Therefore, with probability $1-\delta$,
\begin{align*} 
\sqrt{\prod_{i: \widetilde{\sigma}_i \geq \sigma_i} \frac{\widetilde{\sigma}_i^2}{\sigma_i^2}} &= \prod_{i: \widetilde{\sigma}_i \geq \sigma_i} \paren{1 + \frac{\widetilde{\sigma}_i-\sigma_i}{\sigma_i} } \\
&\quad \leq \exp \paren{\frac{\varepsilon}{32\sqrt{(1+\alpha) t \log (2/\delta)} \log (t/\delta)} \sum_i (\widetilde{\sigma}_i - \sigma_i)} \\
	&\quad \leq e^{\varepsilon_0/2}. \end{align*}

The first inequality holds because $\bPhi \sim \cD_R$ satisfies $(\alpha,\delta)$-{\scshape JLP}  due to the choice of $t$ (second claim of~\lemref{phi}). Since $\bC$ and $\bA$ have same rank, this implies that all the singular values of $\bC$ are  within a $(1\pm \alpha)^{1/2}$ multiplicative factor of $\widehat{\bA}$ due to a result by Sarlos~\cite{Sarlos06}.  In other words, $\sigma_i \geq \sigma_{\mathsf{min}}(\bC)  \geq (1- \alpha)^{1/2} \sigma_\mathsf{min}$.
The case for all $i \in [m]$ when ${\widetilde{\sigma}_i \leq \sigma_i}$ follows similarly as the singular values of $\mathbf{E}$ and $-\mathbf{E}$ are the same. This completes the proof of~\claimref{claim1}.
\end{proof}

\begin{proof}[Proof of~\claimref{claim2}]
Without any loss of generality, we can assume $\bx = \bC \by$. The case for $\bx = \widetilde{\bC} \by$ is analogous. Let  $\bC - \widetilde{\bC}= \mathbf{v} \mathbf{u}^{\mathsf T}$. Note that $\E[(\bOmega)_{i,j}]={0}$ for all $1\leq i,j \leq m$ and $\cov((\bOmega)_{i,j})=1$ if and only if $i=j$; and $0$ otherwise. First note that the following
\begin{align*}	 \bx^{\mathsf T}  (\bC  \bC^{\mathsf T})^{-1} \bx - \bx^{\mathsf T}  (\widetilde{\bC}  \widetilde{\bC}^{\mathsf T})^{-1} \bx &= \bx^{\mathsf T}  (\bC  \bC^{\mathsf T})^{-1} (\widetilde{\bC}  \widetilde{\bC}^{\mathsf T}) (\widetilde{\bC}  \widetilde{\bC}^{\mathsf T})^{-1}  \bx - \bx^{\mathsf T}  (\widetilde{\bC}  \widetilde{\bC}^{\mathsf T})^{-1} \bx \\
	&={\bx^{\mathsf T}   \sparen{(\bC  \bC^{\mathsf T})^{-1} (\bC   \mathbf{u} \mathbf{v}^{\mathsf T} + \mathbf{v}  \mathbf{u}^{\mathsf T}  \widetilde{\bC}^{\mathsf T} ) (\widetilde{\bC}  \widetilde{\bC}^{\mathsf T})^{-1} }   \bx}.  \end{align*}

Using the singular value decomposition of $\bC=\bU_\bC \bSigma_\bC \bV_\bC^{\mathsf T}  $ and $\widetilde{\bC} = \widetilde{\bU}_\bC \widetilde{\bSigma}_\bC \widetilde{\bV}_\bC^{\mathsf T}  $, we have 
\begin{align*} 
&\paren{ \bx ^{\mathsf T}  (\bU_\bC \bSigma_\bC^{-1}\bV_\bC^{\mathsf T}  ) \mathbf{u}} \paren{\mathbf{v}^{\mathsf T}   (\widetilde{\bU}_\bC \widetilde{\bSigma}_\bC^{-2} \widetilde{\bU}_\bC^{\mathsf T}  )  \bx} +   \paren{\bx^{\mathsf T} (\bU_\bC \bSigma_\bC^{-2}\bU_\bC^{\mathsf T}  )\mathbf{v}} \paren{\mathbf{u}^{\mathsf T}   (\widetilde{\bV}_\bC \widetilde{\bSigma}_\bC^{-1} \widetilde{\bU}_\bC^{\mathsf T}  )  \bx} 
\end{align*}

Since $\bx \sim \bC   \by$, where $\by \sim \cN(0,1)^m$, we can write the above expression as   $\tau_1\tau_2 + \tau_3\tau_4$, where
	\begin{align*} 
		\tau_1&=  \paren{ \by^{\mathsf T}  \bC^{\mathsf T}  (\bU_\bC \bSigma_\bC^{-1}\bV_\bC^{\mathsf T}  ) \mathbf{u} } 
	        & \tau_2= \paren{\mathbf{v}^{\mathsf T}    (\widetilde{\bU}_\bC \widetilde{\bSigma}_\bC^{-2} \widetilde{\bU}_\bC^{\mathsf T}  )  \bC  \by}  \\
		\tau_3&= {\paren{\by^{\mathsf T}  \bC^{\mathsf T} (\bU_\bC \bSigma_\bC^{-2}\bU_\bC^{\mathsf T}  ) \mathbf{v}}}
		& \tau_4={\paren{\mathbf{u}^{\mathsf T}    (\widetilde{\bV}_\bC \widetilde{\bSigma}_\bC^{-1} \widetilde{\bU}_\bC^{\mathsf T}  )\bC  \by}}. 
		\end{align*}

Now since $\| \widetilde{\bSigma}_\bC \|_2, \| \bSigma_\bC \|_2 \geq \sigma_{\min}(\bC)$, plugging in the $\mathsf{SVD}$ of $\bC$ and $\bC - \widetilde{\bC} = \mathbf{v} \mathbf{u}^{\mathsf T}$, and that every term $\tau_i$ in the above expression is a linear combination of a Gaussian, i.e., each term is distributed as per $ \cN(0,\|\tau_i\|^2)$, we have the following: 
\begin{align*}
\| \tau_1\|_2 &= \| (\bV_\bC \bSigma_\bC \bU_\bC^{\mathsf T} )    (\bU_\bC \bSigma_\bC^{-1} \bV_\bC^{\mathsf T}  ) \mathbf{u} \|_2  \leq \| \mathbf{u}   \|_2 \leq \sqrt{1+\alpha}, \\
\| \tau_2\|_2 &= \| \mathbf{v}^{\mathsf T}   (\widetilde{\bU}_\bC \widetilde{\bSigma}_\bC^{-2} \widetilde{\bU}_\bC^{\mathsf T}  )  (\widetilde{\bU}_\bC \widetilde{\bSigma}_\bC \widetilde{\bV}_\bC^{\mathsf T}   - \mathbf{v}\mathbf{u}^{\mathsf T}  )\|_2  \\
	&  \leq \| \mathbf{v}^{\mathsf T}   (\widetilde{\bU}_\bC \widetilde{\bSigma}_\bC^{-2} \widetilde{\bU}_\bC^{\mathsf T}  )    \widetilde{\bU}_\bC \widetilde{\bSigma}_\bC  \widetilde{\bU}_\bC^{\mathsf T}  \|_2 
	 + \| \mathbf{v}^{\mathsf T}   (\widetilde{\bU}_\bC \widetilde{\bSigma}_\bC^{-2} \widetilde{\bU}_\bC^{\mathsf T}  )    \mathbf{v}\mathbf{u}^{\mathsf T}   \|_2   \leq \frac{1}{ \sigma_{\mathsf{min}}(\bC)} + \frac{\sqrt{1+\alpha} }{ \sigma_{\mathsf{min}}^2(\bC)},  \\
\| \tau_3\|_2 &= \| (\bV_\bC \bSigma_\bC \bU_\bC^{\mathsf T} )(\bU_\bC \bSigma_\bC^{-2}\bU_\bC^{\mathsf T}  )\mathbf{v} \|_2   \leq  \| \bSigma_\bC^{-1} \|_2 
\leq \frac{ 1}{ \sigma_{\mathsf{min}}(\bC)},  \\
\| \tau_4\|_2 
&= \| \mathbf{u}^{\mathsf T}   (\widetilde{\bV}_\bC \widetilde{\bSigma}_\bC^{-1} \widetilde{\bU}_\bC^{\mathsf T}  )   (\widetilde{\bU}_\bC \widetilde{\bSigma}_\bC \widetilde{\bV}_\bC^{\mathsf T}   - \mathbf{v}\mathbf{u}^{\mathsf T}  ) \|_2  \\ 
		& \quad \leq \|  \mathbf{u}^{\mathsf T}   (\widetilde{\bV}_\bC \widetilde{\bSigma}_\bC^{-1} \widetilde{\bU}_\bC^{\mathsf T}  )  (\widetilde{\bU}_\bC \widetilde{\bSigma}_\bC 
\widetilde{\bV}_\bC^{\mathsf T}   \|_2  + \|\mathbf{u}^{\mathsf T}   (\widetilde{\bV}_\bC \widetilde{\bSigma}_\bC^{-1} \widetilde{\bU}_\bC^{\mathsf T}  )  \mathbf{v}  \|_2   \leq \sqrt{1+\alpha}  + \frac{\sqrt{1+\alpha}}{ \sigma_{\mathsf{min}}(\bC)}.
\end{align*}
 
 Using the concentration bound on the Gaussian distribution, each term, $\tau_1,\tau_2,\tau_3$, and $\tau_4$, is less than $\|\tau_i\| \ln (4/\delta_0)$ with probability $1 - \delta_0/2$.  The second claim now follows because with probability $1-\delta_0$,
 \begin{align*}
 & \left| \bx^{\mathsf T}  (\bC  \bC^{\mathsf T})^{-1} \bx - \bx^{\mathsf T}  (\widetilde{\bC}^{\mathsf T}   \widetilde{\bC})^{-1} \bx\right|   \leq  2\paren{\frac{\sqrt{1+\alpha}}{ \sigma_{\mathsf{min}}(\bC)} + \frac{{1+\alpha}}{ \sigma_{\mathsf{min}}^2(\bC)}} \ln (4/\delta_0)  \leq \varepsilon_0 ,
\end{align*} 
where the second inequality follows from the choice of $\sigma_\mathsf{min}$ and the fact that $\sigma_\mathsf{min}(\bC) \geq (1-\alpha)^{1/2} \sigma_\mathsf{min}$.
\end{proof}
\lemref{low_space_private} follows by combining~\claimref{claim1} and~\claimref{claim2}.


\subsection{{Low Space Differentially private Low-rank Factorization Under $\priv_2$}} \label{app:meta_time}
In~\appref{lowspaceprivate}, we gave an optimal space algorithm  for computing $\lrf$ under $\priv_1$. 
However, we cannot use the algorithm of~\thmref{meta} to simultaneously prove differential privacy under $\priv_2$ and get optimal additive error. This is because we need to perturb the input matrix by a noise proportional to $\min \set{\sqrt{km},\sqrt{kn}}$ to preserve differential privacy under $\priv_2$. As a result, the additive error  would depend linearly on $\min \set{{m},{n}}$. We show that by maintaining noisy sketches $\bY $ and $\bZ$ and some basic linear algebra, we can have a differentially private algorithm that outputs an optimal error $\lrf$ of an $m \times n$ matrix under $\priv_2$. More concretely, we prove the following theorem. 

\begin{figure} [t!]
\begin{center} \fbox{
\begin{minipage}[l]{6.2in}
{
\begin{center} \underline{\scshape Private-Frobenius-$\lrf$} \end{center}
{\bf Initialization.} Set $\eta=\max \set{k,\alpha^{-1}}$ $t=O(\eta \alpha^{-1} \log (k/\delta)), v=O( \eta\alpha^{-2}   \log(k/\delta))$, and $\rho={\sqrt{(1+\alpha)\ln(1/\delta)}}/{\varepsilon}$. 
	  Sample $\bPhi \in \R^{m \times t}$ from $\cD_R$ as in~\lemref{phi} and $\bS \in \R^{v \times n}$ from $\cD_R$ as in~\lemref{S}. 

\medskip
\textbf{Computing the factorization.} On input the matrix $\bA$,
\begin{CompactEnumerate}
	 \item Sample  $\bN_1\sim \cN(0,\rho^2)^{m \times t}, \bN_2\sim \cN(0,\rho^2)^{v \times n}$.  
	 \item Compute $\bY = \bA \bPhi + \bN_1$ and  $\bZ = \bS \bA +\bN_2$.
	\item Compute a matrix $\bU \in \R^{m \times t}$ whose columns are an orthonormal basis  for the column space of $\bY$.
	\item Compute the singular value decomposition of $\mathbf{S} \bU \in \R^{v \times t}$. Let it be $\widetilde{\bU} \widetilde{\bSigma} \widetilde{\bV}^{\mathsf T}.$ 
	\item Compute the singular value decomposition of $ \widetilde{\bV} \widetilde{\bSigma}^{\dagger} [\widetilde{\bU}^{\mathsf T}\bZ]_k$. 
	Let it be $\bU' \bSigma' \bV'^{\mathsf T}$. 
	\item Output $\widetilde{\bU}= \bU \bU'$, $\widetilde{\bSigma}= \bSigma'$ and $\widetilde{\bV}=\bV'$. 
\end{CompactEnumerate}
}
\end{minipage}
} \caption{Differentially private $\lrf$ Under $\priv_2$} \label{fig:private}
\end{center}
\end{figure}

 \begin{theorem} \label{thm:naive}  \label{thm:naive_meta} \label{thm:meta_naive}
	Let  $m, n \in \N$ and $\alpha,\varepsilon,\delta$ be the input parameters.  Let $k$ be the desired rank of the factorization,  $s=\max\{m,n\}$, $u=\min\{m,n\}$,  and $\eta=\max \set{k,\alpha^{-1}}$. Given a private input matrix $\bA \in \R^{m \times n}$, the factorization $\widetilde{\bU}, \widetilde{\bSigma}, \widetilde{\bV}$ outputted by  the algorithm, {\scshape Private-Frobenius-$\lrf$}, presented in~\figref{private}, is a $k$-rank factorization and 
	satisfies the following properties:
	\begin{enumerate}
	\item  {\scshape Private-Frobenius-$\lrf$} is $(\varepsilon,\delta)$-differentially private under $\priv_2$. \label{frob_privacy}
	\item Let $\mathbf{M}_k:=\widetilde{\bU} \widetilde{\bSigma}\widetilde{\bV}^{\mathsf T}$. With probability $9/10$ over the coins of {\scshape Private-Frobenius-$\lrf$}, \label{frob_correctness}
	 \begin{align*}
	   \| \bA - \mathbf{M}_k \|_F \leq (1+\alpha) \| \bA - [\bA]_k \|_F + {O}\paren{\paren{ \sqrt{k s} + \sqrt{ \frac{u \eta}{\alpha^2} } }  \frac{\sqrt{ \log(1/\delta)}}{\varepsilon}  } .
	  \end{align*} 
	\item The space used by {\scshape Private-Frobenius-$\lrf$} is $O((m+n\alpha^{-1}) \eta \alpha^{-1}  \log(k/\delta))$. \label{frob_space}
	\item The time required to compute the factorization is $O\paren{(\mathsf{nn}(\bA) \log(1/\delta) + \frac{(m+n	 \alpha^{-2})  \eta^2  \log^2(k/\delta)}{\alpha^{2}} +\frac{ \eta^3  \log^3(k/\delta)}{\alpha^{s3}}}$. \label{frob_time}
	\end{enumerate}
 \end{theorem}

We present the algorithm when $m \geq n$, i.e., $s = m$ and $u=n$. 
The case when $m < n$ follows by symmetry. 
The space required by the algorithm is the space required to store $\bY$ and $\bZ$, which is $mt + nv = O((m+n\alpha^{-1})\eta \alpha^{-1} \log(k/\delta))$. This proves part~\ref{frob_space} of~\thmref{naive_meta}. 
 For the running time of part~\ref{frob_time} of~\thmref{naive_meta}, we have the following.
\begin{enumerate}
	 \item 
	 Computing the sketch $\bY$ requires $O(\mathsf{nn}(\bA)\log(1/\delta))  +mt^2$ time and computing the sketch $\bZ$ requires $O(\mathsf{nn}(\bA) \log(k/\delta)) + nv^2$. 
	 \item 
	 Computing the orthonormal basis $\bU$ requires $mt^2 = O(m \eta^2 \alpha^{-2}\log^2(k/\delta))$ time.
	 \item 
	 Computing a SVD of the matrix $\bS \bU$ requires $vt^2=O(\eta^3 \alpha^{3}\log^3(k/\delta)).$
	 \item 
	 Computation of $[\widetilde{\bU}^{\mathsf T}\bZ]_k$ requires $O(n\eta^{2} \alpha^{-2}\log^2(k/\delta)).$
	 \item 
	 Computing a SVD 
	 in Step 4 requires $nt^2=O(n \eta^2 \alpha^{-2}\log^2(k/\delta))$ time.
\end{enumerate}
Combining all these terms, we have our claim on running time.

\subsubsection{Privacy Proof of {\scshape Private-Frobenius-$\lrf$}} 
The following  lemma proves Part~\ref{frob_privacy} of~\thmref{naive_meta}.
 \begin{lemma} \label{lem:rho}
 {\scshape Private-Frobenius-$\lrf$} is $(\varepsilon,\delta)$ differentially private.
 \end{lemma}
 \begin{proof}
 We use the second claim of both~\lemref{phi} and~\lemref{S}, i.e., for all $\mathbf{D}$, $\|  \mathbf{D} \bT \|_F^2 \leq (1 - \alpha) \|\mathbf{D}\|_F^2$ for $\bT \sim \cD_A$ and $\|  \mathbf{D} \bPhi  \|_F^2 \leq (1 - \alpha) \|\mathbf{D}\|_F^2$ for  $\bPhi \sim \cD_R$. Let $\bA$ and $\bA'$ be two neighboring matrices such that $\mathbf{E}= \bA - \bA' $ has Frobenius norm $1$. Then $\| \bS \mathbf{E} \|_F^2 \leq (1+\alpha) \| \mathbf{E} \|_F^2 = 1+\alpha$. Publishing $\bZ$ preserves $(\varepsilon,\delta)$-differential privacy follows from considering the vector form of the matrix $\bS \bA$ and $\bN_2$ and~\thmref{gaussian}. Similarly, we use the fact that,  for any matrix $\bC$ of appropriate dimension, $\| \bPhi\bC  \|^2 \leq (1 - \alpha) \| \bC \|_F^2$, to prove that publishing $\bA \bPhi + \bN_1$ preserves differential privacy. The lemma follows by applying~\lemref{post} and~\thmref{DRV10}.
 \end{proof}

\subsubsection{Correctness Proof of {\scshape Private-Frobenius-$\lrf$}} \label{app:privacyJL}
We now prove Part~\ref{frob_correctness} of~\thmref{naive_meta}. We first show the following result.
\begin{theorem} \label{thm:complete}
Let $\mathbf{M}_k=\bU \widetilde{\bV} \widetilde{\bSigma}^\dagger [\widetilde{\bU}^{\mathsf T} \bZ]_k$ be the product of the factorization outputted by the algorithm in~\figref{private}. Then with probability $1-O(\delta)$ over $\bPhi \sim \cD_R$ and $\bS \sim \cD_A$,
\begin{align*}
 \| \mathbf{M}_k - \bA \|_F  &\leq (1+\alpha) \|\bA - [\bA]_k\|_F + 3\| \bS^\dagger \bN_2 \|_F   + 2\| \bN_1 ([\bA]_k \bPhi)^\dagger [\bA]_k \|_F.  \end{align*}
\end{theorem}
 
\begin{proof}
We prove the result by proving a series of results. We provide an upper and a lower bound on $\min_{\bX, \mathsf{r}(\bX) \leq k} \| \bY  \bX - \bB \|_F$ in terms of $\| \bA - [\bA]_k \|_F$ and the output of the algorithm.

\begin{lemma} \label{lem:firstcomplete}	
	Let $\bA$ be the input matrix. Let $\bPhi \sim \cD_R, \bS \sim \cD_A$ be as in~\figref{private} . Let $\bY = \bPhi \bA +\bN_1$ and $\bB=\bA+\bS^\dagger \bN_2$ for $\bN_1, \bN_2$ as defined in~\figref{private}. Then with probability $1-\delta$ over $\bPhi \sim \cD_R$,
	\begin{align*} \min_{\bX \atop \mathsf{r}(\bX) \leq k} \| \bY \bX - \bB \|_F &\leq (1+\alpha) \| \bA - [\bA]_k \|_F  + \| \bS^\dagger \bN_2 \|_F  + \| \bN_1 ([\bA]_k \bPhi)^\dagger [\bA]_k \|_F.  \end{align*}
\end{lemma}
\begin{proof}
	
	Set $r=k$, $\mathbf{P} = [\bA]_k^{\mathsf T} $, and $\mathbf{Q}=\bA^{\mathsf T} $ in~\lemref{phi}. Then using~\lemref{phi}, we have 
	\begin{align*} 
		 \| [\bA]_k^{\mathsf T}   \bX' - \bA^{\mathsf T} \| \leq (1+\alpha) \min_\bX \| [\bA]_k^{\mathsf T} \bX - \bA^{\mathsf T} \|_F, 
	\end{align*}
	where $\bX'=\argmin_{\bX} \| \bPhi^{\mathsf T} ([\bA]_k^{\mathsf T}\bX - \bA^{\mathsf T}  ) \|_F . $
	Let $[\bA]_k = \bU_k \bSigma_k \bV_k^\mathsf T$. Taking the transpose and the fact that the Frobenius norm is preserved under transpose and $\bX'=(([\bA]_k \bPhi)^{\mathsf T})^\dagger (\bA \bPhi)^{\mathsf T}$, we have with probability $1-\delta$ over $\bPhi \sim \cD_R$,  
	\begin{align} 
		\| \bA \bPhi ([\bA]_k \bPhi)^{\dagger} [\bA]_k - \bA \|_F   \leq (1+\alpha) \min_\bX \| [\bA]_k^{\mathsf T} \bX - \bA^{\mathsf T} \|_F \leq (1+\alpha) \| \bA - [\bA]_k \|_F, \label{eq:transpose2}  
	\end{align}
	where the inequality follows by setting $\bX= \bU_k \bU_k^{\mathsf T}$.
	
	Moreover, since $([\bA]_k \bPhi)^{\dagger} [\bA]_k$ has rank at most $k$, $\bB=\bA+\bS^\dagger \bN_2$,  and $\bY=\bA \bPhi + \bN_1$, with probability $1-\delta$ over $\bPhi \sim \cD_R$,  
	\begin{align} 
		\min_{\bX \atop \mathsf{r}(\bX) \leq k} \| \bY  \bX - \bB \|_F
		 & \leq  \| \bA \bPhi ([\bA]_k \bPhi)^{\dagger} [\bA]_k + \bN_1([\bA]_k \bPhi)^\dagger [\bA]_k - \bB \|_F \nonumber \\
		&=  \| \bA \bPhi ([\bA]_k \bPhi)^{\dagger} [\bA]_k + \bN_1([\bA]_k \bPhi)^\dagger [\bA]_k - \bA - \bS^\dagger \bN_2 \|_F \nonumber \\
		& \leq  \| \bA \bPhi ([\bA]_k \bPhi)^{\dagger} [\bA]_k - \bA \|_F + \| \bN_1([\bA]_k \bPhi)^\dagger [\bA]_k \|_F+ \| \bS^\dagger \bN_2 \|_F\label{eq:set2}  
	\end{align}
	Combining~\eqnref{transpose2} and~\eqnref{set2}, we have with probability $1-\delta$ over $\bPhi \sim \cD_R$, 
	\begin{align} 
	& \min_{\bX \atop \mathsf{r}(\bX) \leq k} \| \bY \bX - \bB \|_F \leq (1+\alpha) \| \bA - [\bA]_k \|_F + \| \bN_1 ([\bA]_k \bPhi)^\dagger [\bA]_k \|_F + \| \bS^\dagger \bN_2 \|_F. \label{eq:step12} \end{align}
	This completes the proof of~\lemref{firstcomplete}.
\end{proof}	
	~\lemref{firstcomplete} relates $\min_{\bX, \mathsf{r}(\bX) \leq k} \| \bY \bX - \bB \| $ with $(1+\alpha) \| \bA - [\bA]_k \|_F$, $ \| \bN_1([\bA]_k \bPhi)^\dagger [\bA]_k \|_F$,  and $\| \bS^\dagger \bN_2 \|_F$. 
	Since $\bU$ is the orthonormal basis for the column space of $\bY$, we further have
	\begin{align} 
		\min_{\bX \atop \mathsf{r}(\bX) \leq k} \| \bU \bX - \bB \|_F \leq \min_{\bX \atop \mathsf{r}(\bX) \leq k} \| \bY \bX - \bB \|_F. \label{eq:basis2}  
	\end{align}
	Combining~\eqnref{step12} and~\eqnref{basis2}, we have with probability $1-\delta$ over $\bPhi \sim \cD_R$, 
	\begin{align} 
	&\min_{\bX \atop \mathsf{r}(\bX) \leq k} \| \bU \bX - \bB \|_F \leq (1+\alpha) \| \bA - [\bA]_k \|_F  
+  (\| \bN_1 ([\bA]_k \bPhi)^\dagger [\bA]_k \|_F + \| \bS^\dagger \bN_2 \|_F). \label{eq:step22}
	\end{align}
\begin{lemma} \label{lem:thirdcomplete}
Let $\bU, \bB, \bA, \bS, \bN_1$, and $\bN_2$ be as  above, and let $\widetilde{\bX}= \argmin_{\bX, \mathsf{r}(\bX)=k} \| \bS(\bU \bX -\bB) \|_F$. Let $\cD_A$ be a distribution that satisfies  $(\alpha,\delta)$-subspace embedding.  Then  with probability $1-4\delta$ over $\bS \sim \cD_A$,
\begin{align*}
	& \|( \bU \widetilde{\bX} - \bB) \|_F \leq {(1+\alpha)^2} \| \bA - [\bA]_k \|_F 
+ (1+\alpha)(\| \bS^\dagger \bN_2 \|_F  + \| \bN_1 ([\bA]_k \bPhi)^\dagger [\bA]_k \|_F). 
\end{align*}
\end{lemma}
\begin{proof}
	Set $p=k/\alpha$, $\mathbf{D}=\bU$ and $\mathbf{E}=\bB$ in the statement of~\lemref{S}. Let us restrict our attention to  rank $k$ matrices $\bX$ and denote by $\widehat{\bX}= \argmin_{\bX, \mathsf{r}(\bX)=k} \| \bU \bX -\mathbf{E} \|_F$ and $\widetilde{\bX}= \argmin_{\bX, \mathsf{r}(\bX)=k} \| \bS(\bU \bX -\bB) \|_F$. Then we have with probability $1-\delta$ over $\bS \sim \cD_A$, 
	\begin{align}  
		(1 + \alpha)  \min_{\bX \atop \mathsf{r}(\bX)=k} \| \bU \bX -\bB \|_F& = \| \bU \widehat{\bX} - \bB \|_F  \geq (1 + \alpha)^{1/2} \| \bS(\bU \widehat{\bX} - \bB ) \|_F\nonumber \\
			& \geq (1 + \alpha)^{1/2}  \min_{\bX \atop \mathsf{r}(\bX)} \| \bS(\bU {\bX} - \bB ) \|_F \nonumber \\
			& = (1 + \alpha)^{1/2}  \| \bS(\bU \widetilde{\bX} - \bB  ) \|_F 
			 \geq    \| (\bU \widetilde{\bX} - \bB  ) \|_F 
			.  \label{eq:S2}  
	\end{align}
	Combining~\eqnref{S2} with~\eqnref{step22},  we have with probability $1-2\delta$ over $\bPhi \sim \cD_R$ and $\bS \sim \cD_A$, 
	\begin{align}
	& \|( \bU \widetilde{\bX} - \bB) \|_F \leq {(1+\alpha)}^2 \| \bA - [\bA]_k \|_F 
	+ (1+\alpha) (\| \bS^\dagger \bN_2 \|_F  + \| \bN_1 ([\bA]_k \bPhi)^\dagger [\bA]_k \|_F). \label{eq:step32} 
	\end{align}
This completes the proof of~\lemref{thirdcomplete}.
\end{proof}

	To finalize the proof of~\thmref{complete}, we need to compute $\widetilde{\bX}=\argmin_{\bX, \mathsf{r}(\bX) \leq k} \| \bS(\bU \bX - \bB) \|_F$ and lower bound $\|( \bU \widetilde{\bX} - \bB) \|_F$. 
	Invoking~\cite[Lem 4.2]{CW09} with $\mathbf{O} = \widetilde{\bU}$ and $\bZ = \bS \bB$, we  get $$ [ \widetilde{\bU}^{\mathsf T} \bZ]_k = [ \widetilde{\bU}^{\mathsf T} \bS \bB]_k = \argmin_{\bX \atop \mathsf{r}(\bX) \leq k} \| \widetilde{\bU} \bX - \bS \bB \|_F.$$ 

This in particular implies that 
	 \begin{align}
	 	\widetilde{\bX} =\widetilde{\bV} \widetilde{\bSigma}^\dagger [\widetilde{\bU}^{\mathsf T} \bZ]_k= \argmin_{\bX \atop \mathsf{r}(\bX) \leq k} \| \bS(\bU \bX - \bB) \|_F. \label{eq:solution2}
	 \end{align}

Using~\eqnref{solution2} in~\eqnref{step32},  and the fact that $\bB= \bA + \bS^\dagger \bN_2$, we have the final result.
 \begin{align*}
 	  \| \bU \widetilde{\bV} \widetilde{\bSigma}^\dagger [\widetilde{\bU}^{\mathsf T} \bZ]_k - \bA \|_F - \| \bS^\dagger \bN_2 \|_F &\leq  \| \bU \widetilde{\bV} \widetilde{\bSigma}^\dagger [\widetilde{\bU}^{\mathsf T} \bS \bA]_k - \bB \|_F \\
 			&\leq (1+\alpha)^2 \|\bA - [\bA]_k\|_F+ (1+\alpha) (\| \bS^\dagger \bN_2 \|_F + \| \bN_1 ([\bA]_k \bPhi)^\dagger [\bA]_k \|_F). 
  \end{align*}
			
			{This implies that}
  \begin{align*}
 \| \bU \widetilde{\bV} \widetilde{\bSigma}^\dagger [\widetilde{\bU}^{\mathsf T} \bZ]_k - \bA \|_F 
 			\leq (1+3\alpha) \|\bA - [\bA]_k\|_F  + 3 \| \bS^\dagger \bN_2 \|_F  + 2 \| \bN_1 ([\bA]_k \bPhi)^\dagger [\bA]_k \|_F.  
  \end{align*}
 This completes the proof of~\thmref{complete}.
\end{proof}

\begin{lemma} \label{lem:naiveN_1}
	Let $\bN_1 \sim \cN(0,\rho^2)^{m \times t}$. Then $$ \| \bN_1 ([\bA]_k \bPhi)^\dagger [\bA]_k \|_F= O( \rho \sqrt{km})$$ with probability $99/100$ over $\bPhi \sim \cD_R$.
 \end{lemma}
 \begin{proof}
 Let $\bC=\bN_1 ([\bA]_k \bPhi)^\dagger [\bA]_k$. Then $\bC \bPhi = \bN_1 ([\bA]_k \bPhi)^\dagger [\bA]_k \bPhi$. Now $([\bA]_k \bPhi)^\dagger [\bA]_k \bPhi$ is a projection unto a random subspace of dimension $k$.  Since every entries of $\bN_1$ is picked i.i.d. from $\cN(0,\rho^2)$, $\bC \bPhi = \bN_1 ([\bA]_k \bPhi)^\dagger [\bA]_k \bPhi = \widetilde{\bN}_1$, where $\widetilde{\bN}_1$ is an $m \times k$ matrix with every entries picked i.i.d. from $\cN(0,\rho^2)$. This is because we can write $([\bA]_k \bPhi)^\dagger [\bA]_k \bPhi$  is a projection unto a random subspace of dimension $k$. Using~\lemref{N}, this implies that 
 	\[
	\E[ \| \bC \bPhi  \|_F^2] = \E \sparen{ \| \widetilde{\bN}_1\|_F ^2} =  \sum_{i,j} \E[  (\widetilde{\bN}_1)_{ij}^2] = km \rho^2.
	\]
The result follows using Markov's inequality,  the fact that $\| \bC \bPhi \|_F^2 = (1+\alpha) \| \bC \|_F^2$, and $\alpha < 1$.
 \end{proof}
 
  \begin{lemma} \label{lem:noise} \label{lem:naiveN_2}
 Let $\bN_2 \sim \cN(0,\rho^2)^{v \times n}$. Then  $\| \bS^\dagger \bN_2 \|_F = O( \rho \sqrt{vn})$ with probability $99/100$ over $\bS \sim \cD_A$. 
 \end{lemma}
 \begin{proof} 
   Let $\bC=\bS^\dagger \bN_2$. Then $\bS \bC = \bS \bS^\dagger \bN_2$. Now $\bS \bS^\dagger$ is a projection unto a random subspace of dimension $k$.  Since every entries of $\bN_2$ is picked i.i.d. from $\cN(0,\rho^2)$, $\bS \bC  =  \widetilde{\bN}_2$, where $\widetilde{\bN}_1$ is an $v \times n$ matrix with every entries picked i.i.d. from $\cN(0,\rho^2)$.  Using~\lemref{N}, this implies that 
 	\[
	\E[ \| \bS \bC   \|_F^2] = \E \sparen{ \| \widetilde{\bN}_2\|_F ^2} =  \sum_{i,j} \E[  (\widetilde{\bN}_2)_{ij}^2] = vn \rho^2.
	\]
The result follows using Markov's inequality and the fact that $\|\bS \bC  \|_F^2 = (1+\alpha) \| \bC \|_F^2$ and $\alpha \leq 1$.
 \end{proof}
 
 ~\thmref{naive_meta} now follows from~\thmref{complete},~\lemref{naiveN_1},~\lemref{naiveN_2}, and the choice of $\rho$ in~\lemref{rho}.  



\begin{l1norm}
\input{l1norm}
\end{l1norm}
\section{Differentially private $\lrf$ Under General Turnstile Model} \label{app:streaming}
In this section, we are interested in computing a low-rank factorization of a private matrix in the general turnstile update model while preserving differential privacy. In this setting, we are allowed only one pass over the private matrix, and by the end of the stream, we are required to output a low-rank factorization.  In~\appref{time}, we give a differentially private low rank factorization under $\priv_1$. We give a differentially private low rank factorization under a much stronger privacy guarantee, $\priv_2$, in~\appref{streaming_priv_2}.

\subsection{{Space Optimal Differentially private Low-rank Factorization Under $\priv_1$}} \label{app:time}
The main idea behind the differentially private algorithm for low-rank factorization under $\priv_1$ in turnstile update model is that the corresponding algorithm ({\scshape Private-Space-Optimal-$\lrf$}) maintains linear sketches. It has been shown by Li {\it et al.}~\cite{LNW14} that in general turnstile update model, it is better off to just use linear sketches. Together with our low space algorithm, this gives us the insight to develop a private algorithm in the general turnstile update model.~\figref{streaming_spacelowprivate} gives the detail description of our algorithm. We show the following:
  \begin{figure}[t!]
 \begin{center}
\fbox{
\begin{minipage}[l]{6in}
{
\begin{center} \underline{\scshape Private-Streaming-Space-Optimal-$\lrf$} \end{center}
{\bf Initialization.} Set $\eta=\max\set{k,\alpha^{-1}}$, and the dimension of random projections to be $t=O(\eta \alpha^{-1}  \log (k/\delta)), v=O(\eta \alpha^{-2} \log(k/\delta))$. Let $\rho_1={\sqrt{(1+\alpha)\ln(1/\delta)}}/{\varepsilon}$ and $\rho_2={(1+\alpha)\sqrt{\ln(1/\delta)}}/{\varepsilon}$. Set $\sigma_\mathsf{min}={16 \log(1/\delta) \sqrt{t (1+\alpha)(1-\alpha)^{-1} \ln(1/\delta)}}/{\varepsilon}$.

\begin{enumerate}
	\item Sample $\bOmega \sim \cN(0,1)^{m \times t}$. Let  $\bPhi \in \R^{(m+n) \times m}, \bPsi  \in \R^{t \times m}$ such that $\bPhi^{\mathsf T} \sim \cD_R$, $ \bPsi \sim \cD_R$  satisfies~\lemref{CW13}. 
	\item Let  $\bS\in \R^{v \times m}, \bT \in \R^{v \times (m+n)}$  such that $\bS \sim \cD_A$,  $\bT^{\mathsf T} \sim \cD_A$  satisfies~\lemref{S}. 
	\item Sample $\bN_1 \sim \cN(0,\rho_1^2)^{t \times (m+n)}$ and $\bN_2 \sim \cN(0,\rho_2^2)^{v \times v}$. Define $\widehat{\bPhi} = t^{-1}\bPhi \bOmega \in \R^{(m+n) \times t}$. 
\end{enumerate}

\medskip \noindent \textbf{Update Stage.}
Set $\widehat{\bA} = \begin{pmatrix} \mathbf{0}^{m \times n} & \sigma_\mathsf{min} \I_m \end{pmatrix}$. Compute $\bY_c = \widehat{\bA} \widehat{\bPhi}$,  $\bY_r'=  \bPsi \widehat{\bA},$ and $\bZ'= \bS\widehat{\bA} \bT^{\mathsf T}.$

\medskip \noindent \textbf{Update rule.}  When $(i_\tau,j_\tau,\Delta_{\tau})$, where $(i_\tau,j_\tau) \in [m] \times [n]$ and $\Delta_\tau \in \R$, is streamed, update the matrices by the following rule: 
(i) $\bY_c = \bY_c + \bA_\tau \widehat{\bPhi}$, 
(ii) $\bY_r' = \bY_r' + \bPsi \bA_\tau $, and 
(iii) $\bZ'= \bZ'+ \bS \bA_\tau \bT^{\mathsf T}$, where  $\bA_\tau$ is an $m \times (n+m)$ matrix with only non-zero entry $\Delta_\tau$ in position $(i_\tau,j_\tau)$..

\medskip \noindent \textbf{Computing the factorization.} Once the matrix is streamed, we follow the following steps.
\begin{enumerate}
	\item Compute $\bY_r=  \bY_r' + \bN_1 = \bPsi \widehat{\bA} + \bN_1,$ and $\bZ= \bZ' + \bN_2 = \bS\widehat{\bA} \bT^{\mathsf T}+\bN_2$.
	\item 
	Output  {\scshape Factor}$(\bY_c,\bY_r,\bZ,\bS,\bT,m,m+n,k,t,v)$.
\end{enumerate}

\medskip \noindent \textbf{ {\scshape Factor}$(\bY_c,\bY_r,\bZ,\bS,\bT,m,n,k,t,v)$}
\begin{enumerate}
	\item Compute a matrix $\bU \in \R^{m \times t}$ whose columns are orthonormal basis  for the column space of $\bY_c$ and matrix $\bV \in \R^{t \times n}$ whose rows are the orthonormal basis for the row space of $\bY_r$.
	\item Compute a SVD of $\mathbf{S} \bU := \widetilde{\bU}_s \widetilde{\bSigma}_s \widetilde{\bV}_s^{\mathsf T} \in \R^{v \times t}$ 
and a SVD of $ \bV\mathbf{T}^{\mathsf T}:=\widetilde{\bU}_t \widetilde{\bSigma}_t \widetilde{\bV}_t^{\mathsf T} \in \R^{t \times v}.$ 
	\item Compute a SVD of $ \widetilde{\bV}_s \widetilde{\bSigma}_s^{\dagger} [\widetilde{\bU}_s^{\mathsf T} \mathbf{Z} \widetilde{\bV}_t]_k \widetilde{\bSigma}_t^{\dagger} \widetilde{\bU}_t^{\mathsf T} $. Let it is be $\bU' \bSigma' \bV'^{\mathsf T}$. 
	\item Output the matrix  $\widetilde{\bU}=\bU \bU'$ compromising of left singular vectors, diagonal matrix $\widetilde{\bSigma}=\bSigma'$, and the matrix $\widetilde{\bV}=\bV^{\mathsf T} \bV' $ with right-singular vectors. 
\end{enumerate}

}
\end{minipage}
} \caption{Differentially private $\lrf$ Under $\priv_1$  in Turnstile Update Model} \label{fig:streaming_spacelowprivate}
\end{center}
\end{figure}

\begin{theorem}  \label{thm:low_space_private} 
Let $m,n,k \in \N$ and $\alpha,\varepsilon,\delta$ be the input parameters.  Let $s=\max\{m,n\}, u=\min \set{m,n}$, $\kappa=(1+\alpha)/(1-\alpha)$, $\eta=\max\set{k,\alpha^{-1}}$, and $\sigma_\mathsf{min}={16 \log(1/\delta) \sqrt{t \kappa \ln(1/\delta)}}/{\varepsilon}$. Given an $m \times n$ matrix $\bA$ in a turnstile update model,  {\scshape Private-Streaming-Space-Optimal-$\lrf$}, described in~\figref{streaming_spacelowprivate}, outputs a factorization $\widetilde{\bU}, \widetilde{\bSigma}, \widetilde{\bV}$ such that 
\begin{enumerate}
	\item {\scshape Private-Streaming-Space-Optimal-$\lrf$} is $(3\varepsilon,3\delta)$ differentially private under $\priv_1$.  \label{streaming_privacy}
	\item  Let $\mathbf{M}_k=\widetilde{\bU} \widetilde{\bSigma} \widetilde{\bV}^{\mathsf T}. $ With probability $9/10$ over the random coins of  {\scshape Private-Streaming-Space-Optimal-$\lrf$},  \label{streaming_privatecorrectness}
 \begin{align*} 
 &\| \begin{pmatrix} \bA & \mathbf{0} \end{pmatrix} - \mathbf{M}_k \|_F \leq (1+\alpha) \| \bA - [\bA]_k \|_F  + O ( \sigma_\mathsf{min} \sqrt{u} + \varepsilon^{-1}  {\sqrt{ks  \ln(1/\delta)}} ),
\end{align*}
 	\item The space used by {\scshape Private-Streaming-Space-Optimal-$\lrf$} is $O((m+n)\eta \alpha^{-1}\log(k/\delta))$.  \label{streaming_privatespace}
	\item The initialization time  is $O((m+n)u\log (k/\delta))$ and the total computational time is 
	$$O\paren{  \mathsf{nn}(\bA) \log(1/\delta)+  \frac{(m+n)\eta^2 \log^2(k/\delta) }{\alpha^{2}}  + \frac{\eta^3  \log^3 (k/\delta)}{\alpha^{5}} }.$$ \label{streaming_privatetime}
\end{enumerate}
\end{theorem}
\begin{proof}
Part~\ref{streaming_privatespace} follows immediately by setting the values of $t$ and $v$.
Part~\ref{streaming_privatetime} of~\thmref{low_space_private} requires some computation.
More precisely, we have the following.
	Computing $\bY_c$ requires $ O(\mathsf{nn}(\bA) \log(1/\delta))+ mt^2$ time and computing $\bY_r$ requires $O(\mathsf{nn}(\bA)\log(1/\delta)) + (m+n)t^2$ time. 
	 Computing 
	 $\bU$ and $\bV$ requires 
	 $O((m+n)\eta^2 \alpha^{-2} \log^2( k/\delta))$ time.
	 Computing a SVD of the matrix $\bS \bU$  and $\mathbf{T} \bV^{\mathsf T}$ requires $vt^2 + tv^2=O(k^3 \alpha^{-5} \log^2( k/\delta))$.
	 Computing $\bZ$ requires $O(\mathsf{nn}(\bA) \log( k/\delta)  +  \eta^2/\alpha^2 \log^2( k/\delta) + n\eta^2\alpha^{-2} \log^2( k/\delta))$
	 Computation of $[\widetilde{\bU}_s^{\mathsf T}\bZ \widetilde{\bV}]_k$ requires $O_\delta(\mathsf{nn}(\bA)) + tv^2= O(\mathsf{nn}(\bA) + k^{3} \alpha^{-5} \log^3 ( k/\delta))$ time.
	 Computation of the last 
	 SVD 
	 requires 
	 $O((m+n) \eta^2 \alpha^{-2} \log^2( k/\delta))$ time.
Combining all these terms, we have our claim on the running time.

Furthermore, combining  ~\lemref{N},~\lemref{first}, \lemref{third}, \claimref{S_2},  \claimref{SRHTinversetwice}, and   \lemref{second}, we have part~\ref{streaming_privatecorrectness} while part~\ref{streaming_privacy} follows from~\lemref{low_space_private}. 
This completes the proof of~\thmref{low_space_private}.
\end{proof}

\subsection{{Differentially private Low-rank Factorization Under $\priv_2$ in Turnstile Update Model}} \label{app:streaming_priv_2}
We describe and analyze the algorithm when $m \geq n$, i.e., $s = m$ and $u=n$ in the theorem that follows. The case when $m < n$ follows by symmetry.  We prove the following. 
\begin{figure} [t!]
\begin{center} \fbox{
\begin{minipage}[l]{6in}
{
\begin{center} \underline{\scshape Private-Streaming-Frobenius-$\lrf$} \end{center}
{\bf Initialization.} Set $\eta=\max \set{k,\alpha^{-1}}$ $t=O(\eta \alpha^{-1} \log (k/\delta)), v=O( \eta\alpha^{-2}   \log(k/\delta))$, and $\rho={\sqrt{(1+\alpha)\ln(1/\delta)}}/{\varepsilon}$. 
	 Sample  $\bN_1\sim \cN(0,\rho^2)^{m \times t}, \bN_2\sim \cN(0,\rho^2)^{v \times n}$.  
	 Sample $\bPhi \in \R^{m \times t}$ from $\cD_R$ as in~\lemref{phi} and $\bS \in \R^{v \times n}$ from $\cD_R$ as in~\lemref{S}. 
	 Initialize an all zero  $m \times t$ matrix $\bY'$ and an all zero $v \times n$ matrix $\bZ'$.

\medskip
\textbf{Update rule.}  Suppose at time $\tau$, the stream is $(i_\tau,j_\tau, \Delta_{\tau})$, where $(i_\tau,j_\tau) \in [m] \times [n]$. Let $\bA_\tau$ be a matrix with only non-zero entry $\Delta_\tau$ in position $(i_\tau,j_\tau)$. Update the matrices by the following rule: $\bY' = \bY' + \bA_\tau \bPhi $ and  $\bZ' = \bZ' + \bS \bA_\tau $. 

\medskip
\textbf{Computing the factorization.} Once the matrix is streamed, we follow the following steps.
\begin{enumerate}
	\item  Compute $\bY= \bY' + \bN_1$ and $\bZ=\bZ'+\bN_2$. 
	\item Compute a matrix $\bU \in \R^{m \times t}$ whose columns are an orthonormal basis  for the column space of $\bY$.
	\item Compute the singular value decomposition of $\mathbf{S} \bU \in \R^{v \times t}$. Let it be $\widetilde{\bU} \widetilde{\bSigma} \widetilde{\bV}^{\mathsf T}.$ 
	\item Compute the singular value decomposition of $ \widetilde{\bV} \widetilde{\bSigma}^{\dagger} [\widetilde{\bU}^{\mathsf T}\bZ]_k$. 
	Let it be $\bU' \bSigma' \bV'^{\mathsf T}$. 
	\item Output $\widetilde{\bU}= \bU \bU'$, $\widetilde{\bSigma}= \bSigma'$ and $\widetilde{\bV}=\bV'$. 
\end{enumerate}
}
\end{minipage}
} \caption{Differentially private $\lrf$ Under $\priv_2$  in Turnstile Update Model} \label{fig:streaming_private}
\end{center}
\end{figure}

 \begin{theorem} \label{thm:streaming_naive}
 	Let  $m, n \in \N$ and $\alpha,\varepsilon,\delta$ be the input parameters.  Let $k$ be the desired rank of the factorization and $\eta=\max \set{k,\alpha^{-1}}$. Let $s = \max \set{m,n}$ and $u =\min \set{m,n} $. Given a private input matrix $\bA \in \R^{m \times n}$ recieved in a turnstile model, the factorization $\widetilde{\bU}, \widetilde{\bSigma}, \widetilde{\bV}$ outputted by  the algorithm, {\scshape Private-Streaming-Frobenius-$\lrf$}, presented in~\figref{streaming_private}, is a $k$-rank factorization and 
	satisfies the following properties:
	\begin{enumerate}
	\item  {\scshape Private-Streaming-Frobenius-$\lrf$} is $(\varepsilon,\delta)$-differentially private under $\priv_2$. \label{frob_streaming_privacy}
	\item Let $\mathbf{M}_k:=\widetilde{\bU} \widetilde{\bSigma}\widetilde{\bV}^{\mathsf T}$. With probability $9/10$ over the coins of {\scshape Private-Streaming-Frobenius-$\lrf$}, \label{frob_streaming_correctness}
	 \begin{align*}
	   \| \bA - \mathbf{M}_k \|_F \leq (1+\alpha) \| \bA - [\bA]_k \|_F + {O}\paren{\paren{ \sqrt{k s} + \sqrt{ \frac{u \eta}{\alpha^2} } }  \frac{\sqrt{ \log(1/\delta)}}{\varepsilon}  } .
	  \end{align*} 
	\item The space used by {\scshape Private-Streaming-Frobenius-$\lrf$} is $O((m+n\alpha^{-1}) \eta \alpha^{-1}  \log(k/\delta))$. \label{frob_space}
	\item The time required to compute the factorization is $$O\paren{(\mathsf{nn}(\bA) \log(1/\delta) + \frac{(m+n	 \alpha^{-2})  \eta^2  \log^2(k/\delta)}{\alpha^{2}} +\frac{ \eta^3  \log^3(k/\delta)}{\alpha^{s3}}}.$$ \label{frob_streaming_time}
	\end{enumerate}
 \end{theorem}
\begin{proof}
 The space required by the algorithm is the space required to store $\bY$ and $\bZ$, which is $mt + nv = O((m+n\alpha^{-1})k\alpha^{-1} \log k \log(1/\delta))$. This proves part~\ref{frob_space} of~\thmref{streaming_naive}. 
 For the running time of part~\ref{frob_time} of~\thmref{streaming_naive}, we have the following.
	 Computing the sketch $\bY$ requires $O(\mathsf{nn}(\bA)\log(k/\delta))  +mt^2$ time and computing the sketch $\bZ$ requires $O(\mathsf{nn}(\bA) \log(k/\delta)) + nv^2$. 
	 Computing the orthonormal basis $\bU$ requires $mt^2 = O(m \eta^2 \alpha^{-2}\log^2(k/\delta))$ time.
	 Computing a SVD of the matrix $\bS \bU$ requires $vt^2=O(\eta^3 \alpha^{3}\log^3(k/\delta)).$
	 Computation of $[\widetilde{\bU}^{\mathsf T}\bZ]_k$ requires $O(n\eta^{2} \alpha^{-2}\log^2(k/\delta)).$
	 Computing a SVD 
	 in Step 4 requires $nt^2=O(n \eta^2 \alpha^{-2}\log^2(k/\delta))$ time.
Combining all these terms, we have our claim on running time.

Part~\ref{frob_streaming_correctness} follows from~\lemref{naiveN_2},~\lemref{naiveN_1}, and~\thmref{complete}.
Part~\ref{frob_streaming_privacy} follows from~\lemref{rho}. 
This completes the proof of~\thmref{streaming_naive}.
\end{proof}



\section{Case Study: Normalized Row Matrices} \label{app:normalized}
An important class of matrices is matrices with normalized rows. In this section, we give a bound on such matrices.~\figref{covariance_private} is the detailed description of the algorithm. It receives the matrix row-wise and computes the low-rank factorization.
\begin{figure} [t!]
\begin{center} \fbox{
\begin{minipage}[l]{6in}
{
\begin{center} \underline{\scshape Private-Covariance-$\lrf$} \end{center}
{\bf Initialization.} Set $\eta=\max \set{k,\alpha^{-1}}$ $t=O(\eta \alpha^{-1} \log (k/\delta)), v=O( \eta\alpha^{-2}   \log(k/\delta))$, and $\rho={\sqrt{(1+\alpha)\ln(1/\delta)}}/{\varepsilon}$. 
	 Sample  $\bN_1\sim \cN(0,\rho^2)^{n \times t}, \bN_2\sim \cN(0,\rho^2)^{v \times n}$.  
	 Sample $\bPhi \in \R^{n \times t}$ from $\cD_R$ as in~\lemref{phi} and $\bS \in \R^{v \times n}$ from $\cD_R$ as in~\lemref{S}. 
	 Initialize an all zero  $n \times t$ matrix $\bY'$ and an all zero $v \times n$ matrix $\bZ'$.

\medskip
\textbf{Update rule.}  Suppose at time $\tau$, the stream is an index-row tuple $(i_\tau, \bA_{i_\tau})$, where $i_\tau\in [m] $ and $ \bA_{i_\tau} \in \R^n$. Let $\bA_\tau$ be a matrix with only non-zero row $\bA_{i_\tau}$ in the row $i_\tau$. Update the matrices by the following rule: $\bY' = \bY' + \bA_\tau^{\mathsf T} \bA_\tau \bPhi $ and  $\bZ' = \bZ' + \bS \bA_\tau^{\mathsf T} \bA_\tau $. 

\medskip
\textbf{Computing the factorization.} Once the matrix is streamed, we follow the following steps.
\begin{enumerate}
	\item  Compute $\bY= \bY' + \bN_1$ and $\bZ=\bZ'+\bN_2$. 
	\item Compute a matrix $\bU \in \R^{m \times t}$ whose columns are an orthonormal basis  for the column space of $\bY$.
	\item Compute the singular value decomposition of $\mathbf{S} \bU \in \R^{v \times t}$. Let it be $\widetilde{\bU} \widetilde{\bSigma} \widetilde{\bV}^{\mathsf T}.$ 
	\item Compute the singular value decomposition of $ \widetilde{\bV} \widetilde{\bSigma}^{\dagger} [\widetilde{\bU}^{\mathsf T}\bZ]_k$. 
	Let it be $\bU' \bSigma' \bV'^{\mathsf T}$. 
	\item Output $\widetilde{\bU}= \bU \bU'$, $\widetilde{\bSigma}= \bSigma'$ and $\widetilde{\bV}=\bV'$. 
\end{enumerate}
}
\end{minipage}
} \caption{Differentially private Covariance Approximation Under $\priv_2$  in Row-wise Update Model} \label{fig:covariance_private}
\end{center}
\end{figure}

 \begin{theorem} \label{thm:streaming_covariance}
 	Let  $m, n \in \N$ and $\alpha,\varepsilon,\delta$ be the input parameters (with $m>n$).  Let $k$ be the desired rank of the factorization and $\eta=\max \set{k,\alpha^{-1}}$. Given a private input matrix $\bA \in \R^{m \times n}$ recieved in a row wise update model, the factorization $\widetilde{\bU}, \widetilde{\bSigma}, \widetilde{\bV}$ outputted by  the algorithm, {\scshape Private-Covariance-$\lrf$}, presented in~\figref{streaming_private}, is a $k$-rank factorization and 
	satisfies the following properties:
	\begin{enumerate}
	\item  {\scshape Private-Covariance-$\lrf$} is $(\varepsilon,\delta)$-differentially private under $\priv_2$. \label{frob_covariance_privacy}
	\item Let $\mathbf{M}_k:=\widetilde{\bU} \widetilde{\bSigma}\widetilde{\bV}^{\mathsf T}$. With probability $9/10$ over the coins of {\scshape Private-Covariance-$\lrf$}, \label{frob_covariance_correctness}
	 \begin{align*}
	   \| \bA^{\mathsf T} \bA - \mathbf{M}_k \|_F \leq (1+\alpha) \| \bA^{\mathsf T} \bA - [\bA^{\mathsf T} \bA]_k \|_F + {O}\paren{  \frac{\sqrt{n \eta \log(1/\delta)}}{\varepsilon \alpha}  } .
	  \end{align*} 
	\item The space used by {\scshape Private-Covariance-$\lrf$} is $O(n\alpha^{-2} \eta   \log(k/\delta))$. \label{covariance_space}
	\end{enumerate}
 \end{theorem}
\begin{proof}
 The space required by the algorithm is the space required to store $\bY$ and $\bZ$, which is $mt + nv = O((m+n\alpha^{-1})k\alpha^{-1} \log k \log(1/\delta))$. This proves part~\ref{covariance_space} of~\thmref{streaming_covariance}. 

Part~\ref{frob_covariance_correctness} follows from~\lemref{naiveN_2},~\lemref{naiveN_1}, and~\thmref{complete}.
For part~\ref{frob_covariance_privacy}, 
first notice that since every row has bounded norm $1$, the sensitivity of the function $\bA^{\mathsf T}\bA$ is at most $1$; i.e., the sensitivity of the vector form of $\bA^{\mathsf T}\bA$ is at most $1$. 
Part~\ref{frob_covariance_privacy} then follows from~\lemref{rho}. 
This completes the proof of~\thmref{streaming_covariance}.
\end{proof}

\section{Case Study 2: Low-rank factorization Under Continual Release Model} \label{app:continual}
In this section, we are interested in computing a low-rank factorization of a private matrix in the continual release model while preserving differential privacy. In this setting, we are allowed only one pass over the private matrix, and at every time epoch, we are required to output a low-rank factorization (see~\ldefref{dpcontinual} for a formal definition).  In~\appref{continual_priv_1}, we give a differentially private low rank factorization under $\priv_1$. We give a differentially private low rank factorization under a much stronger privacy guarantee, $\priv_2$, in~\appref{continual_priv_2}.

In past, there are known algorithms for converting any ``one-shot" algorithm for any monotonic function to an algorithm that continually release the output~\cite{DNPR10}. Since optimization function like low-rank factorization are not monotonic, it is not clear whether we can use the generic transformation. Our algorithm generates and maintains linear sketches during the updates and later compute low-rank factorization using these sketches. This allows us to use the generic transformation to  maintain the updates. For computing the factorization, we collect all the sketches for any range using range queries.

\subsection{Differentially Private Continual Release Low Rank Factorization Under $\priv_2$} \label{app:continual_priv_1}
We start by giving a differentially private algorithm under $\priv_2$ that continually release a low rank factorization. We first give an overview of our algorithm with the details of the algorithm appearing in~\figref{continual}. 

The idea behind our algorithm for continual release is the fact that the factorization stage only uses a small space sketches of the matrix and the sketches are linear sketches. Since the  sketches are linear, we can use the binary tree mechanism~\cite{CSS,DNPR10} to get low-rank factorization under continual release model. 
The algorithm stores the sketches of matrix generated at various time epochs 
 in the form of a binary tree. Every leaf node $\tau$  stores the sketches of $\bA_\tau$, where $\bA_\tau$ is the stream at time $\tau$. The root node stores the sketch of the entire matrix streamed in $[0,T]$, and every other node $\mathsf{n}$ stores the sketch corresponding to the updates in  a time range represented by the leaves of the subtree rooted at $\mathsf{n}$, i.e., $\widehat{\bY}_i$ and $\widehat{\bZ}_i$ stores sketches involving $2^i$ updates to $\bA$. If a query is to compute the low-rank factorization of the matrix from a particular time range $[1,\tau]$, we find the nodes that uniquely cover the time range $[1,\tau]$. We then use the value of $\bY(\tau)$ and $\bZ(\tau)$ formed using those nodes to compute the low-rank factorization. From the binary tree construction, every time epoch appears in exactly $O(\log T)$ nodes (from the leaf to the root node). Moreover, every range $[1,\tau]$ appears in at most $O(\log T)$  nodes of the tree (including leaves and root node). 
\begin{figure} [t]
\begin{center} 
\fbox
{
\begin{minipage}[l]{6in}
{
\begin{center} \underline{\scshape Private-Frobenius-Continual-$\lrf$} \end{center}
 \noindent \textbf{Input:} A time upper bound T , privacy parameters $\varepsilon,\delta$, and a stream $\mathbf{s} \in \R^T$. 

\medskip \noindent \textbf{Output:} At each time step $\tau$, output a factorization $\h{\bU}_k(\tau)$, $\h{\bSigma}_k(\tau)$, and $\h{\bV}_k(\tau)$.

\medskip \noindent \textbf{Initialization:} Set $t,v, \bPhi, \bS$ as in~\figref{private}.  Every $\h{\bY}_i$ and $\h{\bZ}_i$ are initialized to an all zero matrices for  $i \in [\log T]$. Set $\varepsilon' = \varepsilon/\sqrt{\log T}, \delta'=\delta/2 \log T$ and $\rho={\sqrt{(1+\alpha) \ln(1/\delta)}}/{\varepsilon'}$.

\medskip \noindent \textbf{Estimating the $\lrf$ at time $t$.}  
On receiving an input $(\mathsf{r},\mathsf{c},\mathbf{s}_\tau)$ where $\mathbf{s}_\tau  \in \R$ at $1\leq \tau \leq T$, form a matrix $\bA_\tau \in \R^{m \times n}$ which is an all zero matrix except with only non-zero entry $ \mathbf{s}_\tau$ at location $(\mathsf{r},\mathsf{c}) \in [m] \times [n]$. 
\begin{enumerate}
	\item Compute $i := \min \set{j : \tau_j \neq 0}$, where  $\tau =\sum_j \tau_j \cdot 2^j$  is the binary expansion of $\tau$.  
	\item Compute $\h{\bY}_i := \bA_\tau \bPhi + \sum_{j <i} \h{\bY}_j$ and $\h{\bZ}_i := \bS \bA_\tau  + \sum_{j <i} \h{\bZ}_j.$ 
	\item For $j:= 0, \cdots, i-1$, set $ \bY_j = \h{\bY}_j  = \mathbf{0}$ and $\bZ_j= \h{\bZ}_j= \mathbf{0}.$ Compute $ {\bY}_i = \h{\bY}_i + \cN(0,\rho^2)^{ \times n}$ and ${\bZ}_i =\h{ \bZ}_i + \cN(0,\rho^2)^{m \times v}.$ \label{item:update}
 Compute $  {\bY}(\tau) = \sum_{j: \tau_j=1} {\bY}_j $ and $ {\bZ}(\tau) = \sum_{j: \tau_j=1} {\bZ}_j .$ \label{item:estimate}
		\item Compute a matrix $\bU \in \R^{m \times t}$ whose columns are an orthonormal basis  for the column space of $\bY(\tau)$.
	\item Compute the singular value decomposition of $\mathbf{S} \bU \in \R^{v \times t}$. Let it be $\widetilde{\bU} \widetilde{\bSigma} \widetilde{\bV}^{\mathsf T}.$ 
	\item Compute the singular value decomposition of $ \widetilde{\bV} \widetilde{\bSigma}^{\dagger} \widetilde{\bU}^{\mathsf T} [\widetilde{\bU} \widetilde{\bU}^{\mathsf T}\bZ(\tau)]_k$. 
	Let it be $\bU' \bSigma' \bV'^{\mathsf T}$. 
	\item Output $\h{\bU}_k(\tau):=\bU \bU'$, $\h{\bSigma}_k(\tau):= \bSigma'$ and $\h{\bV}_k(\tau):=\bV'$. 
	\item Let $\mathbf{M}_k(\tau) := \bU_k(\tau)  \bSigma_k(\tau) \bV_k(\tau)^{\mathsf T}.$
\end{enumerate}	
}\end{minipage}
} \caption{Differentially private Low-rank Factorization Under Continual Release} \label{fig:continual}
\end{center}
\end{figure}
 A straightforward application of the analysis of Chan {\it et al.}~\cite{CSS} to~\thmref{streaming_naive} gives us the following
\begin{theorem}  \label{thm:continualfrob} Let $\bA$ be an $m \times n$ matrix with $\mathsf{nn}(\bA)$ non-zero entries with $m \leq n$. Let $\eta=\max \set{k,\alpha^{-1}}$. Then there is an $(\varepsilon,\delta)$-differentially private algorithm, {\scshape Private-Frobenius-Continual-$\lrf$} defined in~\figref{continual}, under $\priv_2$ that receives $\bA$ as a stream and outputs a rank-$k$ factorization $\widetilde{\bU}:=\h{\bU}_k(\tau), \widetilde{\bSigma} := \h{\bSigma}_k(\tau), \widetilde{\bV}:=\h{\bV}_k(\tau)$ under the continual release for  $T$ time epochs such that, with probability $9/10$,
	 \begin{align*}
  \|  \bA  - \mathbf{M}_k \|_F \leq (1+\alpha) \| \bA - [\bA]_k \|_F  +  {O}\paren{\paren{ \sqrt{km} + \sqrt{ \frac{n \eta}{\alpha^2} } }  \frac{\sqrt{ \log(1/\delta)} \log T}{\varepsilon}  } 
   \end{align*} 
where $\mathbf{M}_k=\widetilde{\bU} \widetilde{\bSigma} \widetilde{\bV}^{\mathsf T}$ and $\bA(\tau)$ is the matrix received till time $\tau$.
\end{theorem}

\subsection{Differentially Private Continual Release Low Rank Factorization Under $\priv_1$} \label{app:continual_priv_2}
We can also convert the algorithm {\scshape Private-Space-Optimal-$\lrf$} to one that outputs a low-rank factorization under continual release by using less space than {\scshape Private-Continual-Frobenius-$\lrf$} and secure under $\priv_2$. We make the following changes to {\scshape Private-Continual-Frobenius-$\lrf$}: 
(i) Initialize $(\h{\bY}_c)_i, (\h{\bY}_r)_i,$ and $(\h{\bZ})_i$ as we initialize $\bY_c, \h{\bY}_r$ and $\h{\bZ}$ in~\figref{streaming_spacelowprivate}  for all $i \in [\log T]$,
(ii) we maintain $ (\bY_c)_j, (\h{\bY_c})_j $, $ (\bY_r)_j, (\h{\bY_r})_j$,  $\bZ_j,$ and $\h{\bZ}_j.$ 
A  straightforward application of the analysis of Chan {\it et al.}~\cite{CSS}  to~\thmref{low_space_private} gives us the following theorem.
\begin{theorem}  \label{thm:continuallowspace}Let $\bA$ be an $m \times n$ matrix with $\mathsf{nn}(\bA)$ non-zero entries with $m \leq n$. Let $\eta=\max \set{k,\alpha^{-1}}$.  Let $s=\max\{m,n\}, u=\min \set{m,n}$, $\kappa=(1+\alpha)/(1-\alpha)$,  and $\sigma_\mathsf{min}={16 \log(1/\delta) \sqrt{t \kappa \ln(1/\delta)}}/{\varepsilon}$.
Then there is an $(\varepsilon,\delta)$-differentially private algorithm under $\priv_1$ that receives $\bA$ as a stream and outputs a rank-$k$ factorization $\widetilde{\bU}:=\h{\bU}_k(\tau), \widetilde{\bSigma} := \h{\bSigma}_k(\tau), \widetilde{\bV}:=\h{\bU}_k(\tau)$ under the continual release for  $T$ time epochs such that, with probability $9/10$,
	 \begin{align*}
  \| \begin{pmatrix} \bA(\tau) & \mathbf{0} \end{pmatrix} - \mathbf{M}_k \|_F \leq (1+\alpha) \| \bA(\tau) - [\bA(\tau)]_k \|_F  + O \paren{ \paren{ \sigma_\mathsf{min} \sqrt{u} + \varepsilon^{-1}  {\sqrt{ks  \ln(1/\delta)}} } \log T },
   \end{align*} 
   where $\mathbf{M}_k=\widetilde{\bU} \widetilde{\bSigma} \widetilde{\bV}^{\mathsf T}$ and $\bA(\tau)$ is the matrix received till time $\tau$.
\end{theorem}


\section{Space Lower Bound for Low-rank Factorization When $\gamma \neq 0$} \label{app:lower}
This section is devoted to proving a lower bound on the space requirement for low-rank factorization with non-trivial additive error. It is well known that any private algorithm (not necessarily differentially private) incurs an additive error $o(\sqrt{k(m+n)})$~\cite{HR12} due to linear reconstruction attack. 
On the other hand, the only known space lower bound of Clarkson and Woodruff~\cite{CW09} holds when $\gamma =0$; therefore, one might hope to construct an improve space algorithm when we allow $\gamma \neq 0$. 
In this section, we  show that for any non-trivial values of $\gamma$, this is not the case. This directly implies that  our algorithm uses optimal space for a large range of parameters. 
\begin{theorem} \label{thm:lower}
Let $m,n,k \in \N$ and $\alpha >0$. Then the space used by any randomized single-pass algorithm for $(\alpha,5/6,O(m+n),k)$-$\lra$ in the general turnstile model is at least $\Omega((m+n)k/\alpha)$.
\end{theorem}
\thmref{lower} shows that {\scshape Private-Space-Optimal}-$\lrf$ uses optimal space when $\gamma=O(m+n)$ and $k \geq 1/\alpha$. 
If we set  $\alpha=\sqrt{2}-1$ (as in Hardt and Roth~\cite{HR12}) and note any non-trivial result implies that $\gamma = o(m+n)$, we have a matching lower bound for all $k \geq 3$. 

The space lower bound in the turnstile update model is shown by showing that any algorithm $\mathsf{Alg}$ in the turnstile model yields a single round communication protocol for some function $f$. The idea is as follows. On input $\bx$, Alice invokes $\mathsf{Alg}$ on its input to compute $\mathsf{Alg}(x)$. She then sends the state $\mathsf{st}$ to Bob, who computes $\mathsf{Alg}(\bx \| \by)$ using his input $\by$ and $\mathsf{st}$, and uses this to compute the function $f$. The communication is therefore the same as  the space required by the algorithm. In what follows, we use the notation $\bC_{:i}$ to denote the $i$-th column of the matrix $\bC$. 

We give a reduction to the augmented indexing problem, $\aind$. It is defined as follows.

\begin{definition} ($\aind$ problem). Alice is given an $N$-bit string $\bx$ and Bob is given an index $\ind \in [N]$ together with $\bx_{\ind+1}, \cdots, \bx_N$. The goal of Bob is to output $\bx_{\ind}$. 
\end{definition}

The communication complexity for solving $\aind$ is well known due to the result of Miltersen {\it et al.}~\cite{MNSW98}.
\begin{theorem} \label{thm:aind}
The minimum  bits of communication required to solve $\aind$ with probability $2/3$, when the message is sent only in one direction, i.e., either from Alice to Bob or from Bob to Alice, is  $\Omega(n)$.  This lower bound holds  even if the index, $\ind$, and the string, $\bx$, is chosen uniformly at random.
\end{theorem}
 
Before we state our result and its proof, we fix a notation. For a matrix $\bA$ and set of indices $C$, we use the notation $\bA(C)$ to denote the submatrix formed by the columns indexed by $C$.

  \begin{proof}[Proof of~\thmref{lower}.]
We adapt the proof of Clarkson and Woodruff~\cite{CW09} for the case when $\gamma \neq 0$. Suppose $m \geq n$ and let $a=k/20 \alpha$. Without loss of generality, we can assume that $a$ is at most $n/2$. Let $\ell$ be the word size. We assume Alice has  a string $\bx \in \set{-1,+1}^{(m-a)a}$ and Bob has an index $\ind \in [(m-a)a]$. The idea is  to define the matrix $\bA$ with high Frobenius norm. The matrix $\bA$ is the summation of the matrix $\widetilde{\bA}$ constructed by Alice and $\bar{\bA}$ constructed by Bob. We first define how Alice and Bob construct the instant $\bA= \widetilde{\bA}+\bar{\bA}$.

 Alice constructs its matrix $\widetilde{\bA}$ as follows. Alice partitions the set $\set{1,\cdots, a}$ in to $\ell$ disjoint sets $I_1, \cdots, I_\ell$ such that $I_i:= \set{(i-1)a/\ell +1, \cdots ia/\ell}.$
Let $\bM \paren{I_i}$ be an $(m-a) \times a/\ell$ matrix for all $1 \leq i \leq \ell$. We form a bijection between entries of $\bx$ and the entries of $\bM$ in the following manner. Every entry of $\bM\paren{I_i}$ is defined by a unique bit of $\bx$, i.e.,  $\bM\paren{I_i}_{j,k} = (-1)^{\bx_{d}} (10)^i$ for $d=(i-1)(m-a)a/\ell + (k-1)(m-a) + j$. The matrix $\widetilde{\bA}$ is now defined as follows.
\[ \widetilde{\bA} = \begin{pmatrix}  \mathbf{0}^{a \times a} &  \mathbf{0}^{a \times (n-a)} \\  \mathbf{M} &  \mathbf{0}^{(m-a) \times (n-a)} \end{pmatrix}, \]
$\text{where}~\mathbf{M} = \begin{pmatrix} \bM_{I_1} & \cdots & \bM_{I_\ell} \end{pmatrix}$.

Suppose Bob is given an index $\ind \in [(m-a)a]$ such that $\bx_\ind$ corresponds to the sub-matrix $\bM\paren{I_\theta}$ for some $1 \leq \theta \leq \ell$. Then we can assume that Bob also knows every entry in the sub-matrix $\bM\paren{I_{\theta'}}$ for $\theta' >\theta$.  Bob forms a second level partition of the columns of $\bM\paren{I_\theta}$ in to equal size groups $G_1, \cdots, G_{a/k \ell}$. Due to our construction, there exists a unique $r$ such that $\bx_\ind$ maps to an entry in the sub-matrix formed by columns indexed by one of the second level partition $G_r$. Let  $C = \set{c, c+1, \cdots, c+k-1}$ be the columns corresponding to the $k$-size group of $I_\theta$ in which $\ind$ is present. As its input, Bob streams a matrix $\bar{\bA}$ which is an all-zero matrix, except for entries $\bar{\bA}_{c+i,c+i} = \zeta$ for $0 \leq i \leq k-1$ and $\zeta$ to be chosen later. In other words, Bob inserts a scaled identity matrix in the stream, where the scaling parameter $\zeta$ is large enough to make sure that most of the error of any randomized algorithm is due to other columns of $\bA$. As we shall see later, we set the value of $\zeta$ as a large polynomial in the approximation error of the algorithm. 

Let $\cA$ be the algorithm that computes $\lra$ under the turnstile model. Alice feeds its matrix $\widetilde{\bA}$ to $\cA$ in the turnstile manner and send the state of the algorithm by the end of her feed to Bob. Bob  uses the state received by Alice and feed the algorithm $\cA$ with its own matrix $\bar{\bA}$ in a turnstile manner. Therefore, the  algorithm $\cA$ gets as input a matrix $\bA = \widetilde{\bA} + \bar{\bA}$ and it is required to output a rank-$k$ matrix $\bB$ with additive error $\gamma=O(m+n)$. We will show that any such   output allows us to solve $\aind$. Denote by ${\bA}(C)$ the sub-matrix formed by the columns $C:=\set{c, c+1, \cdots, c+k-1}$.

Let us first understand the properties of the constructed matrix $\bA$. To compute the Frobenius norm of this matrix, we need to consider two cases: the case for sub-matrices in which $\ind$ belongs, i.e, $\bM\paren{I_r}$, and the rest of the matrix. For the sub-matrix corresponding to the columns indexed by ${C}$, the columns of $\bA\paren{I_\theta}$ have Euclidean length $(\zeta^2 + (m-a)100^\theta)^{1/2}$. 
For $\theta' <\theta$, every columns have Euclidean norm  $(a (m-a))^{1/2}10^{\theta'}$. 
Therefore, we have the following:
\begin{align*}
	\| {\bA} - [{\bA}]_k \|_F^2 &\leq \frac{((a-k)(m-a)100^\theta}{\ell} + \sum_{\theta' <\theta} \frac{a (m-a)100^{\theta'}}{\ell} \\
		& \leq \frac{((a-k)(m-a)100^\theta}{\ell} +  \frac{a (m-a)100^\theta}{99 \ell} \\
		&\leq  2 \cdot (100)^\theta m^2/\ell = \Gamma
\end{align*}

In order to solve $(\alpha,\beta,\gamma,k)$-$\lrf$, the algorithm needs to output a matrix $\bB$ of rank at most $k$ such that, with probability $5/6$ over its random coins, 
\begin{align*}
	\| {\bA} - \bB \|_F^2 &\leq \sparen{ (1+\alpha) \sqrt{\Gamma} + \gamma }^2 \leq 2(1+\alpha) \Gamma + 2 \gamma^2 \\
	& \leq 2\Gamma + 100^\theta k(m-a) \paren{\frac{1}{10} + \frac{1}{99}}  + 2\gamma^2\\
		&\leq 4 \cdot (100)^\theta m^2/\ell  + \frac{100^\theta k(m-a) }{5}  + 2\gamma^2
\end{align*}

Let us denote by $\Upsilon:=4 \cdot (100)^\theta m^2/\ell + 100^\theta k(m-a) \paren{\frac{1}{10} + \frac{1}{99}}+ 2\gamma^2$. 
The proof idea is now to show the following:
\begin{description}
	\item [(i)] Columns of $\bB$ corresponding to index set in $C$ are linearly independent. 
	\item [(ii)] Bound the error incurred by $\| {\bA} - \bB \|_F$ in terms of the columns indexed by $G_r$.
\end{description}

The idea is to show that most of the error is due to the other columns in $\bB$; and therefore, sign in the submatrix $\bA(C)$ agrees with that of the signs of those  in the submatrix $\bB(C)$. This allows Bob to solve the $\aind$ problem as Bob can just output the sign of the corresponding position.

Let $$R:=\set{ra/k+1, \cdots, (r+1)a/k}$$ and $$C:=\set{c,\cdots, c+k-1}.$$ Let $\bY$ be the submatrix of $\bB$ formed by the rows indexed by $R$ and columns indexed by $C$.

The following lemma proves that when $\zeta$ is large enough, then the columns of $\bB$ corresponding to index set $C$ are linearly independent. This proves part~(i) of our proof idea.
\begin{lemma} \label{lem:independent}
Let $\bB(C):= [\begin{matrix} \bB_{:c} & \cdots \bB_{:c+k-1} \end{matrix}]$ be the columns corresponding to the sub-matrix formed by columns $c, \cdots, c+k-1$ of $\bB$. If $\zeta \geq 2\Upsilon^2$, then the columns of $\bB(C)$ spans the column space of $[\bA]_k$.
\end{lemma}
\begin{proof}
We will prove the lemma by considering the $k \times k$ sub-matrix, say $\bY$. Recall that $\bY$ is a submatrix  of $\bB$ formed by the rows indexed by $R$ and the columns indexed by $C$. For the sake of brevity and abuse of notation, let us denote the restriction of $\bB$ to this sub-matrix $\bY:=[\bY_{:1}, \cdots , \bY_{:k}]$. In what follows, we prove a stronger claim that the submatrix $\bY$ is a rank-$k$ matrix. 

Suppose, for the sake of contradiction that the vectors $\set{\bY_{:1}, \cdots , \bY_{:k}}$ are linearly dependent. In other words, there exists a vector $\bY_{:i}$ and real numbers $a_1, \cdots, a_k$, not all of which are identically zero, such that 
\[ \bY_{:i} = \sum_{j=1, j \neq i}^k a_j \bY_{:j}. \]

From the construction, since Bob inserts a sub-matrix $\zeta \I_k$, we know that 
\begin{align}
	\sum_{j=1}^k (\bY_{j,j} - \zeta)^2 &\leq \| \bA - \bB \|_F^2 \leq \Upsilon . \label{eq:equal} \\
	\sum_{j=1}^k \sum_{p \neq j} \bY_{p,j}^2 &\leq \| \bA - \bB \|_F^2 \leq \Upsilon. \label{eq:neq} 	
\end{align}
From~\eqnref{equal} and choice of $\zeta$,  for all $j$, we have $\bY_{j,j} \geq \Upsilon^2$. Further,~\eqnref{neq} implies that $\bY_{p,j} \leq \sqrt{\Upsilon} .$  We have 
\[ \bY_{i,i} = \sum_{j=1, j \neq i}^k a_j \bY_{i,j} \geq \Upsilon^2  
\]
{imply that there is an $p \in \set{1,\cdots, k} \backslash \set{i}$ such that}~ $|a_{p}| \geq  \frac{\Upsilon^2 }{k\sqrt{\Upsilon}}.$ 

Let  $\tilde{i}$ be the index in $\set{1,\cdots, k} \backslash \set{i}$ for which $|a_{\tilde{i}} |$ attains the maximum value. We have $|a_{\widetilde{i}}\bY_{\tilde{i},\tilde{i}}| \geq |a_{\widetilde{i}}| \Upsilon^2$ and  $|a_j \bY_{\tilde{i},j}| \leq |a_{\widetilde{i}}| \sqrt{\Upsilon}$.
Now consider the $\tilde{i}$-entry of $\bY_{:i}$. Note that $\tilde{i} \neq i$. Since $\Upsilon$ depends quadratically on $m$ and $\gamma$, we have
\[ \left|  \sum_{j=1, j \neq i}^k a_j \bY_{\tilde{i},j}  \right|  \geq | a| ( \Upsilon^2 - k \sqrt{\Upsilon}  ) \geq  ( \Upsilon^2 - k \sqrt{\Upsilon}  ) \frac{\Upsilon^2 }{k\sqrt{\Upsilon}} > \sqrt{\Upsilon} .\] This is a contradiction because $\bY_{p,j} \leq \sqrt{\Upsilon} $ due to~\eqnref{neq} for $p \neq j$. This completes the proof.
\end{proof}

For the sake of brevity, let $\bV_{:1}, \cdots, \bV_{:k} $ be the columns of $\bB(C)$ and $\widetilde{\bV}_{:1}, \cdots, \widetilde{\bV}_{:k}$ be the restriction of these column vectors to the rows $a+1, \cdots, m$. In other words, vectors $\widetilde{\bV}_{:1}, \cdots, \widetilde{\bV}_{:k}$ are the column vectors corresponding to the columns in $\mathbf{M}$. We showed in~\lemref{independent} that the columns  $\bB(C)$ spans the column space of $\bB$. We can assume that the last $n-a$ columns of $\bB$ are all zero vectors because $\bB$ is a rank-$k$ matrix. We can also assume without any loss of generality that, except for the entries in the row indexed by $R$, all the other entries of $\bB(C)$ are zero. This is because we have shown in~\lemref{independent}, we showed that the submatrix of $\bB(C)$ formed by rows indexed by $R$ and columns indexed by $C$ have rank $k$. 

Now any row $i$ of $\bB$ can be therefore represented as $\sum \eta_{i,j} \bV_{:j}$, for real numbers $\eta_{i,j}$, not all of which are identically zero. The following lemma proves part~(ii) of our proof idea. For

\begin{lemma} \label{lem:a}
 Let $\bV_{:1}, \cdots, \bV_{:k} $ be as defined above. Then  column $i$ of $\bB$ can be written as linear combination of real numbers $\eta_{i,1}, \cdots \eta_{i,k}$ of the vectors $\bV_{:1}, \cdots, \bV_{:k} $ such that, for all $j$ and $i \in R$, $ \eta_{i,j}^2 \leq 4/\Upsilon^3 $.
\end{lemma}
\begin{proof}
Let $\bM_{:1}, \cdots \bM_{:a}$ be the columns of $\bM$, where $\bM$ is the $(m-a)\times a$ submatrix of the matrix $\widetilde{\bA}$ corresponding to the input of Alice. We have
\begin{align*}
\Upsilon & \geq \| \bA - \bB \|_F^2  \sum_{i=1}^k (\zeta - \bV_{r(a/k)+i,i})^2 + \sum_{i=1}^k \sum_{j \neq i} \bV_{r(a/k)+i,j}^2 + \sum_{i=1}^k \| \bM_{:r(a/k) + i} - \widetilde{\bV}_{:i} \|^2  \nonumber \\
	&\quad + \sum_{i \notin R} \sum_{j=1}^k \paren{ \eta_{i,j} \bV_{ra/k + j,j} + \sum_{j' \neq j} \eta_{i,j'} \bV_{ra/k+j,j'}  }^2 + \sum_{i \notin R} \left\| \bM_{:i} - \sum_{j=1}^k \eta_{i,j} \widetilde{\bV}_{:j} \right\|^2.
\end{align*}
As in the proof of~\lemref{independent}, we have $|\bV_{r(a/k)+i,j}^2 | \leq \sqrt{\Upsilon}$ and $|\bV_{r(a/k)+i,i}| \geq \Upsilon^2$. Let $ {j}_i$ be the index such that $|\eta_{i, j_i} |$ is the maximum. Then the above expression is at least $| \eta_{i,j_i}|^2 (\Upsilon^2  - k \sqrt{\Upsilon} )^2 \geq | \eta_{i,j_i}|^2 \Upsilon^4 /4$. Since this is less than $\Upsilon $, the result follows from the definition of $j_i$.
 \end{proof}

We can now complete the proof. First note that since $\bM$ is a signed matrix, each $\widetilde{\bV}_i$ in the third term of the above expression is at least $\sqrt{\Upsilon}$. Therefore, for all $i \notin S$ and all $j$ 
$$ \left| \sum_{j=1}^k \eta_{i,j} \widetilde{\bV}_{:j} \right| \leq \frac{4k\Upsilon^{1/2} }{\Upsilon^{3/2}} = \frac{4k}{\Upsilon}. $$ 

As $\bM_{:i}$ is a sign vector and if $\gamma=O(m+n) = O(m)$, this implies that 
\begin{align*}
 \sum_{i \notin R} \left\| \bM_{:i} - \sum_{j=1}^k \eta_{i,j} \widetilde{\bV}_{:j} \right\|^2 &\geq \sum_{i \notin R} \| \bM_{:i} \|^2 \paren{1 - \frac{4k}{\Upsilon} }  \geq O(  (100)^\theta m^2 /\ell) - O(100^\theta a) \\
\sum_{i=1}^k \left\| \bM_{:r(a/k) + i} - \widetilde{\bV}_{:i} \right\|^2 &= \sum_{i=1}^k \sum_{j=1}^{m-a} ( \bM_{j,r(a/k) + i} - (\widetilde{\bV}_i)_j )^2\leq  \frac{100^\theta k(m-a)}{5}  + O(100^\theta a) 
 \end{align*}
 
 Now, since there are in total $k(m-a)$ entries in the submatrix formed by the columns indexed by $C$, at least $1- \paren{\frac{1}{10} + \frac{1}{99} +o(1)}$ fraction of the entries have the property that the sign of $\bM_{j,ra/k+i}$ matches the sign of $\widetilde{\bV}_{j,i}$. Since $\ind$ is in one of the columns of $\bM_{:ra/k+1}, \cdots \bM_{:ra/k+k}$, with probability at least $1- \paren{\frac{1}{10} + \frac{1}{99} +o(1)}$, if Bob outputs the sign of the corresponding entry in $\bB$, then Bob succeeds in solving $\aind$. This gives a lower bound of $\Omega((m-a)a) =\Omega(mk \ell/\alpha)$ space. The case when $m \leq n$ is analogous and gives a lower bound of $\Omega(nk \ell/\alpha)$. Thus, there is a lower bound of $\Omega((m+n)k \ell/\alpha)$.
\end{proof}


\section{Noninteractive Local Differentially Private PCA} \label{app:local}
In this section, we give our noninteractive local differentially private principal component analysis (see~\ldefref{dppca}). 
In the \emph{local} model \cite{EGS03}, each individual applies a differentially
private algorithm locally to their data and shares only the output of
the algorithm---called a report---with a server that
aggregates users' reports. 
In principle, one could
also use cryptographic techniques such as secure function evaluation
to simulate central model algorithms in a local model, but such
algorithms currently impose bandwidth and liveness constraints that
make them impractical for large deployments.
  A long line of work studies what is achievable by local differentially private algorithms~\cite{agrawal2009frapp,BS15,DJW13,rappor,EGS03,HKR,KLNRS11,MS06,W65}.
Tight
upper and lower bounds known on the achievable accuracy for many
problems; however, low-rank factorization (and even low-rank approximation) has not been studied in this model. The naive approach to convert existing algorithms to locally private algorithms leads to a large additive error and are interactive. On the other hand, low-rank factorization is a special optimization problem and the role of interaction in local differentially private optimization was recently investigated by Smith {\it et al.}~\cite{STU16}. 

\subsection{Noninteractive Local Differentially Private $\lrf$ Under $\priv_2$} 

\begin{theorem} \label{thm:local}
 	Let  $m, n \in \N$ and $\alpha,\varepsilon,\delta$ be the input parameters.  Let $k$ be the desired rank of the factorization and $\eta=\max \set{k,\alpha^{-1}}$. Let $t=O(\eta \alpha^{-1}  \log (k/\delta))$ and $v=O(\eta \alpha^{-2} \log(k/\delta))$. Given a private input matrix $\bA \in \R^{m \times n}$ distributed in a row-wise manner amongst $m$ users, the output ${\bU}$ of  the algorithm, {\scshape Private-Local-$\lrf$}, presented in~\figref{local_spacelow}, is a $k$-rank orthonormal matrix such that
	\begin{enumerate}
	\item  {\scshape Private-Local-$\lrf$} is a non-interactive $(\varepsilon,\delta)$-local differentially private under $\priv_2$. \label{local_privacy}
	\item  With probability $9/10$ over the coins of {\scshape Private-Local-$\lrf$}, \label{local_utility}
	 \begin{align*}
	   \| \bA - \bU \bU^{\mathsf T}\bA \|_F \leq (1+ O(\alpha)) \| \bA - [\bA]_k \|_F + {O}\paren{ v  \sqrt{m \log (1/\delta)}/\eps } .
	  \end{align*} 
	\item The words of communication used by every users in {\scshape Private-Local-$\lrf$} is $O(v^2)$ words. \label{local_bits}
	\end{enumerate}
 \end{theorem}

 \begin{figure}[t!]
 \begin{center}
\fbox{
\begin{minipage}[l]{6in}
{
\begin{center} \underline{\scshape Private-Local-$\lrf$} \end{center}
{\bf Initialization.} Let $\eta=\max\set{k,\alpha^{-1}}$, $t=O(\eta \alpha^{-1}  \log (k/\delta)), v=O(\eta \alpha^{-2} \log(k/\delta))$. Let  $\rho_1={\sqrt{(1+\alpha)\ln(1/\delta)}}/{\varepsilon},~\rho_2={(1+\alpha)\sqrt{\ln(1/\delta)}}/{\varepsilon}$. 
	 Sample $\bPhi \sim \cN(0,1)^{n \times t}$, $\bPsi \sim \cN(0,1)^{t \times m}$,  $\bS \sim \cN(0,1)^{v \times m}$,  and $\bT \sim \cN(0,1)^{n \times v}$. Make them public.

\medskip \noindent \textbf{User-$i$ computation.} On input the row $\bA_{i:}$, user-$i$ does the following:
\begin{enumerate}
	 \item Sample  $\bN_{1,i} \sim \cN(0,\rho_1^2)^{1 \times t}$, $\bN_{2,i} \sim \cN(0,\rho_2^2)^{t \times v}$ and $\bN_{3,i} \sim \cN(0,\rho_2^2)^{v \times v}$. 
	\item Set $\widehat{\bA}_{i:} \in \R^{m \times n}$ such that every row other than row-$i$ is an all zero vector. 
	Compute $\bY_{i:} = {\bA}_{i:} \bPhi + \bN_{1,i}$, $\widetilde{\bY}_{i:} = \bPsi \widehat{\bA}_{i:} \bT + \bN_{2,i} $, and $\bZ_{i:} = \bS \widehat{\bA}_{i:} \bT + \bN_{3,i}.$
\end{enumerate}

\medskip \noindent \textbf{Server side computation.} Once the server receives the reports from all the users, it follows the following steps.
\begin{enumerate}
	\item Form $\bY$ whose row-$i$ is $\bY_{i:}$. Compute $\bZ = \sum \bZ_{i:}$  and $\widetilde \bY  = \sum \widetilde{\bY}_{i:}$. Compute $\widehat \bY = \bS \bY$.
	\item  Compute $\widetilde \bX:=\argmin_{\mathsf{rk}(\bX)\leq k} \| \widehat \bY   \bX  \widetilde \bY  - \bZ \|_F$. 
	Compute a SVD of $\widetilde \bX$. Let it is be $\bU' \bSigma' \bV'^{\mathsf T}$. 
	\item Output the orthonormal basis $\bU$ for the span of $\bY \bU'$. 
\end{enumerate}

}
\end{minipage}
} \caption{Non-interactive Local Differentially private $\lrf$ Under $\priv_1$} \label{fig:local_spacelow}
\end{center}
\end{figure}

\begin{proof}
The local privacy is easy to follow from the Gaussian mechanism and as in~\lemref{low_space_private} with the choice of $\rho_1$ and $\rho_2$.  For the communication cost, note that every user $i$ has to send a sketch $\bY_{i:}$, $\widehat \bY_{i:}$, and $\bZ_{i:}$. The sketch $\bY_{i:}$ is a real $1 \times t$ matrix, $\widehat \bY_{i:}$ is an $t \times v$ real matrix, and $\bZ_{i:}$ is a $v \times v$ real matrix. 
The total communication cost is $O((tv +v^2) \log (nm))$ words. Since $t \leq v$, the result on the communication cost follows.

We now prove \textcolor{red}{Part}~\ref{local_utility} of~\thmref{local}. Let $\bN_1$ be a random Gaussian matrices whose row-$i$ is $\bN_{1,i}$. Let $\bN_2 = \sum \bN_{2,i}, \bN_3 = \sum \bN_{3,i}$. Note that $\bN_1 \sim \cN(0,\rho_1^2)^{m \times t}$, $\bN_2 \sim \cN(0,m\rho_2^2)^{t \times v}$, and $\bN_2 \sim \cN(0,m\rho_2^2)^{v \times v}$. Let $\bY$  be the matrix whose row-$i$ is $\bY_{i:}$. Further, $\bZ = \sum \bZ_{i:}$  and $\widetilde \bY  = \sum \widetilde{\bY}_{i:}$. If the matrix distributed among the users is $\bA$, then it means that $\bY = \bPhi \bA + \bN_1$, $\widetilde \bY = \bPsi \bA \bT + \bN_3$ and $\bZ = \bS \bA \bT + \bN_2$.

 Let the singular value decomposition of $[\bA]_k$ be $[\bA]_k = \bU_k \bSigma_k \bV_k^\mathsf T$. Let $\bC= \bPsi( \bA  + \bPsi^\dagger \bN_2 \bT^\dagger )$. We will use~\lemref{phi} to relate  $\min_{\mathsf{rk}(\bX)\leq k} \| \bY \bX  \bC  -(\bA  + \bS^\dagger \bN_3 \bT^\dagger) \|_F$ with $\| \bA - [\bA]_k \|_F$. Set $\bPhi= \bPsi$, $\mathbf{P}=[\bA]_k$, $\mathbf{Q}=\bA  +  \bPsi^\dagger\bN_2 \bT^\dagger$ in~\lemref{phi}.
For 
$$\widetilde{\bX}: = (\bPsi [\bA]_k)^\dagger (\bPsi \bA  + \bN_2 \bT^\dagger) = (\bPsi [\bA]_k)^\dagger \bC =\argmin_{\bX} \|\bPsi ([\bA]_k \bX -  \bA + (\bPsi^\dagger\bN_2 \bT^\dagger)) \|,$$ 
we have with probability $1-\delta$ over $\bPsi \sim \cD_R$,
  \begin{align} 
  \|  [\bA]_k \widetilde{\bX} - ( \bA +  \bPsi^\dagger\bN_2 \bT^\dagger) \|_F &\leq (1+\alpha) \min_{\bX} \| [\bA]_k \bX - ( \bA +  \bPsi^\dagger\bN_2 \bT^\dagger) \|_F  \nonumber \\
  		&\leq (1+\alpha)  \| [\bA]_k - \bA\|_F + (1+\alpha)\|  \bPsi^\dagger\bN_2 \bT^\dagger \|_F \label{eq:phipsi2}. 
  \end{align}
  In the above, the second inequality follows by setting $\bX = \bV_k  \bV_k^\mathsf T$. 
 
 Let $\mathbf{W}^\mathsf T := \widetilde \bX = (\bPsi [\bA]_k)^\dagger \bC$. 
 We now use~\lemref{phi} on the following regression problem: 
  \[\min_\bX \| \bPhi^{\mathsf T} (\mathbf{W} \bX - \bB)  \| \quad \text{and} \quad \min_{\bX} \|  \mathbf{W} \bX - \bB \|_F, \quad \text{where}~\bB = (\bA + \bN_1 \bPhi^\dagger)^{\mathsf T} \]
  with the candidate solutions
  \[ \widehat{\bX} = \argmin_{\bX} \| \bPhi^{\mathsf T} (\mathbf{W} \bX - \bA) \|_F \quad \text{and} \quad \widetilde \bX = \argmin_{\bX} \| (\mathbf{W} \bX - \bA) \|_F \]

  One of the candidate solutions to $\argmin_{\bX} \| \bPhi^{\mathsf T} (\mathbf{W} \bX - \bA) \|_F$ is $\widehat{\bX} := (\bPhi^{\mathsf T} \mathbf{W})^\dagger (\bPhi^{\mathsf T}\bB)$.  Since $[\bA]_k$ has rank $k$,~\lemref{phi} and~\eqnref{phipsi2} gives with probability $1-\delta$ over $\bPsi \sim \cD_R$ 
  \begin{align*}
  	\| \widehat \bX^\mathsf T \mathbf{W}^{\mathsf T} - \bB^{\mathsf T}\|_F 	&\leq (1+\alpha)\min_\bX \| \bX^{\mathsf T} \mathbf{W}^{\mathsf T} - \bB^{\mathsf T} \|_F 
	\leq (1+\alpha) \| [\bA]_k (\bPsi [\bA]_k)^\dagger \bC  - \bB^{\mathsf T} \|_F \\
	& \leq (1+\alpha) \| [\bA]_k (\bPsi [\bA]_k)^\dagger \bC  - \bA  \|_F +(1+\alpha) \| \bN_1 \bPhi^\dagger \|_F \\
	&\leq (1+\alpha)^2 \| \bA - [\bA]_k \|_F + (1+\alpha)(2+\alpha) \| \bPsi^\dagger \bN_2 \bT^\dagger \|_F +(1+\alpha) \| \bN_1 \bPhi^\dagger \|_F 
	\intertext{This in particular implies that}
 \| \widehat \bX^\mathsf T \mathbf{W}^{\mathsf T} - \bA\|_F &\leq (1+\alpha)^2 \| \bA - [\bA]_k \|_F + (1 + \alpha)(2+\alpha) \| \bPsi^\dagger \bN_2 \bT^\dagger\|_F +(2+\alpha) \| \bN_1 \bPhi^\dagger \|_F . 
   \end{align*}

Let 
\[  \tau_1 = (1 + \alpha)(2+\alpha) \| \bPsi^\dagger \bN_2 \bT^\dagger\|_F +(2+\alpha) \| \bN_1 \bPhi^\dagger \|_F \]
be the additive error due to the effect of noise $\bN_1$ and $\bN_2$. 

 {Substituting the value of $\widehat{\bX}^{\mathsf T} := (\bB^{\mathsf T} \bPhi) (\mathbf{W}^{\mathsf T} \bPhi)^\dagger  = (\bA \bPhi + \bN_1)  (\mathbf{W}^{\mathsf T} \bPhi)^\dagger = \bY  (\mathbf{W}^{\mathsf T} \bPhi)^\dagger$, with probability $1-2\delta$ over $\bPhi^{\mathsf T}, \bPsi \sim \cD_R$, we have}
  \begin{align*}
  \|  \bY  (\mathbf{W}^{\mathsf T} \bPhi)^\dagger (\bPsi [\bA]_k)^\dagger  \bC  - \bA \|_F \leq (1+\alpha)^2 \| \bA - [\bA]_k \|_F + \tau_1.  
  \end{align*}
 
Let $ \bX_*:=  (\mathbf{W}^{\mathsf T} \bPhi)^\dagger (\bPsi [\bA]_k)^\dagger$, i.e.,
\[ 
  \|  \bY  \bX_*  \bC  - \bA \|_F \leq (1+\alpha)^2 \| \bA - [\bA]_k \|_F + \tau_1.  
\]
Let $\mathbf{E} =\bA  + \bS^\dagger \bN_3 \bT^\dagger$. Since $\bX_*$ has rank at most $k$, this implies that  
\begin{align}
 	 \min_{\bX \atop \mathsf{rk}(\bX)\leq k} \| (\bY \bX  \bC  - \mathbf{E}) \|_F & \leq  \| (\bY \bX_* \bC  - \mathbf{E}) \|_F \nonumber \\
	  & \leq  \| (\bY \bX_* \bC  - \bA) \|_F  +  \|  \bS^\dagger \bN_3 \bT^\dagger  \|_F \nonumber \\
	 &\leq (1+\alpha)^2 \| \bA - [\bA]_k \|_F + \tau_1 +  \|  \bS^\dagger \bN_3 \bT^\dagger  \|_F.  \nonumber
	 \end{align}

Since $\alpha \in (0,1)$ and substituting the value of $\tau_1$, we can get an upper bound on the additive terms.
\begin{align}
 	 \min_{\bX \atop \mathsf{rk}(\bX)\leq k} \| (\bY \bX  \bC  - \mathbf{E}) \|_F & \leq 
		(1+\alpha)^2 \| \bA - [\bA]_k \|_F + O \paren{ \tau_1 +  \|  \bS^\dagger \bN_3 \bT^\dagger  \|_F}. \label{eq:YXC}
\end{align}
	 

Now consider the following two regression problems:
\begin{align}
&	\min_{\bX \atop \mathsf{rk}(\bX)\leq k} \|\bS \bY   \bX \bC \bT  - \bS \mathbf{E}\bT \|_F  \quad \text{and} \quad 		\min_{\bX \atop \mathsf{rk}(\bX)\leq k} \| \bY  \bX \bC   - \mathbf{E} \|_F 
\end{align}
with candidate solutions $\widetilde{\bX}:= \argmin_{\bX \atop \mathsf{rk}(\bX)\leq k} \|\bS \bY   \bX \bC \bT  - \bS \mathbf{E}\bT \|_F$ and $\widehat \bX := \argmin_{\bX \atop \mathsf{rk}(\bX)\leq k} \| \bY  \bX \bC   - \mathbf{E} \|_F$, respectively.  
	Set $p=k/\alpha$, $\mathbf{D}=\bY   \bX \bC $ and $\mathbf{E}=\bA + \bS^\dagger \bN_3 \bT^\dagger$ in the statement of~\lemref{S}. 
	Then we have with probability $1-2\delta$ over $\bS. \bT \sim \cD_A$, 
	\begin{align}  
		 \min_{\bX \atop \mathsf{r}(\bX)=k} \| \bY \bX \bC    - \mathbf{E}  \|_F^2 
		 & = 
		 \| \bY \widehat{\bX}  \bC    - \mathbf{E}   \|_F^2 
		  = \| \bY\widehat{\bX}  \bC  - (\bA + \bS^\dagger \bN_3 \bT^\dagger) \|_F^2 \nonumber \\
		 &  \geq (1 + \alpha) \| \bS(\bY \widehat{\bX}  \bC   - (\bA + \bS^\dagger \bN_3 \bT^\dagger)) \bT  \|_F^2 \nonumber \\
			& \geq (1 + \alpha)  \min_{\bX \atop \mathsf{r}(\bX)} \|\bS \bY {\bX}  \bC \bT  - \bS (\bA + \bS^\dagger \bN_3\bT^\dagger) ) \bT \|_F \nonumber \\
			& = (1 + \alpha)   \| \widehat \bY \widetilde{\bX} \widetilde \bY  - \bZ \|_F. 
			 \label{eq:SAT}
\intertext{The second and last equality follows from the definition, the first inequality follows from~\lemref{S} and the second inequality follows from  the fact that minimum is smaller than any other choice of $\bX$, more specifically $\bX = \widehat \bX$.
Since 
 $\bU$ is in the span of $\bY \bU'$, where $\bU'$ is the left singular vectors of $\widetilde \bX$, using Boutsidis {\it et al.}~\cite{BWZ16}  we have 
$ \| (\I - \bU \bU^{\mathsf T})(\bA + \bN_1 \bPhi^\dagger ) \|_F \leq \| \widehat \bY \widetilde{\bX} \widetilde \bY  - \bZ \|_F $. 
Combining~\eqnref{SAT} and~\eqnref{YXC}, this implies that}
			 \| (\I - \bU \bU^{\mathsf T})\bA \|_F &\leq (1 + O(\alpha)) \| \bA - [\bA]_k \|_F + O \paren{ \| \bPsi^\dagger \bN_2  \bT^\dagger  \|_F +  \| \bN_1 \bPhi^\dagger \|_F +   \|  \bS^\dagger \bN_3 \bT^\dagger  \|_F}.  \nonumber 
	\end{align}

As in the proof of~\thmref{meta}, for the choice of $v$ and $t$, for every matrices $\bD$, $\| \bS \bD \|_F \leq (1+\alpha) \| \bD \|_F$, $\| \bD \bT \|_F \leq (1+\alpha) \| \bD \|_F$, $\| \bPsi \bD \|_F \leq (1+\alpha) \| \bD \|_F$, and $\| \bD \bPhi \|_F \leq (1+\alpha) \| \bD \|_F$. In other words, 
$ \| \bPsi^\dagger \bN_2  \bT^\dagger  \|_F \leq (1+\alpha) \| \bN_2 \|_F, 
\| \bN_1 \bPhi^\dagger \|_F  \leq \sqrt{1+\alpha} \| \bN_1  \|_F ,
 \|  \bS^\dagger \bN_3 \bT^\dagger  \|_F  \leq (1+\alpha) \| \bN_3 \|_F.
$
Since $\alpha \in (0,1)$, 
\[  \| (\I - \bU \bU^{\mathsf T})\bA \|_F \leq (1 + O(\alpha)) \| \bA - [\bA]_k \|_F + O \paren{ \|  \bN_2   \|_F +  \| \bN_1  \|_F +   \|  \bN_3  \|_F}.  \]
The result follows using~\lemref{N}, Markov's inequality, and values of $\rho_1$ and $\rho_2$. 
\end{proof}





 
\section{Empirical Evaluation of Our Algorithms} \label{app:empirical}
Any algorithm to compute the low-rank factorization 
In this section, we give the experimental evaluation of our algorithms and compare it with the best known results. 
We ran our algorithms on a 2.7 GHz Intel Core i5 processor with 16 GB 1867 MHz DDR3 RAM. Our algorithms  keep on sampling a random  matrix  randomly until we sample a matrix with number of columns  more than $200$.

\subsection{Empirical Evaluation of {\scshape Private-Optimal-Space-$\lrf$}}
We first start with the discussion on the empirical evaluation of {\scshape Private-Optimal-Space-$\lrf$} (see~\figref{spacelowprivate} for the detail description and~\appref{codeprivate} for the source code). Since the error incurred by {\scshape Private-Frobenius-$\lrf$} is strictly less than that by {\scshape Private-Optimal-Space-$\lrf$}, we only concern ourselves with {\scshape Private-Optimal-Space-$\lrf$}. 
For the private setting, we sampled matrices from the following  distributions:
\begin{enumerate}
	\item All the entries are sampled uniformly from the interval $[1,5000]$
	\item All the entries are integers sampled uniformly from the interval $[1,5000]$.
\end{enumerate}
 In our experimental set-up, we keep the value of $\alpha = 0.25$ and $k=10$ fixed to get a better understanding of how the approximation changes with the changing values of dimensions. 


We start by explaining what every columns in Table~\ref{tab:empirical} means. The first two columns are the dimension of the private matrix, the third column is the desired rank of the output matrix, and the fourth column is the value of multiplicative approximation. For the ease of comparison, we have set $k$ and $\alpha$ to be a constant parameter in this experiment and let the dimension to be the free parameters. 

Recall that the problem of the low-rank factorization is to output a singular value decomposition $\widetilde \bU, \widetilde \bSigma, \widetilde \bV$ such that $\bM_k = \widetilde \bU \widetilde \bSigma \widetilde \bV^\mathsf T$ is a rank-$k$ matrix and
\[  \| \bA - \bM_k \|_F \leq (1+\alpha) \| \bA - [\bA]_k \|_F + \gamma, \]
where $\gamma$ is the additive error. The fifth and the sixth columns enumerate the value of the expression resulting from running our algorithm {\scshape Private-Optimal-Space-$\lrf$} and that by Hardt and Roth~\cite{HR12}, respectively. The last column represents the optimal low-rank factorization, $\| \bA - [\bA]_k \|_F$. 

There is no way to compute the actual additive error, $\gamma$, empirically. This is because there is a factor of multiplicative error and it is tough to argue what part of error is due to the multiplicative factor alone. In other words, the best we can compute is $\alpha \Delta_k + \gamma$ or the total approximation error incurred by our algorithm. From the practitioner point of view, the total error is a much useful parameter than the $\alpha \Delta_k + \gamma$. Therefore, in this section, we use the total approximation error ($\| \bA - \bM_k \|_F$) as the measure for our evaluation.

\begin{table}[h]
\small{
 \begin{center}
\begin{tabular}{|c|c|c|c|c|c|c|c|}
\hline
\multirow{2}{*}{Distribution of $\bA$} & \multirow{2}{*}{Rows} & \multirow{2}{*}{Columns} & \multirow{2}{*}{$k$} & \multirow{2}{*}{$\alpha$} & Our error & Hardt-Roth~\cite{HR12} & Optimal Error \\ 
& & & & & $\| \bA - \bM_k \|_F$ & $\| \bA - \bM_k \|_F$ & {$\| \bA - [\bA]_k \|_F$} \\ \hline
\multirow{15}{*}{Uniform real} 
& 535 & 50 & 10 & 0.25 & 223649.755822 & 552969.553361 & 190493.286508  \\ \cline{2-8}
& 581 & 57 & 10 & 0.25 & 254568.54093 & 491061.894752 & 213747.405532  \\ \cline{2-8}
& 671 & 65 & 10 & 0.25 & 295444.274372 & 470153.533646 & 250629.568178  \\ \cline{2-8}
& 705 & 70 & 10 & 0.25 & 317280.295345 & 546149.007321 & 269647.886009  \\ \cline{2-8}
& 709 & 68 & 10 & 0.25 & 309627.397618 & 664748.40864 & 265799.14431  \\ \cline{2-8}
& 764 & 74 & 10 & 0.25 & 344154.666385 & 529618.155224 & 291053.598305  \\ \cline{2-8}
& 777 & 50 & 10 & 0.25 & 270458.465497 & 436864.395454 & 235057.184632  \\ \cline{2-8}
& 861 & 57 & 10 & 0.25 & 311968.552859 & 494331.526734 & 269761.822539  \\ \cline{2-8}
& 1020 & 65 & 10 & 0.25 & 367175.274642 & 562322.74973 & 317998.616149  \\ \cline{2-8}
& 1054 & 70 & 10 & 0.25 & 389357.219211 & 490379.45171 & 338605.316751  \\ \cline{2-8}
& 1061 & 68 & 10 & 0.25 & 386772.176623 & 497648.401337 & 334581.574424  \\ \cline{2-8}
& 1137 & 74 & 10 & 0.25 & 413134.221292 & 528692.808214 & 364187.907915  \\ \cline{2-8}
& 1606 & 158 & 10 & 0.25 & 736233.063187 & 848827.366953 & 654600.528481  \\ \cline{2-8}
& 1733 & 169 & 10 & 0.25 & 786963.961154 & 932695.219591 & 706550.246496  \\ \cline{2-8}
\hline 
\hline 

\multirow{17}{*}{Uniform integers} 
& 522 & 50 & 10 & 0.25 & 217497.498819 & 496080.416815 & 185817.742179  \\ \cline{2-8}
& 555 & 51 & 10 & 0.25 & 229555.549295 & 463022.451669 & 195569.953293  \\ \cline{2-8}
& 605 & 60 & 10 & 0.25 & 267625.256679 & 525350.686285 & 225614.671569  \\ \cline{2-8}
& 714 & 70 & 10 & 0.25 & 316232.378407 & 477066.707503 & 270968.006565  \\ \cline{2-8}
& 804 & 51 & 10 & 0.25 & 284102.975661 & 548426.535153 & 241720.509615  \\ \cline{2-8}
& 899 & 86 & 10 & 0.25 & 402886.168791 & 554702.285328 & 346840.731082  \\ \cline{2-8}
& 906 & 60 & 10 & 0.25 & 328747.816311 & 455091.762984 & 284433.77154  \\ \cline{2-8}
& 913 & 90 & 10 & 0.25 & 412114.948358 & 634520.151202 & 358345.162361  \\ \cline{2-8}
& 1061 & 106 & 10 & 0.25 & 486139.117249 & 618819.626784 & 423775.149619  \\ \cline{2-8}
& 1063 & 70 & 10 & 0.25 & 395772.128472 & 485655.074685 & 339950.706212  \\ \cline{2-8}
& 1305 & 86 & 10 & 0.25 & 488729.886028 & 551863.893152 & 427234.256941  \\ \cline{2-8}
& 1383 & 90 & 10 & 0.25 & 513573.18853 & 595195.801858 & 451019.165808  \\ \cline{2-8}
& 1486 & 145 & 10 & 0.25 & 677118.945777 & 776008.62945 & 600584.597101  \\ \cline{2-8}
& 1481 & 146 & 10 & 0.25 & 670290.341074 & 733574.295922 & 600877.636254  \\ \cline{2-8}
& 1635 & 106 & 10 & 0.25 & 616323.217861 & 652624.510827 & 541305.826364  \\ \cline{2-8}
& 1848 & 180 & 10 & 0.25 & 836139.102987 & 884143.446663 & 755160.753156  \\ \cline{2-8}
& 1983 & 194 & 10 & 0.25 & 896926.450848 & 1005652.63777 & 814717.343468  \\ \cline{2-8}
\hline
\end{tabular}
 \caption{Empirical Comparison Between {\scshape Private-Optimal-Space-$\lrf$} and Hardt and Roth~\cite{HR12}.} \label{tab:empirical}
\end{center}
}
\end{table}

\begin{table}[h]
\small{
 \begin{center}
\begin{tabular}{|c|c|c|c|c|c|c|}
\hline
 {Rows} & {Columns} &{$k$} & {$\alpha$} & Our additive error  & Hardt-Roth Additive Error& Expected Additive Error \\ \hline
546 & 50 & 10 & 0.25 & 665.797531323 & 13971.0499468 & 818.4149601308452 \\ \hline
780 & 50 & 10 & 0.25 & 777.586619111 & 16772.4716145 & 915.2974642186384 \\ \hline
532 & 51 & 10 & 0.25 & 719.492368601 & 23512.4449181 & 817.9937262895351 \\ \hline
808 & 51 & 10 & 0.25 & 796.220653146 & 14613.575971 & 932.3276903711178 \\ \hline
655 & 62 & 10 & 0.25 & 845.550304391 & 34161.5584705 & 941.6210056415899 \\ \hline
951 & 62 & 10 & 0.25 & 903.21367849 & 14101.7225434 & 1055.244933017033 \\ \hline
891 & 89 & 10 & 0.25 & 1040.90463257 & 17863.8728746 & 1190.7199415503355 \\ \hline
1344 & 89 & 10 & 0.25 & 1273.05275389 & 18646.717977 & 1342.904603982863 \\ \hline
522 & 50 & 10 & 0.25 & 691.996294265 & 23193.0915951 & 806.8851715645378 \\ \hline
791 & 50 & 10 & 0.25 & 764.535817382 & 16245.8487264 & 919.3095213779045 \\ \hline
1449 & 140 & 10 & 0.25 & 1392.08606822 & 20835.1618122 & 1639.1144500265145 \\ \hline
2143 & 140 & 10 & 0.25 & 1518.90720786 & 13521.4077062 & 1827.1906485042728 \\ \hline
834 & 80 & 10 & 0.25 & 969.883327308 & 22051.2853027 & 1119.4759720856953 \\ \hline
1234 & 80 & 10 & 0.25 & 1005.64736332 & 13617.4286089 & 1257.2282266407094 \\ \hline
682 & 64 & 10 & 0.25 & 833.555604169 & 19875.8713339 & 965.2425226826088 \\ \hline
967 & 64 & 10 & 0.25 & 872.993347497 & 15774.1091244 & 1073.4637306190796 \\ \hline
924 & 90 & 10 & 0.25 & 1024.67412984 & 20018.1421648 & 1208.8365389072335 \\ \hline
1374 & 90 & 10 & 0.25 & 1168.33857697 & 15267.6596425 & 1357.3931639793143 \\ \hline
1981 & 194 & 10 & 0.25 & 2845.12484193 & 19457.1056713 & 2035.7863227397222 \\ \hline
2945 & 194 & 10 & 0.25 & 1938.83169063 & 14210.6828353 & 2263.796573983761 \\ \hline
1022 & 100 & 10 & 0.25 & 1130.29734323 & 14839.3883841 & 1298.0001587502882 \\ \hline
1530 & 100 & 10 & 0.25 & 1289.31236852 & 14886.7349415 & 1458.2654318931454 \\ \hline
1867 & 182 & 10 & 0.25 & 1806.04962492 & 13443.8218792 & 1952.5511080639412 \\ \hline
2757 & 182 & 10 & 0.25 & 1983.37270829 & 13509.2925192 & 2168.928742648721 \\ \hline
\end{tabular}
 \caption{Empirical Comparison of Additive Error of {\scshape Private-Optimal-Space-$\lrf$} and Hardt and Roth~\cite{HR12}.} \label{tab:additive_empirical}
\end{center}
}
\end{table}

The empirical evaluations, listed in Table~\ref{tab:empirical}, reflect that our algorithm perform consistently  better than that of Hardt and Roth~\cite{HR12} for all the dimension range. This agrees  with our analysis and our discussion in~\secref{compare}. In particular, we showed  that theoretically we perform better than Hardt and Roth~\cite{HR12} by a factor of $O(c \sqrt{k})$, where $c$ is the largest entry in their projection matrix. 

Another observation that one can make from our evaluation is that the error of our algorithm is quiet close to the actual error of approximation for all the dimension range. On the other hand, the error of Hardt and Roth~\cite{HR12} is close to optimal error in the large dimensional matrices. For small dimensional matrices, the error incurred by Hardt and Roth~\cite{HR12} is  lot more than the actual error. The error of Hardt and Roth~\cite{HR12} starts getting better as the dimension increases. We believe that the fact that our algorithm performs well over all the range of dimensions makes it more stable with respect to different datasets. In practice, this is highly desirable as one would like the algorithm to perform well on both large and small datasets. 

The last key observation one can gather from the empirical evaluation is that though the total error depends on the Frobenius norm, the additive error is  independent of the Frobenius norm of the original matrix. It is expected that the total error depends on the Frobenius norm because of the multiplicative factor, but if we see the difference between the errors (of both our and Hardt and Roth's algorithm) and the optimal error, the difference scales proportional to the dimensions of the matrices. 

\subsubsection{Empirical Evaluation of Additive Error  for various dimension}
As we mentioned earlier, if the matrix has rank greater than $k$, it is not possible to empirically evaluate the additive error. However, we believe it is still imperative to analyze the effect of differential privacy on the low-rank approximation of matrices. This in turn implies that one should also define experiments to empirically evaluate the additive error. One easy way to do this is to take as input a matrix with rank exactly $k$ and compare the error incurred with that of expected error promised by our theoretical results. In the next experiment we do the same and prune the last $n-k$ columns of the matrix and make it identically zero (see~\figref{spacelowprivate} for the detail description and~\appref{codeprivateadditive} for the source code). The result of our experiment is presented in Table~\ref{tab:additive_empirical}. Since every entries of the matrix is identically zero, we notice that the same trend as in Table~\ref{tab:empirical}:
\begin{enumerate}
	\item The additive error incurred by our algorithm is way less than the additive error incurred by Hardt and Roth~\cite{HR12} for all ranges of the dimension. We note that the matrices are highly incoherent as all the entries are sampled i.i.d. We believe the reason for this behavior is the fact that the theoretical result provided by Hardt and Roth~\cite{HR12} for incoherent matrices depended on the Frobenius norm of the input matrix. 
	\item Our algorithm consistently perform better than the additive error guaranteed by the theoretical results, but the difference becomes smaller as the dimension increases. This trend can be seen as due to the fact that our results are asymptotic and we believe as $m$ and $n$ are sufficiently large, our theoretical result would match the empirical results.  
\end{enumerate}

\subsubsection{Empirical Evaluation of Additive Error for various values of $\alpha$.}
An important parameter that comes in our bounds and is absent in the bounds of Hardt and Roth~\cite{HR12} is the factor of $\alpha$. This is because Hardt and Roth~\cite{HR12} consider a constant $\alpha$. Therefore, we feel it is important to analyze the additive error with respect to the change in $\alpha$ in order to better understand the effect of differential privacy on the low-rank approximation of matrices. Again, we take as input a matrix with rank exactly $k$ and compare the error incurred with that of expected error promised by our theoretical results while keeping the dimensions and the value of $k$ constant (see~\figref{spacelowprivate} for the detail description and~\appref{codeprivateadditive} for the source code).  The result of our experiment is presented in Table~\ref{tab:additive_empirical_alpha}. Since every entries of the matrix is identically zero, we notice that the same trend as in Table~\ref{tab:empirical}:
\begin{enumerate}
	\item The additive error incurred by our algorithm is way less than the additive error incurred by Hardt and Roth~\cite{HR12} for all ranges of the dimension. We note that the matrices are highly incoherent as all the entries are sampled i.i.d. We believe the reason for this behavior is the fact that the theoretical result provided by Hardt and Roth~\cite{HR12} for incoherent matrices depended on the Frobenius norm of the input matrix. 
	\item Our algorithm consistently perform better than the additive error guaranteed by the theoretical results, except for certain values of the dimensions and multiplicative error ($m=2800, n=200,\alpha = 0.24$). Even in these cases, the error is not that far from what is predicted from the theoretical analysis. 
\end{enumerate}

\begin{table}[h]
\small{
 \begin{center}
\begin{tabular}{|c|c|c|c|c|c|c|}
\hline
 {Rows} & {Columns} &{$k$} & {$\alpha$} & Expected Additive Error & Our additive error  & Hardt-Roth Additive Error \\ \hline
1800 & 200 & 10 & 0.1 & 7889.477972559828 & 8056.91337611 & 18155.7964938 \\ \hline
1800 & 200 & 10 & 0.12 & 7114.371042156239 & 6986.62461896 & 16748.9933963 \\ \hline
1800 & 200 & 10 & 0.14 & 6234.5851506132285 & 6171.60624904 & 26904.2257957 \\ \hline
1800 & 200 & 10 & 0.16 & 6234.5851506132285 & 5982.88510433 & 16983.3495414 \\ \hline
1800 & 200 & 10 & 0.18 & 5190.996544903864 & 4956.753081 & 15746.884528 \\ \hline
1800 & 200 & 10 & 0.2 & 5190.996544903864 & 5044.67253402 & 14124.190425 \\ \hline
1800 & 200 & 10 & 0.22 & 5190.996544903864 & 5012.95041479 & 18030.3508985 \\ \hline
1800 & 200 & 10 & 0.24 & 5190.996544903864 & 4951.97119364 & 18013.7573095 \\ \hline
2800 & 200 & 10 & 0.1 & 8690.661799154163 & 8550.34968943 & 12281.3954039 \\ \hline
2800 & 200 & 10 & 0.12 & 7846.020358487806 & 7607.18833498 & 14296.1174909 \\ \hline
2800 & 200 & 10 & 0.14 & 6887.309251977769 & 6494.49329799 & 11674.4990144 \\ \hline
2800 & 200 & 10 & 0.16 & 6887.309251977769 & 6603.16942717 & 13860.6899516 \\ \hline
2800 & 200 & 10 & 0.18 & 5750.100760072804 & 5417.53433303 & 13425.7590356 \\ \hline
2800 & 200 & 10 & 0.2 & 5750.100760072804 & 5612.34884207 & 12731.6645942 \\ \hline
2800 & 200 & 10 & 0.22 & 5750.100760072804 & 5524.92292528 & 10703.6065701 \\ \hline
2800 & 200 & 10 & 0.24 & 5750.100760072804 & 6450.77223767 & 12610.5718019 \\ \hline
\end{tabular}
 \caption{Empirical Comparison of Additive Error of {\scshape Private-Optimal-Space-$\lrf$} for various values of $\alpha$ with Hardt and Roth~\cite{HR12}.} \label{tab:additive_empirical_alpha}
\end{center}
}
\end{table}

\subsubsection{Empirical Evaluation of Additive Error for various values of $k$.}
The last parameter that comes in our bounds and in the bounds of Hardt and Roth~\cite{HR12} is the factor of $k$. Therefore, we feel it is important to analyze the additive error with respect to the change in $k$ in order to better understand the effect of differential privacy on the low-rank approximation of matrices. Again, we take as input a matrix with rank exactly $k$ and compare the error incurred with that of expected error promised by our theoretical results while keeping the dimensions and the value of $k$ constant (see~\figref{spacelowprivate} for the detail description and~\appref{codeprivateadditive} for the source code). The result of our experiment is presented in Table~\ref{tab:additive_empirical_k}. Since every entries of the matrix is identically zero, we notice that the same trend as in Table~\ref{tab:empirical}:
\begin{enumerate}
	\item The additive error incurred by our algorithm is way less than the additive error incurred by Hardt and Roth~\cite{HR12} for all ranges of the dimension. We note that the matrices are highly incoherent as all the entries are sampled i.i.d. We believe the reason for this behavior is the fact that the theoretical result provided by Hardt and Roth~\cite{HR12} for incoherent matrices depended on the Frobenius norm of the input matrix. 
	\item Our algorithm is almost the same as the additive error guaranteed by the theoretical results. 
\end{enumerate}

\begin{table}[h]
\small{
 \begin{center}
\begin{tabular}{|c|c|c|c|c|c|c|}
\hline
 {Rows} & {Columns} &{$k$} & {$\alpha$} & Expected Additive Error & Our additive error  & Hardt-Roth Additive Error \\ \hline
450 & 50 & 10 & 0.25 & 1955.5571620375354 & 1898.46860098 & 20016.9761114 \\ \hline
450 & 50 & 11 & 0.25 & 1967.6211924065349 & 2067.40361185 & 24063.2388836 \\ \hline
450 & 50 & 12 & 0.25 & 1979.148228191017 & 1915.44710901 & 23731.1221157 \\ \hline
450 & 50 & 13 & 0.25 & 1990.2041707005583 & 1975.07412336 & 28148.8596261 \\ \hline
450 & 50 & 14 & 0.25 & 2000.8424524636212 & 1937.32846282 & 44746.5335036 \\ \hline
450 & 50 & 15 & 0.25 & 2395.058604918562 & 2656.30764859 & 38466.5764635 \\ \hline
450 & 50 & 16 & 0.25 & 2404.9864394488554 & 2396.43074838 & 51068.6496986 \\ \hline
450 & 50 & 17 & 0.25 & 2414.6085814624234 & 2518.70267919 & 59951.048053 \\ \hline
450 & 50 & 18 & 0.25 & 2423.9516422006986 & 2412.15253004 & 49652.4140161 \\ \hline
450 & 50 & 19 & 0.25 & 2433.0385823622364 & 2421.29320351 & 60475.2111995 \\ \hline
700 & 50 & 10 & 0.25 & 2210.505167206932 & 2137.57732175 & 18322.4177099 \\ \hline
700 & 50 & 11 & 0.25 & 2225.805492035416 & 2244.60952111 & 15915.3100262 \\ \hline
700 & 50 & 12 & 0.25 & 2240.424768215734 & 2341.16599056 & 18904.9300369 \\ \hline
700 & 50 & 13 & 0.25 & 2254.4465757300545 & 2356.95623935 & 23088.1856217 \\ \hline
700 & 50 & 14 & 0.25 & 2267.9386809066186 & 2251.25863623 & 29899.7509467 \\ \hline
700 & 50 & 15 & 0.25 & 2707.304866721297 & 2948.78705396 & 33654.4497548 \\ \hline
700 & 50 & 16 & 0.25 & 2719.895940227483 & 2906.80323902 & 36885.4983615 \\ \hline
700 & 50 & 17 & 0.25 & 2732.0993162013783 & 2677.54649011 & 36275.4890195 \\ \hline
700 & 50 & 18 & 0.25 & 2743.9487446109097 & 2849.04550002 & 37581.6388795 \\ \hline
700 & 50 & 19 & 0.25 & 2755.473345587197 & 2658.52670606 & 47314.5277845 \\ \hline
\end{tabular}
 \caption{Empirical Comparison of Additive Error of {\scshape Private-Optimal-Space-$\lrf$}  and Hardt and Roth~\cite{HR12} for various values of $k$.} \label{tab:additive_empirical_k}
\end{center}
}
\end{table}

\subsection{Empirical Evaluation of {\scshape Private-Local-$\lrf$}}
In this section, we understand the result of our empirical evaluation of {\scshape Private-Local-$\lrf$}  (see~\figref{spacelowprivate} for the detail description and~\appref{codelocal} for the source code). Recall that in this case, we output a rank-$k$ orthonormal matrix $\bU$ such that $\bU \bU^{\mathsf T} \bA$ well approximates the matrix $\bA$ with high probabiliy, i.e., 
$$  { \| \bA - \bU \bU^{\mathsf T}\bA \|_F \leq (1+\alpha) \| \bA - [\bA]_k  \|_F +\gamma }$$ with probability at least $1-\beta$.

\begin{table}[h]
\small{
 \begin{center}
\begin{tabular}{|c|c|c|c|c|c|c|}
\hline
 {Rows} & {Columns} &{$k$} & {$\alpha$} & Our error ($\| \bA - \bU \bU^{\mathsf T} \bA \|_F$)  & {Optimal Error ($\| \bA - [\bA]_k \|_F$)} \\ \hline
 460 & 50 & 10 & 0.25 & 26730.7062683 & 18376.5128345  \\ \hline
 486 & 50 & 10 & 0.25 & 28080.8322915 & 18964.3784188  \\ \hline
 516 & 55 & 10 & 0.25 & 30526.641347 & 20870.0427927  \\ \hline
 553 & 56 & 10 & 0.25 & 29787.7727748 & 21862.0631087  \\ \hline
 568 & 59 & 10 & 0.25 & 31632.1354094 & 22760.6941581  \\ \hline
 616 & 64 & 10 & 0.25 & 33981.2524409 & 25159.7572035  \\ \hline
 709 & 50 & 10 & 0.25 & 36449.213243 & 23110.1192303  \\ \hline
 730 & 50 & 10 & 0.25 & 35232.7419048 & 23414.630713  \\ \hline
 742 & 50 & 10 & 0.25 & 36167.200779 & 23613.9643873  \\ \hline
 805 & 56 & 10 & 0.25 & 38175.3014108 & 26613.9920423  \\ \hline
 817 & 56 & 10 & 0.25 & 38604.715515 & 26721.6604408  \\ \hline
 846 & 59 & 10 & 0.25 & 39702.5550961 & 28179.965914  \\ \hline
 907 & 64 & 10 & 0.25 & 41889.0912424 & 30717.3035814  \\ \hline
 924 & 101 & 10 & 0.25 & 50376.7374653 & 40511.7157875  \\ \hline
 1195 & 130 & 10 & 0.25 & 66049.3867783 & 53122.2489582  \\ \hline
 1262 & 138 & 10 & 0.25 & 70596.6678822 & 56526.8303644  \\ \hline
 1433 & 101 & 10 & 0.25 & 63976.8587336 & 50810.8395315  \\ \hline
 1698 & 186 & 10 & 0.25 & 91159.7299285 & 77358.5067276  \\ \hline
 1857 & 130 & 10 & 0.25 & 85812.7647106 & 66588.7353796  \\ \hline
 1956 & 138 & 10 & 0.25 & 82720.7186764 & 70710.4258919  \\
\hline 
\end{tabular}
 \caption{Empirical Evaluation of {\scshape Private-Local-$\lrf$}.} \label{tab:local_empirical}
\end{center}
}
\end{table}

 For empirical evaluation of our locally-private algorithm, we sampled matrices such that every entries are uniform real number between $[0,500]$.  In our experimental set-up, we keep the value of $\alpha = 0.25$ and $k=10$ fixed to get a better understanding of how the approximation changes with the changing values of dimensions. 
Our empirical results are listed in Table~\ref{tab:local_empirical}. From the table, we can immediately see the effect of the role of the dimension $m$ and can see the difference between the optimal approximation and our approximation scale faster than in the case of {\scshape Private-Optimal-Space-$\lrf$}. However, even in the case of very large matrices, our approximation error is very small compared to the actual approximation error. In other words, this evaluation gives us a hint that our algorithm does not pay a lot for local model of computation.  We make this more explicit in our next set of experiments given in the following sections when we empirically evaluate the additive error to better understand the effect of local differential privacy on the accuracy of the algorithm.

\begin{table}[h]
\small{
 \begin{center}
\begin{tabular}{|c|c|c|c|c|c|c|}
\hline
 {Rows} & {Columns} &{$k$} & {$\alpha$} & Our Additive Error  & Expected Additive Error \\ \hline
454 & 50 & 10 & 0.25 & 1591.70826988 & 2236.0029137367846  \\ \hline
494 & 53 & 10 & 0.25 & 1659.12206578 & 2348.466431926056  \\ \hline
511 & 52 & 10 & 0.25 & 1737.01336788 & 2395.0391957271527  \\ \hline
562 & 60 & 10 & 0.25 & 2000.5063552 & 2530.799729601783  \\ \hline
622 & 66 & 10 & 0.25 & 1580.25894963 & 2683.7137939487484  \\ \hline
643 & 70 & 10 & 0.25 & 1962.46122303 & 2735.6746706859394  \\ \hline
645 & 70 & 10 & 0.25 & 2172.80349046 & 2740.583809665359  \\ \hline
702 & 50 & 10 & 0.25 & 1989.93430285 & 2877.768275821911  \\ \hline
728 & 52 & 10 & 0.25 & 2059.26581077 & 2938.6953727048567  \\ \hline
743 & 53 & 10 & 0.25 & 1993.78092528 & 2973.406277353555  \\ \hline
844 & 60 & 10 & 0.25 & 2328.32881246 & 3199.4689043324324  \\ \hline
853 & 94 & 10 & 0.25 & 2557.29517064 & 3219.0131110666084  \\ \hline
925 & 66 & 10 & 0.25 & 2484.75888129 & 3372.181278187486  \\ \hline
981 & 70 & 10 & 0.25 & 2198.25471982 & 3487.6698882062387  \\ \hline
983 & 70 & 10 & 0.25 & 2597.33562269 & 3491.7393627228817  \\ \hline
1055 & 113 & 10 & 0.25 & 3217.07445615 & 3635.863197534171  \\ \hline
1190 & 131 & 10 & 0.25 & 4220.63858322 & 3894.742583462755  \\ \hline
1320 & 94 & 10 & 0.25 & 2808.85540377 & 4131.88571345083  \\ \hline
1583 & 113 & 10 & 0.25 & 3870.14181092 & 4581.675389464531  \\ \hline
1713 & 186 & 10 & 0.25 & 4445.32063486 & 4791.555177594492  \\ \hline
1835 & 131 & 10 & 0.25 & 5447.86868861 & 4982.106342591863  \\ \hline
1900 & 206 & 10 & 0.25 & 4732.78855445 & 5081.305279633226  \\ \hline
2213 & 158 & 10 & 0.25 & 4145.11357428 & 5539.0037396773205  \\ \hline
2888 & 206 & 10 & 0.25 & 6234.47976474 & 6436.032325032722  \\ \hline
\hline 
\end{tabular}
 \caption{Empirical Evaluation of Additive Error of {\scshape Private-Local-$\lrf$} for varying values of dimensions.} \label{tab:local_additive}
\end{center}
}
\end{table}

\subsubsection{Empirical Evaluation of Additive Error for Various Dimension}
As we argued earlier,  it is still imperative to analyze the effect of differential privacy on the principal component analysis of matrices. This in turn implies that one should also define experiments to empirically evaluate the additive error. As earlier, we take as input a matrix with rank exactly $k$ and compare the error incurred with that of expected error promised by our theoretical results. In the next experiment we do the same and prune the last $n-k$ columns of the matrix and make it identically zero (see~\figref{spacelowprivate} for the detail description and~\appref{codelocaladditive} for the source code). The result of our experiment is presented in Table~\ref{tab:local_additive}. 
Recall that the additive error incurred by our algorithm is 
$$ \gamma := {O}\paren{ \eta \alpha^{-2} \log(k/\delta)  \sqrt{m \log (1/\delta)}/\eps }$$ for $\eta = \max \set{k,1/\alpha}.$

Since every entries of the matrix $ \bA - [\bA]_k $ is identically zero, the error listed in the Table~\ref{tab:local_additive} is the additive error. We note that the trend of Table~\ref{tab:local_additive} shows the same trend as in Table~\ref{tab:local_empirical}.

 Our algorithm consistently perform better than the additive error guaranteed by the theoretical results (except for $m=1835$ and $n= 131$), but the difference becomes smaller as the dimension increases. This trend can be seen as due to the fact that our results are asymptotic and we believe when $m$ is sufficiently large, our theoretical result would match the empirical results.

\subsubsection{Empirical Evaluation of Additive Error for various values of $\alpha$.}
An important parameter that comes in our bounds  is the factor of $\alpha$. Our theoretical result shows a tradeoff  between the additive error and the multiplicative approximation factor. This makes it important to analyze the additive error with respect to the change in $\alpha$ in order to better understand the effect of differential privacy on the low-rank approximation of matrices. Again, we take as input a matrix with rank exactly $k$ and compare the error incurred with that of expected error promised by our theoretical results while keeping the dimensions and the value of $k$ constant. The result of our experiment is presented in Table~\ref{tab:local_additive_alpha}. 

Recall that the additive error incurred by our algorithm is 
$$ \gamma := {O}\paren{ \eta \alpha^{-2} \log(k/\delta)  \sqrt{m \log (1/\delta)}/\eps }$$ for $\eta = \max \set{k,1/\alpha}.$

Since every entries of the matrix $ \bA - [\bA]_k $ is identically zero, the error listed is the additive error incurred by our algorithm. We notice that the same trend as in Table~\ref{tab:local_empirical}.

	We ran our algorithm for $m=\set{450,700}$ and $n=50$ with $k=10$ and $\alpha$ ranging from $0.10$ to $0.24$ in the step-size of $0.02$. Our algorithm consistently perform better than the additive error guaranteed by the theoretical results.  The empirical error is way better than the error predicted by the theoretical result for small values of $\alpha$ and it starts getting closer as the values of $\alpha$ increases. 

\begin{table}[h]
\small{
 \begin{center}
\begin{tabular}{|c|c|c|c|c|c|c|}
\hline
 {Rows} & {Columns} &{$k$} & {$\alpha$} & Our Additive Error  & Expected Additive Error \\ \hline
450 & 50 & 10 & 0.10 & 750.710811389 & 14088.628533821675  \\ \hline
450 & 50 & 10 & 0.12 & 748.91024846 & 9639.587944193776  \\ \hline
450 & 50 & 10 & 0.14 & 764.598714683 & 7415.067649379828  \\ \hline
450 & 50 & 10 & 0.16 & 727.388618355 & 5190.54735456588  \\ \hline
450 & 50 & 10 & 0.18 & 748.286674358 & 4449.040589627896  \\ \hline
450 & 50 & 10 & 0.20 & 722.269126382 & 3707.533824689914  \\ \hline
450 & 50 & 10 & 0.22 & 695.354504806 & 2966.027059751931  \\ \hline
450 & 50 & 10 & 0.24 & 687.993333996 & 2224.520294813948  \\ \hline
700 & 50 & 10 & 0.1 & 932.294964215 & 18195.922623848826  \\ \hline
700 & 50 & 10 & 0.12 & 937.676945315 & 12449.841795264987  \\ \hline
700 & 50 & 10 & 0.14 & 939.872486613 & 9576.801380973067  \\ \hline
700 & 50 & 10 & 0.16 & 942.039074746 & 6703.760966681148  \\ \hline
700 & 50 & 10 & 0.18 & 909.247538784 & 5746.08082858384  \\ \hline
700 & 50 & 10 & 0.20 & 948.147219877 & 4788.400690486534  \\ \hline
700 & 50 & 10 & 0.22 & 860.725294579 & 3830.7205523892267  \\ \hline
700 & 50 & 10 & 0.24 & 882.673634462 & 2873.04041429192  \\ \hline
\end{tabular}
 \caption{Empirical Evaluation of Additive Error of {\scshape Private-Local-$\lrf$} for various values of $\alpha$.} \label{tab:local_additive_alpha}
\end{center}
}
\end{table}

\begin{table}[h]
\small{
 \begin{center}
\begin{tabular}{|c|c|c|c|c|c|c|}
\hline
 {Rows} & {Columns} &{$k$} & {$\alpha$} & Our Additive Error  & Expected Additive Error \\ \hline
450 & 50 & 10 & 0.2 & 739.297184071 & 2966.027059751931  \\ \hline
450 & 50 & 11 & 0.2 & 792.139775851 & 3707.533824689914  \\ \hline
450 & 50 & 12 & 0.2 & 827.193502967 & 3707.533824689914  \\ \hline
450 & 50 & 13 & 0.2 & 857.238846842 & 4449.040589627896  \\ \hline
450 & 50 & 14 & 0.2 & 881.890196768 & 4449.040589627896  \\ \hline
450 & 50 & 15 & 0.2 & 877.407208023 & 5190.54735456588  \\ \hline
450 & 50 & 16 & 0.2 & 955.688935848 & 5190.54735456588  \\ \hline
450 & 50 & 17 & 0.2 & 942.773082147 & 5932.054119503862  \\ \hline
450 & 50 & 18 & 0.2 & 1019.96432587 & 5932.054119503862  \\ \hline
450 & 50 & 19 & 0.2 & 1033.82124639 & 6673.560884441846  \\ \hline
700 & 50 & 10 & 0.2 & 921.415775041 & 3830.7205523892267  \\ \hline
700 & 50 & 11 & 0.2 & 903.693226831 & 4788.400690486534  \\ \hline
700 & 50 & 12 & 0.2 & 961.550364155 & 4788.400690486534  \\ \hline
700 & 50 & 13 & 0.2 & 1055.58902486 & 5746.08082858384  \\ \hline
700 & 50 & 14 & 0.2 & 1102.54543656 & 5746.08082858384  \\ \hline
700 & 50 & 15 & 0.2 & 1139.854348 & 6703.760966681148  \\ \hline
700 & 50 & 16 & 0.2 & 1188.983938 & 6703.760966681148  \\ \hline
700 & 50 & 17 & 0.2 & 1216.1836631 & 7661.441104778453  \\ \hline
700 & 50 & 18 & 0.2 & 1207.76296999 & 7661.441104778453  \\ \hline
700 & 50 & 19 & 0.2 & 1303.58983727 & 8619.121242875759  \\ \hline
\end{tabular}
 \caption{Empirical Evaluation of Additive Error of {\scshape Private-Local-$\lrf$} for various values of $k$.} \label{tab:local_additive_k}
\end{center}
}
\end{table}

\subsubsection{Empirical Evaluation of Additive Error for various values of $k$.}
The last parameter that comes in our bounds for local differentially private algorithm is the factor of $k$. Therefore, we feel it is important to analyze the additive error with respect to the change in $k$ in order to better understand the effect of differential privacy on the low-rank approximation of matrices. Again, we take as input a matrix with rank exactly $k$ for varying values of $k$ and compare the error incurred with that of expected error promised by our theoretical results while keeping the dimensions and the value of $k$ constant. The result of our experiment is presented in Table~\ref{tab:local_additive_k}. 

Recall that the additive error incurred by our algorithm is 
$$ \gamma := {O}\paren{ \eta \alpha^{-2} \log(k/\delta)  \sqrt{m \log (1/\delta)}/\eps }$$ for $\eta = \max \set{k,1/\alpha}.$

Since every entries of the matrix $ \bA - [\bA]_k $ is identically zero, the error of our algorithm is due to the additive error. Again as predicted by our results, we notice that the same trend as in Table~\ref{tab:local_empirical}:

	We run our algorithm for $m=\set{450,700}$ and $n=50$ with $\alpha = 0.25$ and $k$ ranging from $10$ to $19$. Our algorithm consistently perform better than the additive error guaranteed by the theoretical results.  The empirical error is way better than the error predicted by the theoretical result for small values of $\alpha$ and it starts getting closer as the values of $\alpha$ increases. 

\subsection{Empirical Evaluation of {\scshape Optimal-Space-$\lrf$}}
\begin{table}[t]
\small{
 \begin{center}
\begin{tabular}{|c|c|c|c|c|c|}
\hline
Rows & Columns & $k$ & $\alpha$ & Total error ($\| \bA - \bM_k \|_F$) & Optimal Error ($\| \bA - [\bA]_k \|_F$) \\ \hline
498 & 52 & 10 & 0.25 & 203636.487171 & 197577.81058 \\ \hline
565 & 62 & 10 & 0.25 & 241322.216245 & 234585.873351 \\ \hline
600 & 66 & 10 & 0.25 & 258293.615653 & 251581.697045 \\ \hline
634 & 68 & 10 & 0.25 & 272545.864945 & 265331.39875 \\ \hline
701 & 50 & 10 & 0.25 & 236960.63674 & 230358.636959 \\ \hline
719 & 76 & 10 & 0.25 & 311278.453125 & 302812.926139 \\ \hline
736 & 50 & 10 & 0.25 & 243505.180195 & 236083.969219 \\ \hline
775 & 52 & 10 & 0.25 & 255720.933052 & 248391.690752 \\ \hline
780 & 86 & 10 & 0.25 & 348405.00368 & 340148.071775 \\ \hline
818 & 90 & 10 & 0.25 & 364630.812101 & 355410.694701 \\ \hline
888 & 62 & 10 & 0.25 & 306043.651247 & 297740.00906 \\ \hline
954 & 66 & 10 & 0.25 & 330351.384387 & 321254.463405 \\ \hline
966 & 68 & 10 & 0.25 & 337962.399198 & 329606.663985 \\ \hline
1102 & 76 & 10 & 0.25 & 385372.431742 & 376530.056249 \\ \hline
1149 & 127 & 10 & 0.25 & 529627.844251 & 516435.343938 \\ \hline
1184 & 130 & 10 & 0.25 & 545340.762064 & 531377.485355 \\ \hline
1206 & 86 & 10 & 0.25 & 434705.450777 & 424823.785653 \\ \hline
1288 & 90 & 10 & 0.25 & 461257.277745 & 451382.89648 \\ \hline
1549 & 169 & 10 & 0.25 & 729124.244066 & 701894.583657 \\ \hline
1612 & 113 & 10 & 0.25 & 587309.340404 & 575305.302626 \\ \hline
1802 & 127 & 10 & 0.25 & 666645.03627 & 650960.768662 \\ \hline 
1866 & 130 & 10 & 0.25 & 686638.036543 & 669968.400955 \\ \hline
2367 & 169 & 10 & 0.25 & 904435.683545 & 870912.869511 \\ \hline
\end{tabular}
 \caption{Empirical Evaluation of  {\scshape Optimal-Space-$\lrf$}.} \label{tab:nonprivate_evaluation}
\end{center}
}
\end{table}
In this section, we understand the result of our empirical evaluation of {\scshape Private-Local-$\lrf$}  (see~\figref{spacelow} for the detail description and~\appref{codenonprivate} for the source code).
Recall that in the non-private setting, we want to output a low-rank factorization such that its product $\bM_k$ satisfies the following inequality with high probability:
$$ \| \mathbf{M}_k - \bA \|_F  \leq (1+\alpha) \|\bA - [\bA]_k\|_F. $$

 For the empirical evaluation of our non-private algorithm, we sampled matrices such that every entries are uniform real number between $[0,5000]$. In our experimental set-up, we keep the value of $\alpha = 0.25$ and $k=10$ fixed to get a better understanding of how the approximation changes with the changing values of dimensions. 
Our empirical results are listed in Table~\ref{tab:nonprivate_evaluation}. We see that for all the ranges of dimensions we evaluated on, the value in the column marked as $\| \bA - \bM_k \|_F$ is well within a $(1 + \alpha )$ factor of the corresponding entries marked under the column $\| \bA - [\bA]_k \|_F$. In fact, empirical evidence suggests that our algorithms gives a much better approximation than $\alpha = 0.25$ (in fact, it is closer to $\alpha = 0.05$).  This gives a clear indication that our algorithm performs as the theoretical bound suggests for a large range of dimensions.



\renewcommand{\bL}{\mathbf{L}}
\section{Analysis of Boutsidis {\it et al.} Under Noisy Storage} \label{app:BWZ16}
The algorithm of Boutsidis {\it et al.}~\cite{BWZ16} maintains fives sketches, $\bM = \bT_l \bA \bT_r$, $\bL=\bS \bA \bT_r$, $\bN=\bT_l \bA \bR$, $\bD=\bA \bR$ and $\bC = \bS \bA$, where $\bS$ and $\bR$ are the embedding for the generalized linear regression and $\bT_l$ and $\bT_r$ are affine embedding matrices. It then computes 
\[ \bX_* = \argmin_{\mathsf{r}(\bX)\leq k } \| \bN \bX \bL - \bM \|_F \] and its SVD as $\bX_* = \bU_* \bSigma_* \bV_*^{\mathsf T}$. It then outputs $\bD \bU_*$, $\bSigma_* $, and $\bV_*^{\mathsf T} \bC$. Note that it does not compute a singular value decomposition of the low-rank approximation, but a different form of factorization. 

Boutsidis {\it et al,}~\cite{BWZ16} prove the following three lemmas, combining which they get their result.
\begin{lemma} \label{lem:BWZ1}
	For all matrices $\bX \in \R^{t \times t}$, with probability at least $98/100$, we have
	\begin{align*}
		(1-\alpha)^2 \| \bA \bR \bX \bS \bA - \bA \|_F^2 & \leq \| \bT_l (\bA \bR \bX \bS \bA - \bA) \bT_r \|_F^2 \leq (1+\alpha)^2 \| \bA \bR \bX \bS \bA - \bA \|_F^2.
	\end{align*} 
\end{lemma}

\begin{lemma} \label{lem:BWZ2}
	Let $\widetilde{\bX}= \argmin_{\mathsf{r}(\bX)\leq k } \| \bA \bR \bX \bS \bA  - \bA \|_F^2$. Then with probability at least $98/100$, we have 
 \begin{align*}
 	\min_{\mathsf{r}(\bX)\leq k } \| \bT_l( \bA \bR \bX \bS \bA  - \bA) \bT_r \|_F^2 & \leq  \| \bA \bR \bX_* \bS \bA  - \bA \|_F^2 \leq  \| \bA \bR \widetilde{\bX} \bS \bA  - \bA \|_F^2.
 \end{align*}	
\end{lemma}

\begin{lemma} \label{lem:BWZ3}
	Let $\widetilde{\bX}= \argmin_{\mathsf{r}(\bX)\leq k } \| \bA \bR \bX \bS \bA  - \bA \|_F^2$. Then with probability at least $98/100$, we have 
 \begin{align*}
 	\|  \bA \bR \widetilde{\bX} \bS \bA  - \bA \|_F^2 & \leq  (1+\alpha) \Delta_k(\bA)^2.
 \end{align*}	
\end{lemma}

To compute the value of $\widetilde{\bX}$, they use a result in generalized rank-constrained matrix approximations, which says that 
\begin{align}
\bN^\dagger \sparen{ \bU_\bN \bU_\bN^{\mathsf T}  \bM \bV_\bL \bV_\bL^{\mathsf T}  }_k \bL^\dagger=  \argmin_{\mathsf{r}(\bX)\leq k } \| \bN \bX \bL - \bM \|_F,   \label{eq:generalizedconstrained}
\end{align}
where $\bU_\bN$ is the matrix whose columns are the left singular vectors of $\bN$ and $\bV_\bL$ is the matrix whose columns are the right singular vectors of $\bL$.

Now, in order to make the above algorithm differentially private, we need to use the same trick as we used in the algorithm for {\scshape Private-Optimal-Space-$\lrf$}; otherwise, we would suffer from the problems mentioned earlier in~\secref{techniques}. More precisely, we compute $\widehat{\bA} = \begin{pmatrix} \bA & \sigma_{\mathsf min} \I \end{pmatrix}$ and then store the following sketches:
$\bM = \bT_l \widehat{\bA} \bT_r + \bN_1$, $\bL=\bS \widehat{\bA} \bT_r + \bN_2$, $\bN=\bT_l \widehat{\bA} \bR + \bN_3$, $\bD=\widehat{\bA} \bR$ and $\bC = \bS \widehat{\bA} + \bN_4$, where $\bR = t^{-1} \bOmega \bPhi$ with $\bOmega$ being a Gaussian matrix.


As in the proof of~\thmref{low_space_private}, there are three terms that contributes to the additive error. The first and the third term are the same as in the proof of~\thmref{low_space_private}. However, the second term differs. 
In order to compute the additive error incurred due to the second term, we note that $\bX$ in~\lemref{BWZ1} needs to be $$\bX:= \widehat{\bX} = \argmin_{\mathsf{r}(\bX)\leq k } \| \bT_l(\widehat{\bA} \bR \bX \bS \widehat{\bA}  - \widehat \bA)\bT_r \|_F^2.$$ 
The reason for this value of $\widehat{\bX}$ is the same as in the proof of~\claimref{thirdspace}.
Moreover, all the occurrence of $\bA$ is replaced by $\widehat{\bA}$. In other words, we get the following by combining~\lemref{BWZ1}, \lemref{BWZ2}, and \lemref{BWZ3}.
\begin{align*}
(1-\alpha)^2 \| \widehat{\bA} \bR \widehat{\bX} (\bS  \widehat{\bA} + \bN_4) -  \widehat{\bA} \|_F^2 \leq  (1+\alpha) \Delta_k(\widehat{\bA})^2.
\end{align*}

In other words, the additive error incurred by this expression is $(1-\alpha) \| \widehat{\bA} \bR \widehat{\bX} \bN_4 \|_F$. Using~\eqnref{generalizedconstrained}, we have 
$\gamma_2= (1-\alpha) \| \widehat{\bA} \bR   \bN^\dagger \sparen{ \bU_\bN \bU_\bN^{\mathsf T}  \bM \bV_\bL \bV_\bL^{\mathsf T}  }_k \bL^\dagger  \bN_4 \|_F$, where $\bN,\bM,$ and $\bL$ are as defined above. This term depends on the  singular values of $\widehat \bA$ and hence can be arbitrarily large.

\section{Source Codes for Our Experiments} \label{app:sourcecodes}
We run our algorithm of random matrices of size that is randomly sampled. In other words, all the private matrices used in our evaluation are dense matrices. As a result, there is no benefit of using the random projection matrices of Clarkson and Woodruff~\cite{CW13}. Therefore, to ease the overload on the code, we use random Gaussian matrices with appropriate variance as all our random projection matrices. 

\lstset{language=Python}          

\subsection{Source Code for {\scshape Optimal-Space-$\lrf$}} \label{app:codenonprivate}
\small
{
\begin{lstlisting}[frame=single]  % Start your code-block
import numpy as np
import math
from numpy import linalg as la

def randomGaussianMatrix(rows,columns,variance):
    G = np.zeros((rows, columns))
    for i in range(rows):
        for j in range(columns):
            G[i][j] = np.random.normal(0, variance)
    return G

def computeSingular(singular, Matrix):
    Sigma = np.zeros(np.shape(Matrix))
    i = 0
    while i < len(Matrix) and i < len(Matrix[0]):
        Sigma[i][i] = singular[i]
        i += 1
    return Sigma

def lowrank(k, A):
    a = np.zeros(np.shape(A))
    Left, singular, Right = np.linalg.svd(A,full_matrices=1, compute_uv=1)
    singularmatrix = computeSingular(singular,A)
    top_right = Right[:k,:]
    top_left = Left[:,:k]
    top_singular = singularmatrix[:k,:k]
    a +=  np.dot(np.dot(top_left, top_singular),top_right)
    return a

def Upadhyay(Ycolumn, Yrow, Z, S, T):
    U, sc, Vc = np.linalg.svd(Ycolumn)
    Ur, sr, V = np.linalg.svd(Yrow)
    t = len(Yrow)
    v = len(S)
    U = U[:, :t]
    V = V[:t, :]
    TopU = np.dot(S, U)
    TopV = np.dot(V, T)
    Us, Ss, Vs = np.linalg.svd(TopU)
    Ut, St, Vt = np.linalg.svd(TopV)
    inn = np.dot( Us.T, np.dot( Z, Vt.T) )
    innlow = lowrank(k,inn)
    SingularS = computeSingular( np.reciprocal(Ss), TopU )
    SingularT = computeSingular( np.reciprocal(St), TopV )
    outerT = np.dot(innlow ,np.dot(SingularT.T, Ut.T) )
    outerS = np.dot(Vs.T,np.dot( SingularS.T, outerT))
    B = np.dot(U, np.dot(outerS,V))
    return B

def Initialization(A,m,n,k,alpha):
    Ak = lowrank(k, A)
    t = int(k/alpha)
    v = int(k/(alpha ** 2))
    Phi = randomGaussianMatrix(n,t,1/t)
    Psi = randomGaussianMatrix(t,m,1/t)
    S = randomGaussianMatrix(v,m,1/v)
    T = randomGaussianMatrix(n,v,1/v)
    Ycolumn = np.dot(A, Phi)
    Yrow = np.dot(Psi, A)
    Z = np.dot(S, np.dot(A, T))
    B =Upadhyay(Ycolumn, Yrow, Z, S, T)
    frobupadhyay = la.norm(A - B, 'fro')
    actual = la.norm(A - Ak, 'fro')
    print (m, n, k,alpha, frobupadhyay,actual)

i=1
n=50
while n< 200:
    m=9*n+int(50*np.random.rand())
    alpha=0.25
    k=10
    Initialization(randomMatrix(m,n,1,5000), m,n,k,alpha)
    Initialization(5000*np.random.rand(m, n), m,n,k,alpha)
    m = 14 * n + int(50*np.random.rand())
    Initialization(randomMatrix(m,n,1,5000), m,n,k,alpha)
    Initialization(5000*np.random.rand(m, n),m,n,k,alpha)
    n=50 + (i * int(20 * np.random.rand()))
    i=i + 1
\end{lstlisting}
}

\subsection{Source Code for Comparing Private Algorithms} \label{app:codeprivate}
In this section, we present our source code to compare our  algorithm with that of Hardt and Roth~\cite{HR12}.  
\begin{lstlisting}[frame=single]  % Start your code-block
import numpy as np
import math
from numpy import linalg as la

def randomGaussianMatrix(rows,columns,variance):
    G = np.zeros((rows, columns))
    for i in range(rows):
        for j in range(columns):
            G[i][j] = np.random.normal(0, variance)
    return G

def randomMatrix(rows,columns,a,b):
    A = np.zeros((rows, columns))
    for i in range(rows):
        for j in range(columns):
            A[i][j] = np.random.randint(a, b)
    return A
    
def computeSingular(singular, Matrix):
    Sigma = np.zeros(np.shape(Matrix))
    i = 0
    while i < len(Matrix) and i < len(Matrix[0]):
        Sigma[i][i] = singular[i]
        i += 1
    return Sigma

def lowrank(k, A):
    a = np.zeros(np.shape(A))
    Left, singular, Right = np.linalg.svd(A,full_matrices=1, compute_uv=1)
    singularmatrix = computeSingular(singular,A)
    top_right = Right[:k,:]
    top_left = Left[:,:k]
    top_singular = singularmatrix[:k,:k]
    a +=  np.dot(np.dot(top_left, top_singular),top_right)
    return a

def Upadhyay(Ycolumn, Yrow, Z, S, T):
    U, sc, Vc = np.linalg.svd(Ycolumn)
    Ur, sr, V = np.linalg.svd(Yrow)
    t = len(Yrow)
    v = len(S)
    U = U[:, :t]
    V = V[:t, :]
    TopU = np.dot(S, U)
    TopV = np.dot(V, T)
    Us, Ss, Vs = np.linalg.svd(TopU)
    Ut, St, Vt = np.linalg.svd(TopV)
    inn = np.dot( Us.T, np.dot( Z, Vt.T) )
    innlow = lowrank(k,inn)
    SingularS = computeSingular( np.reciprocal(Ss), TopU )
    SingularT = computeSingular( np.reciprocal(St), TopV )
    outerT = np.dot(innlow ,np.dot(SingularT.T, Ut.T) )
    outerS = np.dot(Vs.T,np.dot( SingularS.T, outerT))
    B = np.dot(U, np.dot(outerS,V))
    return B

def Initialization(A,m,n,k,alpha):
    Ak = lowrank(k, A)
    actual = la.norm(A - Ak, 'fro')
    HR = HardtRoth(A,m,n,k)
    frobHR = la.norm(A - HR, 'fro')
    m = m+n
    epsilon = 1
    delta = 1/(m)
    t = int(k/alpha)
    v = int(k/(alpha ** 2))
    Phi = randomGaussianMatrix(n,t,1/t)
    Psi = randomGaussianMatrix(t,m,1/t)
    S = randomGaussianMatrix(v,m,1/v)
    T = randomGaussianMatrix(n,v,1/v)
    sigma = 4*math.log(1/delta)*math.sqrt(t*math.log(1/delta)) / epsilon
    Ahat = A
    scaledI = sigma*np.identity(n)
    for i in range(n):
        Ahat = np.vstack([Ahat, scaledI[i]])
    for i in range(n):
        A = np.vstack([A, scaledI[i] - scaledI[i]])
    rho1 = math.sqrt(-4*math.log(delta)/epsilon**2)
    rho2 = math.sqrt(-6*math.log(delta)/epsilon**2)
    Ycolumn = np.dot(Ahat, Phi) + randomGaussianMatrix(m,t,rho1)
    Yrow = np.dot(Psi, Ahat)
    Z = np.dot(S, np.dot(Ahat, T)) + randomGaussianMatrix(v,v,rho2)
    L = np.dot(Yrow, T)
    N = np.dot(S, Ycolumn)
    B =Upadhyay(Ycolumn, Yrow, Z, S, T)
    frobupadhyay = la.norm(Ahat - B, 'fro')
    print (m, n, k,alpha,frobupadhyay,frobHR,actual)

def HardtRoth(A,m,n,k):
    Omega = randomGaussianMatrix(n,2*k,1)
    Y = np.dot(A, Omega)
    epsilon = 1
    delta = 1/(m)
    sigma = - 32 * k * math.log(delta)/epsilon**2
    N = randomGaussianMatrix(m, 2*k , sigma)
    Y += N
    Q = np.zeros(Y.shape)
    for i in range(Y.shape[1]):
        avec = Y[:, i]
        q = avec
        for j in range(i):
            q = q - np.dot(avec, Q[:, j]) * Q[:, j]

        Q[:, i] = q / la.norm(q)
    a = Q[0][0]
    for i in range(Y.shape[0]):
        for j in range(Y.shape[1]):
            if a < Q[i][j]:
                a = Q[i][j]
    rho =  32*(a**2)*k*math.log(8*k/delta)*math.log(1/delta)/epsilon**2
    inner = np.dot(Q.T, A) + randomGaussianMatrix(2*k,n,rho)
    output = np.dot(Q, inner)
    return output

i=1
n=50
while n < 200:
    m = 9 * n + int(50*np.random.rand())
    alpha = 0.25
    k = 10
    Initialization(randomMatrix(m,n,1,5000), m,n,k,alpha)
    Initialization(5000*np.random.rand(m, n), m,n,k,alpha)
    m = 14 * n + int(50*np.random.rand())
    Initialization(randomMatrix(m,n,1,5000), m,n,k,alpha)
    Initialization(5000*np.random.rand(m, n),m,n,k,alpha)
    n = 50 + (i * int(20 * np.random.rand()))
    i = i + 1
\end{lstlisting}

\subsection{Source Code for Analyzing Additive Error for {\scshape Optimal-Space-Private-$\lrf$}} \label{app:codeprivateadditive}
\begin{lstlisting}[frame=single]
import numpy as np
import math
from numpy import linalg as la

def randomGaussianMatrix(rows, columns, variance):
    G = np.zeros((rows, columns))
    for i in range(rows):
        for j in range(columns):
            G[i][j] = np.random.normal(0, variance)
    return G

def randomMatrixUnit(rows, columns, a, b):
    return (b-a) ** np.random.rand(rows, columns)

def randomMatrix(rows, columns,k, a, b):
    A = np.zeros((rows, columns))
    for i in range(rows):
        a = 0
        for j in range(k):
            A[i][j] = np.random.randint(a, b)
    return A

def computeSingular(singular, Matrix):
    Sigma = np.zeros(np.shape(Matrix))
    i = 0
    while i < len(Matrix) and i < len(Matrix[0]):
        Sigma[i][i] = singular[i]
        i += 1
    return Sigma

def lowrank(k, A):
    a = np.zeros(np.shape(A))
    Left, singular, Right = np.linalg.svd(A, full_matrices=1, compute_uv=1)
    singularmatrix = computeSingular(singular, A)
    top_right = Right[:k, :]
    top_left = Left[:, :k]
    top_singular = singularmatrix[:k, :k]
    a += np.dot(np.dot(top_left, top_singular), top_right)
    return a

def Upadhyay(A, m, n, k, alpha, epsilon, delta):
    m = m + n
    t = int(0.05 * k / alpha)
    v = int(0.05 * k / (alpha ** 2))
    Phi = randomGaussianMatrix(n, t, 1 / t)
    Psi = randomGaussianMatrix(t, m, 1 / t)
    S = randomGaussianMatrix(v, m, 1 / v)
    T = randomGaussianMatrix(n, v, 1 / v)
    sigma = 4*math.log(1 / delta)*math.sqrt(t*math.log(1 / delta)) / epsilon
    Ahat = A
    scaledI = sigma * np.identity(n)
    for i in range(n):
        Ahat = np.vstack([Ahat, scaledI[i]])
    for i in range(n):
        A = np.vstack([A, scaledI[i] - scaledI[i]])
    rho1 = math.sqrt(-4 * math.log(delta) / epsilon ** 2)
    rho2 = math.sqrt(-6 * math.log(delta) / epsilon ** 2)
    Ycolumn = np.dot(Ahat, Phi) + randomGaussianMatrix(m, t, rho1)
    Yrow = np.dot(Psi, Ahat)
    Z = np.dot(S, np.dot(Ahat, T)) + randomGaussianMatrix(v, v, rho2)
    U, sc, Vc = np.linalg.svd(Ycolumn)
    Ur, sr, V = np.linalg.svd(Yrow)
    U = U[:, :t]
    V = V[:t, :]
    TopU = np.dot(S, U)
    TopV = np.dot(V, T)
    Us, Ss, Vs = np.linalg.svd(TopU)
    Ut, St, Vt = np.linalg.svd(TopV)
    inn = np.dot(Us.T, np.dot(Z, Vt.T))
    innlow = lowrank(k, inn)
    SingularS = computeSingular(np.reciprocal(Ss), TopU)
    SingularT = computeSingular(np.reciprocal(St), TopV)
    outerT = np.dot(innlow, np.dot(SingularT.T, Ut.T))
    outerS = np.dot(Vs.T, np.dot(SingularS.T, outerT))
    B = np.dot(U, np.dot(outerS, V))
    expected =sigma * math.sqrt(n) + math.sqrt(-k *m * math.log(delta))
    return (B, Ahat, expected)

def HardtRoth(A, m, n, k, epsilon, delta):
    Omega = randomGaussianMatrix(n, 2 * k, 1)
    Y = np.dot(A, Omega)
    epsilon = 1
    delta = 1 / (m)
    sigma = - 32 * k * math.log(delta) / epsilon ** 2
    N = randomGaussianMatrix(m, 2 * k, sigma)
    Y += N
    Q = np.zeros(Y.shape)
    for i in range(Y.shape[1]):
        avec = Y[:, i]
        q = avec
        for j in range(i):
            q = q - np.dot(avec, Q[:, j]) * Q[:, j]

        Q[:, i] = q / la.norm(q)
    a = Q[0][0]
    for i in range(Y.shape[0]):
        for j in range(Y.shape[1]):
            if a < Q[i][j]:
                a = Q[i][j]
    rho = 32*(a**2)*k*math.log(8*k/delta)*math.log(1/delta)/epsilon**2
    inner = np.dot(Q.T, A) + randomGaussianMatrix(2 * k, n, rho)
    output = np.dot(Q, inner)
    return output

def Initialization(m, n, k, alpha):
    epsilon = 1
    delta = 1 / (m ** 2)
    A = randomMatrix(m,n,k,1,20)
    Ak = lowrank(k, A)
    actual = la.norm(A - Ak, 'fro')

    HR = HardtRoth(A, m, n, k, epsilon, delta)
    frobHR = la.norm(A - HR, 'fro')

    (Upadhyay16, Ahat, expected) = Upadhyay(A, m,n, k, alpha, epsilon, delta)
    frobupadhyay = la.norm(Ahat - Upadhyay16, 'fro')
     print (expected, frobupadhyay, frobHR)

#Use this part of the code for varying dimension
i = 1
n = 50
while n < 200:
    m = 9 * n + int(50 * np.random.rand())
    alpha = 0.1
    k = 10
    Initialization(randomMatrix(m, n, 0, 5), m, n, k, alpha)
    m = 14 * n + int(5 * np.random.rand())
    Initialization(randomMatrix(m, n, 0, 5), m, n, k, alpha)
    n = 50 + (i * int(5 * np.random.rand()))
    i = i + 1

#Use this part of the code for varying alpha
n=200
alpha = 0.10
while alpha < 0.25:
    m = 9 * n 
    k = 10
    Initialization(m,n,k,alpha)
    m = 14 * n
    Initialization(m,n,k,alpha)
    alpha = alpha + 0.02

#Use this part of the code for varying k
n=50
k = 10
alpha = 0.25
while k<20:
    m = 9 * n
    Initialization(m,n,k,alpha)
    m = 14 * n
    Initialization(m,n,k,alpha)
    k +=1
\end{lstlisting}

\subsection{Source Code for {\scshape Private-Local-$\lrf$}} \label{app:codelocal}
\begin{lstlisting}[frame=single]  % Start your code-block
import numpy as np
import math
from numpy import linalg as la

def randomGaussianMatrix(rows,columns,variance):
    G = np.zeros((rows, columns))
    for i in range(rows):
        for j in range(columns):
            G[i][j] = np.random.normal(0, variance)
    return G

def randomMatrix(rows,columns,a,b):
    A = np.zeros((rows, columns))
    for i in range(rows):
        for j in range(columns):
            A[i][j] = np.random.randint(a, b)
    return A

def lowrank(k, A):
    a = np.zeros(np.shape(A))
    Left, singular, Right = np.linalg.svd(A,full_matrices=1, compute_uv=1)
    top_left = Left[:,:k]
    a +=  np.dot(np.dot(top_left, top_left.T),A)
    return a

def Local(A,Ycolumn, L, Z, S):
    N = np.dot(S,Ycolumn)
    Un, Sn, Vn = np.linalg.svd(N)
    Ul, Sl, Vl = np.linalg.svd(L)
    inner = np.dot(Un,np.dot(Un.T,np.dot(Z,np.dot(Vl.T,Vl))))
    innerlowrank = lowrank(k,inner)
    output = np.dot(la.pinv(N),np.dot(innerlowrank,la.pinv(L)))
    U, s, V = np.linalg.svd(output)
    Y = np.dot(Ycolumn, U)
    Q = np.zeros(Y.shape)
    for i in range(Y.shape[1]):
        avec = Y[:, i]
        q = avec
        for j in range(i):
            q = q - np.dot(avec, Q[:, j]) * Q[:, j]
        Q[:, i] = q / la.norm(q)
    B = np.dot(Q, np.dot(Q.T,A))
    return B

def Initialization(A,m,n,k,alpha):
    t = int(k / (alpha))
    v = int(k / (alpha ** 2))
    Ak = lowrank(k, A)
    actual = la.norm(A - Ak, 'fro')
    epsilon = 0.1
    delta = 1/(m**10)
    Phi = randomGaussianMatrix(n, t, 1 / t)
    Psi = randomGaussianMatrix(t, m, 1 / t)
    S = randomGaussianMatrix(v, m, 1 / v)
    T = randomGaussianMatrix(n, v, 1 / v)
    rho1 = math.sqrt(-4 * math.log(delta)/epsilon**2)
    rho2 = m * math.sqrt(-6 * math.log(delta)/epsilon**2)
    Ycolumn = np.dot(A, Phi) + randomGaussianMatrix(m,t,rho1)
    Yrow = np.dot(np.dot(Psi, A), T) + randomGaussianMatrix(t,v,rho2)
    Z = np.dot(S, np.dot(A, T)) + randomGaussianMatrix(v,v,rho2)
    B = Local(A, Ycolumn, Yrow, Z, S)
    frobupadhyay = la.norm(A - B, 'fro')
    print (m, n, k, alpha, frobupadhyay, actual)

i=1
n=50
while n < 200:
    m = 9 * n + int(50*np.random.rand())
    alpha = 0.25
    k = 10
    Initialization(randomMatrix(m,n,1,5000), m,n,k,alpha)
    Initialization(5000*np.random.rand(m, n), m,n,k,alpha)
    m = 14 * n + int(50*np.random.rand())
    Initialization(randomMatrix(m,n,1,5000), m,n,k,alpha)
    Initialization(5000*np.random.rand(m, n),m,n,k,alpha)
    n = 50 + (i * int(20 * np.random.rand()))
    i = i + 1
\end{lstlisting}

\subsection{Source Code for Analyzing Additive Error for {\scshape Private-Local-$\lrf$}} \label{app:codelocaladditive}
\begin{lstlisting}[frame=single]  % Start your code-block
import numpy as np
import math
from numpy import linalg as la

def randomGaussianMatrix(rows,columns,variance):
    G = np.zeros((rows, columns))
    for i in range(rows):
        for j in range(columns):
            G[i][j] = np.random.normal(0, variance)
    return G

def randomMatrix(rows,columns,k,a,b):
    A = np.zeros((rows, columns))
    for i in range(rows):
        for j in range(k):
            A[i][j] = np.random.randint(a, b)
    return A

def lowrank(k, A):
    a = np.zeros(np.shape(A))
    Left, singular, Right = np.linalg.svd(A,full_matrices=1, compute_uv=1)
    top_left = Left[:,:k]
    a +=  np.dot(np.dot(top_left, top_left.T),A)
    return a

def Local(A,Ycolumn, L, Z, S):
    N = np.dot(S,Ycolumn)
    Un, Sn, Vn = np.linalg.svd(N)
    Ul, Sl, Vl = np.linalg.svd(L)
    inner = np.dot(Un , np.dot(Un.T , np.dot(Z , np.dot(Vl.T,Vl ) ) ) )
    innerlowrank = lowrank(k,inner)
    output = np.dot( np.linalg.pinv(N),np.dot(innerlowrank,np.linalg.pinv(L)))
    U, s, V = np.linalg.svd(output)
    Y = np.dot(Ycolumn, U)
    Q = np.zeros(Y.shape)
    for i in range(Y.shape[1]):
        avec = Y[:, i]
        q = avec
        for j in range(i):
            q = q - np.dot(avec, Q[:, j]) * Q[:, j]

        Q[:, i] = q / la.norm(q)
    B = np.dot(Q, np.dot(Q.T,A))
    return B

def Initialization(m,n,k,alpha):
    t = int(0.05*k / (alpha))
    v = int(0.02*k / (alpha ** 2))
    A = randomMatrix(m,n,k,1,20)
    Ak = lowrank(k, A)
    actual = la.norm(A - Ak, 'fro')
    epsilon = 0.1
    delta = 1/(m**2)
    Phi = randomGaussianMatrix(n, t, 1 / t)
    Psi = randomGaussianMatrix(t, m, 1 / t)
    S = randomGaussianMatrix(v, m, 1 / v)
    T = randomGaussianMatrix(n, v, 1 / v)
    rho1 = math.sqrt(-4 * math.log(delta)/epsilon**2)
    rho2 = m * math.sqrt(-6 * math.log(delta)/epsilon**2)

    Ycolumn = np.dot(A, Phi) + randomGaussianMatrix(m,t,rho1)
    Yrow = np.dot(np.dot(Psi, A), T) + randomGaussianMatrix(t,v,rho2)
    Z = np.dot(S, np.dot(A, T)) + randomGaussianMatrix(v,v,rho2)
    B = Local(A, Ycolumn, Yrow, Z, S)
    frobupadhyay = la.norm(A - B, 'fro')
    expected = v * math.sqrt(-m * math.log(delta))/epsilon
    print (frobupadhyay,  expected)


#Use this part of the code for varying dimension
i = 1
n = 50
while n < 200:
    m = 9 * n + int(50 * np.random.rand())
    alpha = 0.1
    k = 10
    Initialization(randomMatrix(m, n, 0, 5), m, n, k, alpha)
    m = 14 * n + int(5 * np.random.rand())
    Initialization(randomMatrix(m, n, 0, 5), m, n, k, alpha)
    n = 50 + (i * int(5 * np.random.rand()))
    i = i + 1

#Use this part of the code for varying alpha
i=1
n=200
alpha = 0.10
while alpha < 0.25:
    m = 9 * n
    k = 10
    Initialization(m,n,k,alpha)
    m = 14 * n
    Initialization(m,n,k,alpha)
    alpha = alpha + 0.02

#Use this part of the code for varying k
n=50
k = 10
alpha = 0.25
while k<20:
    m = 9 * n 
    Initialization(m,n,k,alpha)
    m = 14 * n
    Initialization(m,n,k,alpha)
    k +=1
\end{lstlisting}

\end{appendix}

\addcontentsline{toc}{section}{References}

\bibliographystyle{alpha}
{ \bibliography{low}}

\pagebreak

\end{document}